\documentclass[letterpaper,
twocolumn,10pt,
accepted=2025-07-21,
longbibliography]{quantumarticle}

\pdfoutput=1

\usepackage[numbers,sort&compress]{natbib}
\usepackage[utf8]{inputenc}
\usepackage[english]{babel}
\usepackage[T1]{fontenc}
\usepackage{hyperref}

\usepackage[usenames,dvipsnames]{xcolor}
\usepackage{cancel}
\usepackage{amsfonts,amsmath,amssymb}
\usepackage[english]{babel}
\usepackage[normalem]{ulem}
\usepackage{hyperref}
\usepackage{amsthm}
\usepackage[linesnumbered,noend]{algorithm2e}
\usepackage{subcaption,caption}
\usepackage{epsfig}

\usepackage{tikz}
\usepackage{mathrsfs}
\usepackage{bbold} 

\usepackage{clipboard} 

\usepackage{pgfplots}
\usepackage{pgfplotstable} 
\usepgfplotslibrary{groupplots}
\pgfplotsset{compat = newest}

\usepackage{bm}
\usetikzlibrary{shapes,arrows}
\usetikzlibrary{matrix}
\usetikzlibrary{3d}
\usetikzlibrary{calc,patterns,angles,quotes,arrows,automata}
\graphicspath{ {./imgs/}}
\usepackage{tikz-3dplot}

\makeatother
\usepackage{thmtools}
\usepackage{soul}

\theoremstyle{theorem}
\newtheorem{lemma}{Lemma}
\newtheorem{proposition}{Proposition}

\newtheorem{theorem}{Theorem}
\newtheorem{corollary}{Corollary}
\newtheorem{definition}{Definition}

\theoremstyle{definition}

\newtheoremstyle{probstyle}
{}                
{}                
{\itshape}    
{}                
{\bfseries}       
{:}               
{ }               
{}                
\theoremstyle{probstyle}
\newtheorem{problem}{Problem}

\newtheoremstyle{nonum}
{}                
{}                
{\itshape}        
{}                
{\bfseries}       
{:}               
{ }               
{\thmname{#1}\thmnumber{ #2}\thmnote{ #3}}                
\theoremstyle{nonum}
\newtheorem*{lemma*}{Lemma}
\newtheorem*{proposition*}{Proposition}
\newtheorem*{theorem*}{Theorem}

\newtheoremstyle{rem}
{}                
{}                
{\normalfont}        
{}                
{\mdseries}       
{.}               
{ }               
{\thmname{#1}\thmnumber{ #2}\thmnote{ #3}}                
\theoremstyle{rem}
\newtheorem{remark}{Remark}
\newtheorem{example}{Example}

\RestyleAlgo{ruled}
\SetKwInOut{Input}{Input}
\SetKwInOut{Output}{Output}
\SetKwInOut{Parameters}{Parameters}
\SetKwIF{If}{ElseIf}{Else}{if}{:}{else if}{else}{endif}

\newcommand{\Bc}{\mathcal{B}}
\newcommand{\Cc}{\mathcal{C}}
\newcommand{\Dc}{\mathcal{D}}

\newcommand{\Fc}{\mathcal{F}}

\newcommand{\Oc}{\mathcal{O}}
\newcommand{\Hc}{\mathcal{H}}
\newcommand{\Jc}{\mathcal{J}}
\newcommand{\Rc}{\mathcal{R}}

\newcommand{\Sc}{\mathcal{S}}

\newcommand{\Zc}{\mathcal{Z}}

\newcommand{\Ic}{\mathcal{I}}

\newcommand{\Lc}{\mathcal{L}}
\newcommand{\Tc}{\mathcal{T}}

\newcommand{\As}{\mathscr{A}}
\newcommand{\Bs}{\mathscr{B}}

\newcommand{\Ds}{\mathscr{D}}

\newcommand{\Os}{\mathscr{O}}
\newcommand{\Ss}{\mathscr{S}}
\newcommand{\Ns}{\mathscr{N}}
\newcommand{\Vs}{\mathscr{V}}
\newcommand{\Ws}{\mathscr{W}}
\newcommand{\Xs}{\mathscr{X}}
\newcommand{\Ys}{\mathscr{Y}}

\newcommand{\Rs}{\mathscr{R}}
\newcommand{\Es}{\mathscr{E}}

\newcommand{\Gs}{\mathscr{G}}

\newcommand{\Hs}{\mathscr{H}}

\newcommand{\Rb}{\mathbb{R}}
\newcommand{\Nb}{\mathbb{N}}
\newcommand{\Cb}{\mathbb{C}}

\newcommand{\Bf}{\mathfrak{B}}

\newcommand{\Hf}{\mathfrak{H}}
\newcommand{\Df}{\mathfrak{D}}

\newcommand{\Sf}{\mathfrak{S}}

\newcommand{\Uf}{\mathfrak{U}}

\newcommand{\CE}{\mathbb{E}|}

\newcommand{\SE}{\mathbb{J}|}

\newcommand{\ro}{\mathcal{R}_0}
\newcommand{\eo}{\mathcal{J}_0}
\newcommand{\rs}{\mathcal{R}}
\newcommand{\es}{\mathcal{J}}

\newcommand{\one}{\mathbb{1}}
\newcommand{\zero}{\mathbb{0}}

\newcommand{\ad}{{\rm ad}}

\newcommand{\tr}{{\rm tr}}

\newcommand{\Span}{{\rm span}}

\newcommand{\alg}{{\rm alg}}

\newcommand{\zentrum}{\mathcal{Z}}

\newcommand{\ket}[1]{\left| #1 \right>}
\newcommand{\bra}[1]{\left< #1 \right|}

\newcommand{\expect}[1]{\left< #1 \right>}
\newcommand{\braket}[2]{\left< #1 \middle\vert #2 \right>}
\newcommand{\ketbra}[2]{\left\vert #1 \middle>\middle< #2 \right\vert}

\begin{document}

\title{Exact Model Reduction for\newline Continuous-Time Open Quantum Dynamics}

\author{Tommaso Grigoletto}
\orcid{0000-0002-2891-3528}
\affiliation{Department of Information Engineering, University of Padova, Italy} 

\author{Yukuan Tao}
\affiliation{Department of Physics and Astronomy, Dartmouth College, Hanover, New Hampshire 03755, USA}
\altaffiliation{Present address: Department of Mathematics, Nottingham, NG7 2RD, UK} 

\author{Francesco Ticozzi}
\orcid{0000-0002-1742-119X}
\affiliation{Department of Information Engineering, University of Padova, Italy} 
\affiliation{Department of Physics and Astronomy, Dartmouth College, Hanover, New Hampshire 03755, USA}

\author{Lorenza Viola}
\orcid{0000-0002-8728-9235}
\affiliation{Department of Physics and Astronomy, Dartmouth College, Hanover, New Hampshire 03755, USA}

\maketitle

\onecolumn
\begin{abstract}
We consider finite-dimensional many-body quantum systems described by time-independent Hamiltonians and Markovian master equations, and present a systematic method for constructing {\em smaller-dimensional}, reduced models that \emph{exactly} reproduce the time evolution of a set of initial conditions or observables of interest. 
Our approach exploits Krylov operator spaces and their extension to operator algebras, and may be used to obtain reduced {\em linear} models of minimal dimension, well-suited for simulation on classical computers, or reduced {\em quantum} models that preserve the structural constraints of physically admissible quantum dynamics, as required for simulation on quantum computers. Notably, we prove that the reduced quantum-dynamical generator is still in Lindblad form. By introducing a new type of {\em observable-dependent symmetries}, we show that our method provides a non-trivial generalization of techniques that leverage symmetries, unlocking new reduction opportunities. We quantitatively benchmark our method on paradigmatic open many-body systems of relevance to condensed-matter and quantum-information physics. In particular, we demonstrate how our reduced models can quantitatively describe decoherence dynamics in central-spin systems coupled to structured environments, magnetization transport in boundary-driven dissipative spin chains, and unwanted error dynamics on information encoded in a noiseless quantum code.
\end{abstract} \bigskip

\twocolumn
\section{Introduction}
\subsection{Context and motivation}

Simulating the dynamics of many-body quantum systems is a longstanding challenge of computational physics, with implications ranging from strongly-correlated quantum matter, quantum chemistry, and high-energy physics to quantum technologies. Ultimately, the factor responsible for the computational hardness of the problem is the exponential growth of the system's tensor-product state space with the number of degrees of freedom that quantum mechanics mandates \cite{nielsen_chuang_2010}. While this complexity is already present for closed many-body systems evolving unitarily under Hamiltonian dynamics, additional challenges emerge from considering general, {\em open} quantum dynamics \cite{breuer2002theory}: this may be motivated by the recognition that unwanted dissipation and decoherence caused by the coupling to an uncontrolled environment are inevitably present in real-word systems to some extent \cite{ZurekRMP}; or, conversely, by the realization that suitably engineered dissipation may lead, alone or together with coherent evolution, to quantum states and phenomena not accessible otherwise \cite{verstraete2009quantum,TV2009,fazio}. Simulation of many-body quantum dynamics is a prime target application envisioned for quantum computers and quantum simulators \cite{AspuruG,Fauseweh,DiMeglio}, notably, open-quantum system simulators \cite{ViolaLloyd,Blatt}. Despite the enormous potential that quantum simulation brings, and impressive experimental advances \cite{Monroe,Lukin1,Lukin2,Google}, the ability to classically simulate large-scale quantum-circuit dynamics remains nonetheless crucial both for benchmarking purposes, to verify that the system is performing as expected, and to set a classical computational bar that quantum computation must pass to demonstrate quantum advantage \cite{Mandra}. 

Motivated by the above applications, a host of strategies for reducing the demands of direct simulation of quantum dynamics and obtaining tractable ``effective'' models have been developed. As in the classical case, a main guiding principle is the fact that, even if access to a complete description could be granted, this would neither be desirable nor necessary, as only certain information is ultimately ``relevant'' on physical grounds. Discarding ``irrelevant'' information can be achieved by invoking different forms of coarse-graining, aggregating, or averaging over microscopic degrees of freedom, typically in combination with perturbative arguments -- paradigmatic example being renormalization-group approaches, mean-field and hydrodynamic descriptions, or cluster-expansion techniques \cite{Gorkov,Subir,BeiZ}. For open quantum systems specifically,  weak-coupling and timescale-separation assumptions are used in standard Born-Markov derivations of master equations \cite{alicki-lendi,breuer2002theory,Riva}, along with more sophisticated approaches based, for instance, on different forms of clustering or coarse graining in space \cite{Liu,Rossini} or time \cite{Lidar, Tureci}. Further to that, acknowledging that the dynamics of {realistic} many-body quantum systems (subject to specified locality, energy, and time constraints) may {\em a priori} be confined to an exponentially small fraction of the available state space \cite{poulin2011} in turn legitimates efficient parametrizations of ``relevant'' quantum states in terms of tensor networks \cite{Orus,BeiZ}. The use of tensor-network methods, in combination with chain-mapping techniques \cite{Plenio,Tamascelli} or process-tensor representations \cite{Pollock,Keeling}, lies at the core of current state-of-the-art numerical methods available for open quantum system simulation in the presence of highly structured, possibly strongly-coupled environments. As the conceptual and practical significance of open quantum dynamics continues to expand across quantum science, efforts for obtaining improved simulation methods are multiplying, with a special focus on the simplest yet paradigmatic case of dynamics described by {\em time-independent Lindblad master equations} \cite{alicki-lendi}. Recent representative contributions include more accurate and efficient classical simulation algorithms based on high-order adiabatic elimination \cite{azouit2017towards,Regent} or low-rank structure \cite{Appelo}, as well as quantum simulation algorithms that exploit unitary embeddings \cite{Lin} or classical randomness \cite{Marvian}.

From a system-theoretical standpoint, the different approaches mentioned above can all be seen as instances of approximate {\em model (-order) reduction} (MR), a general framework aimed at reducing the complexity of a mathematical model while retaining some of the essential features of the original one \cite{antoulas}. A familiar setting in which MR takes place, in fact in its strongest {\em exact} form, arises when the dynamics exhibits a manifest symmetry. If the dynamics are described by a Hamiltonian, a decomposition of the Hilbert space according to irreducible representations of the symmetry group naturally furnishes invariant subspaces \cite{Ballentine}; hence, if the initial state resides entirely in one such subspace, an exact lower-dimensional Hamiltonian model may be found by restricting to this invariant subspace via projection. This holds true both for ``conventional'' symmetries, for which the number of dynamically disconnected invariant, ``Krylov subspaces'' \cite{kry31,nandy2024quantumdynamicskrylovspace} grows at most polynomially with system size, and for non-manifest ``unconventional'' symmetries, as occurring in systems exhibiting so-called Hilbert space fragmentation \cite{PhysRevX.9.041017,fragmentation,NumFragmentation}, whereby exponentially many disconnected Krylov subspaces exist. In this context, an important step was taken by Kumar and Sarovar \cite{kumar2014model}, where conditions for certifying the existence of a non-trivial proper invariant subspace were identified and, building on that, exact MR methods were developed for quantum many-body {\em Hamiltonian} models. A subspace of the Hilbert space, however, is not the natural candidate for reduction when the initial condition is not a pure state or the dynamics is not unitary. As such, the approach does not lend itself to the application to genuine open quantum dynamics.

{\bf Key advances.} In this work, we present a systematic framework for implementing {\em exact Lindbladian MR}, applicable to many-body quantum systems described by time-independent Hamiltonians and Markovian master equations in finite dimension. This is achieved by lifting the notion of a Krylov subspace to  operator subspaces that are naturally motivated by control theory \cite{kalman1969topics, wonham}: namely, a \textit{reachable subspace}, containing the trajectory of the system when initialized in a given, arbitrary initial condition; or, in a dual manner, an \textit{observable subspace}, containing the trajectory generated by a given observable evolved in Heisenberg picture. 

Our framework may describe a many-body open quantum system $S$ evolving on a (finite-dimensional) Hilbert space $\Hc$ and subject to Markovian dissipation, or it can be applied to reduce the joint (Hamiltonian or Lindblad) dynamics of a many-body bipartite system evolving on $\Hc \equiv \Hc_S\otimes \Hc_E$, with $E$ being a quantum environment -- in which case the reduced dynamics of $S$ alone is generically {\em non-Markovian}. While being suitable for describing arbitrary quantum states (including {\em non-factorized} system-environment initial conditions), our approach explicitly leverages the fact that, in typical applications, 
the dynamical behavior of the system is not only needed solely for {\em certain} initial conditions, but also for only {\em certain} observable properties of interest (e.g., expectations of few-body or collective observables, diagonal elements of the density operator, and so on). This results in the possibility of carrying out dimensional MR based on {\em both} initial conditions and {\em output observables}, which has not been exploited before to the best of our knowledge. As standard in both quantum probability \cite{accardi1982quantum} and rigorous formulations of quantum statistical mechanics \cite{bratteli1,bratteli2}, we define the dynamical models of interest on {\em operator ($*$-)algebras}. By insisting that projections are also taken onto operator $*$-subalgebras, we can guarantee that the reduced models we obtain obey the structural constraints of physically admissible (completely positive and trace-preserving, CPTP) quantum dynamics -- which we prove are {\em still} in Lindblad form. An iterative algorithm is provided to construct a reduced, lower-dimensional Lindbladian model, when one exists. We further show that, even in situations where symmetries are present in the original Lindblad dynamics, our general framework may allow for more effective reduction than achievable by known symmetry-based approaches, through the identification of a new type of {\em observable-dependent symmetries}. 

{\bf Related work.} From a technical standpoint, we borrow some key tools and results from our previous work on MR for classical hidden Markov models \cite{tac2023} and quantum discrete-time dynamics \cite{tit2023}. While our guiding philosophy here is similar, extension to continuous-time Markovian dynamics nonetheless entails, as we will see, several non-trivial extensions. In particular, it is by no means obvious -- and it is in fact remarkable -- that our MR procedures yield a {\em valid Lindblad generator}. Further to that, extension to continuous-time dynamics substantially broaden the class of physical settings and phenomena to which our methods may be usefully applied to, given the foundational importance of Lindblad master equations. 

Concerning early contributions to the use of MR techniques for quantum dynamics, projection-based methods were employed for nonlinear conditional (stochastic) quantum master equations in the context of cavity QED \cite{Ramon}, and later used to exactly derive a special (Maxwell-Bloch-type) class of Lindblad dynamics \cite{Hideo2008}. In \cite{nurdin2014structures}, MR by balanced truncation was exploited to derive reduced (exact or approximate) models for linear quantum stochastic dynamics, with emphasis on providing conditions under which physical realizability constraints are preserved. Closer to our approach, both in spirit and in terms of the algebraic formalism employed, are the work by Kumar and Sarovar we mentioned above  \cite{kumar2014model} and the work in \cite{Kabernik}, where a notion of state-space coarse-graining is formulated for quantum systems. As we already noted, while the approach in \cite{kumar2014model} also achieves exact quantum MR, a quantum model is obtained ``for free'' because only Hamitonian and pure-state dynamics are considered. On the flip side, while the approach of \cite{Kabernik} is not {\em a priori} restricted to Hamiltonian dynamics, the coarse-grained reduced models are not always exact nor provably CPTP or in Lindblad form. More recently, building on the concept of ``lumpability'' of Markov processes \cite{kemenyFiniteMarkovChains1983a}, exact quantum MR techniques based on constrained bisimulation have been developed for boosting the simulation of quantum circuits \cite{jimenez2023efficient}. Aside from the fact that this approach targets discrete-time dynamics (rather than continuous-time as we consider here), as the one in \cite{kumar2014model} it is again restricted to unitary evolution and pure states in current form.
  
While not explicitly cast in the language of MR, methods that describe quantum dynamics using {\em Krylov operator subspaces} have received intense attention recently, in the context of foundational work on universal operator growth \cite{PhysRevX.9.041017} and its implications for quantum chaos, complexity, and scrambling \cite{nandy2024quantumdynamicskrylovspace} -- including extensions to Lindbladian dynamics \cite{Nandy2022,Nandy2023}. 
Central to these approaches is the determination of an operator basis (via a Lanczos or, for dissipative systems, a bi-Lanczos algorithm) relative to which the (vectorized) generator $\Lc$ attains a tridiagonal form, such that the evolution of observables is mapped to a (generally non-Hermitian) tight-binding model on a ``Krylov lattice'', with Lanczos coefficients obeying closed recurrence relationships \cite{nandy2024quantumdynamicskrylovspace}. Although, on a formal level, such Krylov subspaces are related to the observable subspaces we consider, the resulting MR is, in general, approximate. More importantly, in contrast with our approach, translating the positivity constraints of general density operators on the structure of the Lanczos coefficients and the associated dynamics is non-trivial even for Hamiltonian evolution, and understanding how general CPTP constraints manifest in the Krylov representation of a Lindbladian remains an open problem as yet.
    
The fact that MR returns a valid quantum model is crucial for reducing the demands of simulation on a quantum processor. 
Nonetheless, simpler forms of MR suffice if the aim is to perform simulations of quantum systems using classical computers. In this vein, a data-driven approach based on dynamic-mode decomposition was introduced in \cite{goldschmidt2021bilinear} for predicting the behavior of the target system in the presence of control, without {\em a priori} obtaining an accurate characterization. Likewise, a scheme for obtaining approximate expectation values of Pauli observables was presented in \cite{rudolph2023classical}, based on combining a classically constructed ``surrogate'' landscape with truncation techniques.  As in our case, assuming a specified set of initial conditions and output observables plays a key role in these (approximate) approaches. A fundamental difference, however, stems from the fact these ``non-quantum'' MR approaches only aim to a more efficient classical simulation of the output, rather than a proper reduced-order effective model that can be used to represent the dynamics at all times. 

\subsection{Structure and summary of main results}
\label{sub:structure}

The dynamical models we focus on are formally introduced in Sec.\,\ref{sec:problems}, in terms of a Lindblad master equation (in the Schr\"{o}dinger or dual, Heisenberg picture) {paired with a linear output equation} -- which we term a {\em quantum-dynamical semigroup with output}. After exemplifying some basic scenarios of interest that our formalism encompasses, we pose two dimensional MR problems, distinguished by the type of reduced model one is seeking: a reduced CPTP quantum dynamics or, respectively, a general linear dynamical system -- in such a way that the specified outputs are {\em exactly} reproduced for all times. The essential idea we use for obtaining the desired reduction is also anticipated therein (Sec.\,\ref{sub:reduce}): namely, the reduction of the dynamics to an appropriate Krylov operator subspace. 

The next two sections of the paper, Sec.\,\ref{sec:reachable} and Sec.\,\ref{sec:obmr}, are devoted to introduce the mathematical apparatus needed for tackling and solving the MR problems we posed. In particular, the requisite Krylov operator subspaces are built by leveraging the available knowledge of the initial conditions of interest for the system -- formalized in the notion of a {\em reachable subspace}, introduced in Sec.\,\ref{sec:reachablelin} -- and the target output observables for the dynamics -- formalized in the dual notion of an {\em observable subspace}, constructed in Sec.\,\ref{sec:non-observable_ct}). In order to extend such subspaces to operator subalgebras, amenable to support quantum dynamics, some existing structural results on operator $*$-algebras and their associated quantum conditional expectations and ``distortion maps'' are needed; these are recalled in Sec.\,\ref{sec:algebraic_model_reduction}.

In essence, our key results in these theory sections may be summarized as follows.

$\bullet$ {\em Linear reductions.} The solution to the linear MR problem is derived in Theorem \ref{thm:linear_reachable_model_reduction} based on the reachable subspace. Informally, we show that a linear dynamical model of provably minimal dimension and able to reproduce the full state evolution is obtained by ``discarding degrees of freedom'' that are not reached by the dynamics -- formally, by projecting the dynamics onto the Krylov operator subspace generated by the initial conditions. Likewise, by assuming arbitrary initial conditions, a reduced linear dynamical model that is also provably optimal and able to reproduce the full evolution of observables is obtained in Theorem \ref{thm:linear_observable_model_reduction} by discarding degrees of freedom that have no observable output -- formally, by projecting the dynamics onto the Krylov subspace orthogonal to unobservable operators. These two MRs can be combined together to yield a minimal realization of the dynamics. If $n$ is the dimension of the original state space, the reduced model can be obtained with a {\em single} run of an algorithm whose complexity scales may be shown to scale roughly as $n^8$. 

$\bullet$ {\em Quantum reductions.} The solution to the quantum MR problem requires, as mentioned, that the reduced dynamics be themselves defined on an operator $*$-algebra. We obtain such algebras by essentially ``closing'', in an appropriate mathematical sense, the reachable and observable subspaces we mentioned above. If the relevant $*$-algebra, say, $\As$, is smaller than the full operator algebra associated to the system, the key feature we leverage to achieve MR stems from a well-known structural characterization that $\As$ enjoys, the {\em Wedderburn decomposition} \cite{wedderburn1908hypercomplex, arveson}. Pictorially, the idea is illustrated in Fig.\,\,\ref{fig:block-diagonal_representation}: since every $*$-algebra can be seen as the direct sum of possibly repeated irreducible blocks, a smaller-dimensional representation may be obtained by ``compressing'' $\As$ into a reduced $*$-algebra, say, $\check{\As}$, for which no repeated block appears. In Theorem \ref{thm:Lindblad reduction}, we formally prove that a Lindblad generator ${\cal L}$ 

retains its CPTP properties when appropriately restricted to an operator subalgebra in this way, and a reduced Lindblad generator $\check{\cal L}$ may be obtained from ${\cal L}$ through composition with suitable ``reduction'' and  ``injection'' maps.  A constructive algorithm is proposed, to compute the desired reduction starting from the specified set of initial conditions. Notably, the initial states can belong to an {\em arbitrary} linear set, without the need to assume initial factorization in the case a composite system-environment setting, as in standard derivations of master equations \cite{breuer2002theory}. In Sec.\,\ref{sec:observablered}, we derive an equivalent result for reductions based on the observable set of interest, and the two procedures are combined in an iterative algorithm described in Sec.\,\ref{sec:jointred}. In a worst-case scenario, {\em each run} of this algorithm still scales as $O(n^8)$. 

\begin{figure}[t]
\centering
    \includegraphics[width=.9\columnwidth]{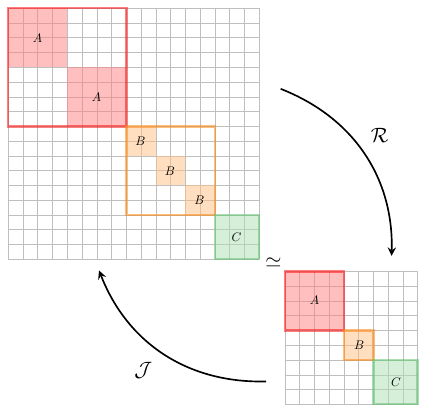}
\caption{\normalfont 
Pictorial illustration of the basic principle enabling quantum MR. The Wedderburn decomposition of a representative  operator $*$-algebra, $\As = ( \one_2\otimes\Cb^{4\times 4})\bigoplus(\one_3\otimes\Cb^{2\times2})\bigoplus \Cb^{3\times3}\subseteq\Cb^{17\times17}$, is depicted in the top, while its compressed representation, $\check{\As} = \Cb^{4\times 4}\bigoplus\Cb^{2\times2}\bigoplus \Cb^{3\times3}\subseteq\Cb^{9\times9}$, is given at the bottom. The action of the reduction and injection maps, $\Rc$ and $\Jc$, is also shown, with formal definitions given in Sec.\,\ref{sec:algebraic_model_reduction} and \ref{sec:lindblad_reduction}.}
\label{fig:block-diagonal_representation}
\end{figure}

Section\,\ref{sec:symred} bridges between the general theory sections and the illustrative applications that will follow, by focusing on the interplay of our general MR procedures with symmetries which, for Lindblad dynamics, are usually categorized as ``strong'' or ``weak'' -- depending on whether, loosely speaking, invariance occurs at the level of individual (Hamiltonian and Lindblad) operators or only at the superoperator level \cite{Viktor,Buca_2012}. We show how viable subspaces and algebras for MR naturally emerge from considering symmetries, and characterize the sense in which a strong symmetry leads to a computational advantage over a weak one, from a dimensional-reduction standpoint. 
Nonetheless, focusing on observable-based MR, we prove in Theorem \ref{maxsymm} that our methods can improve reductions based on both strong and weak symmetries. In fact, finding the relevant output algebra by the procedure described above corresponds to finding ``generalized symmetries'' that emerge only when considering the observables of interest: this new class of {\em observable-dependent symmetries} is defined in Sec.\,\ref{sec:ods}.

\smallskip

In Sec.\,\ref{sec:examples}, we examine concrete examples motivated by condensed-matter physics and quantum information science applications, with a twofold goal: on the one hand, to benchmark our general approach, by showing that it consistently recovers known exact results when available, or it yields the correct numerical results of the full model, as expected; on the other hand, to showcase additional flexibility and resource reduction, \emph{beyond} what achievable with existing methods to the best of our knowledge. Specifically, in Sec.\,\ref{sec:examples_central_spin} we study variations of the well-known central-spin model, beginning with the case where the dynamics is Hamiltonian (no dissipation) and enjoys permutation symmetry on the bath spins. Here, we observe how the introduction of a bath Hamiltonian $H_B$ influences the achievable MR, by showing how, even when $H_B$ breaks permutation symmetry, a useful observable-dependent symmetry group may be identified for physically relevant observables, thereby leading to a reduced model of the same size (Sec.\,\ref{sec:examples_intrabath couplings}). The reduced models may be able to be scaled well-beyond the maximal manageable size of the full model: for example, in some cases we are able to simulate up to 46 bath spins on a typical laptop. We further consider the effect of Markovian dissipation on the bath spins (Sec.\,\ref{sec:examples_central_spin_dissipative}), by contrasting the degree of MR attainable under permutation-invariant (collective) noise operators, in which case a strong symmetry is present, with the one attainable if the dissipation is local, in which case only a weak symmetry remains. The reduced models we obtain are refined in Sec.\,\ref{sec:examples_central_spin_reachable}, where we show how further reachable-based MR may be obtained by considering specific initial states. 

As a second example, in Sec.\,\ref{sec:examples_XXZ} we analyze a boundary-driven XXZ spin chain. Here, we again demonstrate how observable-based MR can be applied to studying physical quantities of interest to quantum transport, beyond the steady-state regime that is usually considered in the literature. We also again contrast the level of MR that is achievable depending on whether a weak or strong symmetry is present in the dynamics. Finally, in Sec.\,\ref{sec:examples_encoding} we consider the smallest noiseless-subsystem code for a non-Abelian class of (collective) errors, based on three physical qubits \cite{KLV}, and illustrate how MR can be implemented on the \emph{logical} level, to obtain a more efficient description of the error dynamics that an encoded qubit undergoes as a result of unwanted error terms. In doing so, we provide another example of MR afforded by an observable-dependent symmetry which is not a (weak) symmetry of the dynamics. 

We conclude in Sec.\,\ref{sec:end} by iterating our key findings and their significance, and by discussing some important open questions -- along with an outlook to future steps. Technical proofs of mathematical results stated in the main text are provided in Appendix\,\ref{appendix:supplementary_results}, whereas considerations on the implementation complexity of our MR procedures are given in Appendix\,\ref{app:complexity}.
Appendixes\,\ref{centralspin} and \ref{sec_appendix_symm} expand, respectively, on MR derivations and symmetry properties of some of the examples we presented in Sec.\,\ref{sec:examples}.

\section{Quantum semigroups with outputs}

\label{sec:problems}

\subsection{Dynamical settings and problem statement}

Throughout this work, we focus on Markovian, time-independent continuous-time quantum dynamics taking place on a Hilbert space $\Hc$ of finite dimension, $n<\infty$, with $\Bf(\Hc)$ denoting the associated space of linear operators. Thus, the dynamics are described by a continuous one-parameter semigroup of CPTP quantum maps $\{\Tc_t\}_{t\geq0}$, with $\Tc_0 = \Ic$ being the identity map and the (homogeneous) composition property $\Tc_t\circ\Tc_s = \Tc_{t+s}$ holding for all $t,s \geq 0$ or, in differential form, by the Markovian generator $\Lc$ \cite{breuer2002theory, alicki-lendi}.  

Let $\Hf(\Hc)$ and $\Df(\Hc) \subset \Hf(\Hc)$ denote the set of self-adjoint (Hermitian) operators and the (convex) set of density operators on $\Hc$. Then {\em for any initial condition} $\rho_0\in \Df(\Hc)$ the time-evolved state in the Schr\"{o}dinger picture, $\rho(t)=\Tc_t(\rho_0)=e^{\Lc t}(\rho_0)$, is the solution of the master equation
\begin{equation}
    \label{eqn:Lindblad}
\dot\rho(t)=\Lc[\rho(t)], \quad t \geq 0, 
\end{equation} 
where is well-known  \cite{lindblad1976generators,Gorini:1975nb} that $\Lc$ takes the following canonical Lindblad form (in units $\hbar =1$): 
\begin{align}
     \Lc (\rho) & \equiv  -i [H,\rho] + \sum_{u}\Dc_{L_u}(\rho) \notag \\
     	& = -i [H,\rho] + \sum_u \Big( L_u \rho L_u^\dag -\frac{1}{2}\{L_u^\dag L_u,\rho\} \Big).
    \label{eqn:Lindblad_generator}
\end{align}
Here, $H=H^\dag$ describes the Hamiltonian contribution to the dynamics and Markovian dissipation is characterized in terms of linear operators $\{L_u\}$ on $\Hc$ which we refer to as noise (or Lindblad) operators. 
Closed-system, unitary dynamics on $\Hc$ is recovered in the limit where all the noise operators (hence the dissipators) vanish.  
Equivalently, the dynamics may be described in the Heisenberg picture by considering the evolution of observables, $X=X^\dag$, under the dual generator $\Lc^\dag$:
\begin{equation}
\dot{X}(t)= \Lc^\dag [X(t)], \quad t\geq 0, 
\label{eqn:Heis}
\end{equation} 
with $ \Lc^\dag (X) =  i [H,X] + \sum_u \big( L_u^\dag X L_u -\frac{1}{2}\{L_u^\dag L_u, X\} \big).$
The solutions of Eq.\,\eqref{eqn:Lindblad} (or, respectively, Eq.\,\eqref{eqn:Heis}) provide access to the time-dependent expectation values of {\em any observable} $X \in \frak{B}(\Hc)$, starting from {\em any initial condition} $\rho_0$: 
$$\langle X(t)\rangle = \tr (X \rho(t)) = 
\tr (X e^{\Lc t}[\rho_0]) = \tr( e^{\Lc^\dag t}[X]\rho_0).$$ 

Physically, the above framework may either directly describe an $n$-dimensional open quantum system evolving on $\Hc$, subject to purely Markovian dissipation, or it allows for Hamiltonian coupling to a quantum ``bath''
to be included -- by considering a bipartition $\Hc \equiv \Hc_S\otimes\Hc_B$, with dim$(\Hc_S)\equiv n_S$, and by defining the reduced state of the system alone via the partial trace operation, $\rho_S(t) \equiv \tr_B(\rho(t))$. For the important case of {\em local} observables on $\Hc_S$, $X \equiv X_S \otimes \one_B$, the relevant time-dependent expectation values may be equivalently computed as 
$$ \langle (X_S \otimes \one_B)(t) \rangle = \tr(X \rho(t)) = \tr_S (X_S \rho_S(t)),$$
with the reduced dynamics of $\rho_S(t)$ being non-unitary (and generally {\em non-Markovian}) even in the case where $\Dc_{L_u}\equiv 0$.  
Regardless, in many situations of fundamental and practical relevance, interest may be {\em a priori} restricted to a subset of initial input states, and a subset of output quantities that depend upon the final, time-evolved state $\rho(t)$ and may be directly associated to or required for computing experimentally accessible properties. For instance, one may want to be able to predict: 

(i) the probabilities associated to the measurement outcomes of a specific observable of interest, say, $X\in\Hf(\Hc)$; 

(ii) the expectation values of a finite set of observables, say, $\{X_i \}\subset\Hf(\Hc)$; 

(iii) the reduced state $\rho_S(t)$ of a subsystem coupled to a quantum bath as above or, more generally, the reduced state of a subsystem of interest, say, $\rho_1(t)$, in a multipartite setting, whereby $\Hc=\bigotimes_j \Hc_j$ and $\rho_1(t)$ is obtained via a partial trace operation on all but the first subsystem, $\rho_1(t)\equiv \tr_{\overline{\Hc_1}} (\rho(t))$.  

Borrowing from system-theoretic terminology, we can think of all of these quantities as \textit{outputs} of the model and view them as a functional of the state, $Y(t) \equiv \Oc[\rho(t)]$, in terms of an {\em output map} $\Oc:\Df(\Hc)\to\Ys,$
with $\Ys \subseteq  \Bf(\Hc)$ being an appropriate operator subspace.
In the above cases, for instance, it is easy to see that the following identifications may be made: (i) $\Oc[\rho(t)] = \sum_j \bm e_j \tr[\pi_j \rho(t)]$, where $\pi_j$ are the spectral projectors of the target observable, $X = \sum_j x_j \pi_j$, $x_j\in {\mathbb R}$, and $\{\bm e_j\}$ the standard basis for $\Ys=\Rb^{|\{\pi_j\}|}$; 
(ii) $\Oc[\rho(t)] = \sum_i \bm e_i \tr[X_i\rho(t)]$, where now $\Ys=\Rb^{|\{X_i\}|}$; 
(iii) $\Oc[\rho(t)] = \tr_B[\rho(t)]$ and $\Ys=\Bf(\Hc_S)$, and similarly if $\Hc$ is multipartite. 
We also stress that, although nonlinear functions of the state $\rho(t)$ are {\em not} considered outputs in this work, 
one can retrieve any nonlinear functions of the output $Y(t)$ by composing the output map with the nonlinear function of interest. For example, in the bipartite system setting, we may obtain the von Neumann entropy of the reduced state by considering $\Hs\circ\tr_B$, with $\Hs(\rho_S) \equiv -\tr(\rho_S\ln \rho_S)$, or similarly for the purity, ${\mathscr{P}}(\rho_S) \equiv \tr(\rho_S^2)$.

More precisely, with the above considerations in mind, in this work we consider models where one is interested in: 
\begin{itemize}
    \item {\em A subset of initial conditions}. Specifically, we shall assume $\Sf \subseteq\Df(\Hc)$ to be either a finite set, or a finitely-generated convex or linear set in $\Df(\Hc)$; 
    \item {\em A subset of output quantities}. Specifically, we shall assume that the outputs of interest depend {\em linearly} on the state $\rho(t)$ of the system. 
The structure of the  output map then takes the following general form:
\begin{eqnarray}
\Oc(\cdot) = \sum_i E_i \tr(O_i^\dag \cdot) ,
\label{eq:omap}
\end{eqnarray}
where $\{E_i\}$ is an orthogonal basis for $\mathscr{Y}$, and $\{O_i\}$ is a finite set of (not necessarily Hermitian) operators on $\Hc$.
\end{itemize} 

Thus, the class of models of interest will consist of a dynamical equation, Eq.\,\eqref{eqn:Lindblad} or \eqref{eqn:Heis}, along with an output equation as in Eq.\eqref{eq:omap}. However, instead of defining the relevant dynamics on the {\em full} operator space ${\mathfrak{B}(\Hc})$, it will be convenient for our MR purposes to allow for restriction to {\em operator $*$-subalgebras}, which provide the most general setting on which physically admissible CPTP evolutions can be defined \cite{accardi1982quantum}, and support descriptions of quantum-classical hybrid systems \cite{barchielli2023markovian,Dammeier2023quantumclassical}, as well as general quantum information encodings \cite{knill-encoding,ticozzi2010quantum}.
Formally, we introduce the following: 

\begin{definition}
Let $\As$ be a $*$-subalgebra and  $\Ys \subseteq {\mathfrak{B}(\Hc})$ an operator space. Given a Lindblad generator $\Lc:\As\mapsto\As$ as in Eq.\,\eqref{eqn:Lindblad_generator}, a linear output map $\Oc:\As\mapsto\Ys$ as in Eq.\,\eqref{eq:omap}, and a set of initial conditions $\Sf\subseteq\Df(\Hc)\bigcap\As$, a {\em quantum dynamical semigroup with output} (QSO) is defined as:
\begin{equation}
    \begin{cases}
        \dot{\rho}(t) = \Lc[\rho(t)] \\
        Y(t) = \Oc[\rho(t)]
    \end{cases}\!\!\!, \quad \rho(0)\in\Sf.
    \label{eqn:QHM_model_ct}
\end{equation}
\end{definition}

Standard semigroups dynamics are recovered if we consider $\As=\frak{B}(\Hc)=\Ys$ and $\Oc = \Ic$, in which case the operators  $\{O_i\}$ in \eqref{eq:omap} form an orthonormal basis for the full operator space. Notably, the QSO class of models defined above can be seen as quantum, continuous-time version of classical {\em hidden-Markov models} \cite{vidyasagarHiddenMarkovProcesses2011}, \cite{tac2023,tit2023}. These models do not include the conditioning effects (back-action) that are associated with a quantum measurement, and describe the expectations of the operators $O_i$ if measurements were effected only at the final time $t$, following evolution from $t_0=0.$ While the effect of multiple, intermediate measurements can be studied within a similar framework to the one we present here \cite{conditional}, {\em unconditional, one-time dynamics} already afford relevant applications, while allowing for the basic framework to be laid out in a simpler way. The fact that in a QSO model we are not considering every initial condition and we are only interested in reproducing certain output quantities raises a natural question: {\em Is it possible to find a smaller model that reproduces exactly the same output trajectories, for all times?}

As stressed in the introduction, a proper quantum structure is both desirable or essential in several contexts. The key problem we solve in this work may be formulated as follows:
 
\begin{problem}
{\bf Continuous-time quantum MR.}
\label{prob:quantum_model_reduction_ct}
Given a QSO $(\Lc,\Oc,\Sf)$ as in Eq.\,\eqref{eqn:QHM_model_ct}, find the (minimal) {\em quantum} model $(\check{\Lc},\check{\Oc},\check{\Sf})$ defined on an algebra $\check {\As},$ with (minimal) $\textrm{dim}(\check {\mathscr A})<\textrm {dim}({\mathscr A})$, and a linear map $\mathcal{R}:\As\to\check\As$, such that for any initial condition $\rho_0\in\Sf$, we have
\[ \Oc e^{\Lc t} (\rho_0) = \check{\Oc} e^{\check{\Lc} t}\Rc (\rho_0), \quad \forall t\geq0.\] 
\end{problem}

In solving this problem, a simpler problem may be formulated, that is of intrinsic interest and useful in applications as well. Namely, we may just want to construct the minimal {\em linear} system, not necessarily of QSO form, that reproduces exactly the output of the system of interest, without imposing any physical consistency constraint. This is very natural, for example, when one looks for the smallest linear system that allows for simulation of a quantum model on a {classical} computer.  A linear MR problem may then be stated as follows:

\begin{problem}
{\bf Continuous-time linear MR.} 
\label{prob:linear_model_red_ct}
Given a QSO $(\Lc,\Oc,\Sf)$ as in Eq.\,\eqref{eqn:QHM_model_ct}, find the (minimal) {\em linear} model $(\Fc_L, \Oc_L,\Vs)$, defined on an operator subspace $\Vs$, with (minimal) $\textrm{dim}(\Vs)\leq \textrm {dim}(\Hc)^2$, of the form
\begin{equation}
    \begin{cases}
        \dot{\xi}(t) = \Fc_L[\xi(t)] \\ Y(t) = \Oc_L[\xi(t)]
    \end{cases}\!\!\!, \quad \xi(0)  \in \Vs , 
    \label{eqn:lin_model}
\end{equation}
and a linear map $\mathcal{R}:\Bf(\Hc)\to\Vs$ such that, for any initial condition $\rho_0\in\Sf$, we have
\[ \Oc [\rho(t)] = \Oc e^{\Lc t}( \rho_0) = \Oc_L e^{\Fc_L t} \Rc (\rho_0),  \quad \forall t \geq 0.\]
\end{problem}

Note that the state $\xi(t)$ of such a linear model need not to have any physical meaning. As an operator, it nonetheless contains all the required information that is necessary to ensure that the outputs of the full and reduced model exactly coincide. As we describe in the following sections, the solution of this relaxed problem is instrumental to solving Problem \ref{prob:quantum_model_reduction_ct}, as it allows one to extract the minimal resources needed for reproducing $Y(t)$ when starting from $\rho_0\in\Sf$.

\subsection{Reductions to operator subspaces}
\label{sub:reduce}

Our general strategy to accomplish MR is inspired by the projection-based MR approach in the control literature \cite{antoulas}. That is, we proceed by first restricting the dynamics to suitable Krylov operator subspaces, which include a set of variables sufficient to reproduce the evolution of the output of interest from the relevant initial conditions, and then by ``closing'' these subspaces to operator $*$-algebras.

Assume that a target subspace, say, $\Ws\subseteq\Bf(\Hc)$ of dimension $d$, has been identified, and let $\Pi=\Pi^2$ be a (not necessarily orthogonal) projection superoperator of rank $d$ on such a subspace. For any $d$-dimensional space $\cal X$ isomorphic to $\Ws$,  $\Pi$ can be factorized as 
\begin{equation}
    \Pi={\cal J} \circ {\cal R},
 \label{eq:fact} 
 \end{equation} 
where the {\em reduction map} ${\cal R}:\frak{B}(\Hc)\rightarrow {\cal X}$ and the {\em injection map} ${\cal J}:{\cal X}\rightarrow \Ws$ are such that ${\cal R}\circ{\cal J}={\cal I}_{\cal X},$ the identity super-operator on ${\cal X}$. Then a QSO of the form \eqref{eqn:QHM_model_ct} can be reduced to ${\cal X}$ by defining:
$${\cal L}_{\cal X} \equiv {\cal R}\circ {\cal L}\circ {\cal J},\quad \Oc_{\cal X}\equiv \Oc\circ {\cal J},$$
and associating to each $\rho\in \frak{S}$ a corresponding $\rho_{\cal X}={\cal R}(\rho).$ Of course, in general such a reduced model will not reproduce correctly the output of the initial model, nor will it maintain the requisite CPTP properties.
The following sections focus on constructing $\Ws$ and ${\cal X}$ so that the resulting reduced model does solve the problems we posed above. 

\section{Reachable-based model reduction}
\label{sec:reachable}

In this section we start to examine the MR Problems \ref{prob:quantum_model_reduction_ct} and \ref{prob:linear_model_red_ct} in a particular case, namely, when the output map of interest $\Oc$ is the identity superoperator. In this scenario, the reduced QSO must be able to reconstruct \emph{all} the states and expectations that are attained by the original model, and the reduction can leverage only the knowledge of the initial condition set.

\subsection{Minimal reachable linear models and their limitations}
\label{sec:reachablelin}

Given an initial condition $\rho_0\in\Sf$, a trajectory $\mathfrak{T}_{\rho_0}$ is the set of density operators that the system's state assumes as a function of time under the specified generator, i.e., $\mathfrak{T}_{\rho_0} = \{\rho(t) = e^{\Lc t}[\rho_0], t\geq 0\}$. By extension, we will denote the set of states belonging to trajectories starting from initial conditions in $\Sf$ as $\mathfrak{T}_\Sf = \bigcup_{\rho_0\in\Sf}\mathfrak{T}_{\rho_0}$. The reachable space $\Rs$ is then the smallest operator subspace generated from $\mathfrak{T}_\Sf$: 

\begin{definition}[Reachable subspace]
Given a QSO $(\Lc,\Oc,\Sf)$ as in Eq.\,\eqref{eqn:QHM_model_ct}, the {\em reachable subspace from $\Sf$} is the operator subspace 
    \begin{equation}
        \Rs \equiv \Span \{e^{\Lc t} (\rho_0),\, t\geq 0,\, \rho_0\in\Sf \} \subseteq  \Bf(\Hc). 
        \label{eqn:reachable_space_ct}
    \end{equation}
\end{definition}

Note that the reachable subspace $\Rs$ is effectively a {\em Krylov operator subspace} \cite{kry31} and, from a system-theory viewpoint, it can be seen as the subspace reachable {\em from the input} of a fictitious linear input-output system (see e.g., \cite[Sec.\,III.A]{cdc2022}).  As in the well-studied classical case, $\Rs$ enjoys some relevant properties that are summarized in the following proposition, whose proof can be found in Appendix \ref{appendix:supplementary_results}:

\begin{proposition}
    \label{prop:reachable_characterization}
    Given a QSO $(\Lc,\Oc,\Sf)$ as in Eq.\,\eqref{eqn:QHM_model_ct}, defined on a subalgebra $\As\subseteq \frak{B}(\Hc),$ with $n=\dim(\Hc)$, the reachable space can be computed as
    \begin{equation}
        \Rs = \Span\{\Lc^i(\rho_0),\, i=0,1,\dots,n^2-1,\,\rho_0\in\Sf\}.
    \end{equation}
    Moreover, $\Rs$ is the smallest $\Lc$-invariant (and $e^{\Lc t}$-invariant) subspace containing $\Span\{\Sf\}$.
\end{proposition}

In linear system theory \cite{kalman1969topics, wonham, marrobasile}, 
it is well known that MR can be obtained by restricting the dynamics to the subspace that is reachable from the input. The next theorem extends this result to {\em any space that contains the space reachable from the initial state}, stated for the particular case of time-invariant Lindblad dynamics we consider. Such an extension will be instrumental to constructing a solution for Problem \ref{prob:quantum_model_reduction_ct}.

\begin{theorem}
\label{thm:linear_reachable_model_reduction}
Consider a QSO $(\Lc,\Oc, \Sf)$, defined on a sub-algebra $\As\subseteq\Bf(\Hc)$, and its reachable subspace $\Rs$. Let $\Vs$ be an operator subspace that contains the reachable space, i.e., $\Rs\subseteq\Vs$, with $\Pi_\Vs$ denoting a (non-necessarily orthogonal) projector superoperator on $\Vs$. Let $\rs$ and $\es$ be two full-rank factors  of $\Pi_\Vs$, i.e., $\es\rs = \Pi_\Vs$, and $\rs\es = \Ic_{\Vs}$, and define $\Fc_L \equiv \rs\mathcal{L}\es,$ $\Oc_L=\Oc\Jc$ and $\Sf_L \equiv \rs(\Sf)$. We then have 
\begin{equation}
    e^{\Lc t}(\rho_0) = \es e^{\Fc_L t}\rs (\rho_0), \quad  \forall t\geq0, \forall \rho_0\in\Sf.
    \label{eq:Rreduction}
    \end{equation} 
Moreover, $\Vs = \Rs$ is an operator subspace of \emph{minimal} dimension for which Eq.\,\eqref{eq:Rreduction} holds.
\end{theorem}

\noindent 
The proof can also be found in Appendix \ref{appendix:supplementary_results}.

\smallskip

Thanks to the above theorem, we can solve the linear MR Problem \ref{prob:linear_model_red_ct} in the case where $\Oc = \Ic$ by simply  restricting the original QSO to $\Rs$. More importantly, the model thus obtained has the \emph{smallest possible dimension}. In this sense, the reachable space $\Rs$ contains the {\em minimal} number of degrees of freedom required to reproduce the state when starting from initial conditions $\rho_0\in\Sf$.  The dimension of $\Rs$ can still grow exponentially with the number of physical degrees of freedom (e.g., $N$ qubits) in the original system, however, depending on the complexity of $\Lc$.

\smallskip

\begin{remark} It is worth noting that one can obtain an alternative representation of the linear reduced model, which may be especially convenient for numerically simulating a quantum system on a classical computer, by exploiting a standard vectorization procedure over $\Rs$ \cite{Havel}. Let $\{R_i\}$ be an orthogonal operator base for $\Rs$ and define 
$${\rm vec}_{\Rs}(\cdot) \equiv \sum_i \bm e_j \tr (R_i^\dag \cdot)  , \quad
{\rm unvec}_\Rs(\cdot) = \sum_i R_i \bm e_i^T (\cdot) ,$$
where $\{\bm e_i\}$ is the standard base for $\Cb^{\dim(\Rs)}$ and, by construction, ${\rm unvec}_\Rs\circ{\rm vec}_\Rs  = \Pi_\Rs$. The reduced state $\bm x(t) \equiv {\rm vec}_{\Rs}(\rho(t))$ then evolves in time according to the linear differential equation $\dot{\bm x}(t) = F \bm x(t)$, with the generator $F\in\Cb^{\dim(\Rs)\times\dim(\Rs)}$ computed as $[F]_{i,j} = \tr[R_i^\dag \Lc (R_j)]$. The time-evolved state of the system is retrieved by letting $\rho(t) = {\rm unvec}_\Rs[x(t)]$, $t \geq 0$. 
\end{remark}

\smallskip

Despite being natural and practically useful in some cases, the reduction on the reachable operator subspace $\mathscr R$ does not retain, in general, the structure of the original QSO model: namely, ${\cal F}_L$ is not, in general, a generator of a one-parameter semigroup of CPTP maps. Since the QSO models we consider are defined on operator *-subalgebras, in order to obtain a reduced model in the same class we need, at a minimum, to construct a supporting algebraic structure. Before introducing the general formalism needed to develop our results, some additional motivating discussion may be in order.

A natural approach for using $\mathscr R$ to build an algebraic model is to consider the smallest *-algebra containing $\mathscr R,$  $\mathscr{A}=\textrm{alg}(\mathscr{R}).$  This choice is particularly convenient, as any {\em orthogonal} projection onto a *-subalgebra $\Pi_\As$ can be shown to be a CPTP map \cite[Thorem II.6.10.2]{blackadar2006operator}, which admits a further factorization in full-rank maps of the form $\Pi_{\As}={\cal J}\circ {\cal R}$ as in Eq.\,\eqref{eq:fact}, with all maps involved being also CPTP \cite[Theorem 1]{tit2023}. This property, as we shall see in the following, is {\em sufficient} to ensure that the reduced generator $\check\Lc=\Rc\Lc\Jc$ is also a Lindblad generator. However, this choice of $\As$ is in general {\em non-optimal}, as illustrated by the following example.

\begin{example}
\label{example:distorted}
Consider a generator $\Lc$ with a fixed point $\rho$, i.e. $\Lc(\rho)=0$, and take such an equilibrium as an initial condition for the OQS, $\Sf=\{\rho\}$. The reachable space is then $\Rs = \Span \{\rho\},$ and  $\alg(\Rs) = \Span\{\Pi_i\}$, where $\Pi_{i}$ are the spectral projections of $\rho.$ Thus, the dimension of the reduced model, $\dim(\alg(\Rs)),$ equals the number of distinct eigenvalues of $\rho$.
However, by choosing $\check{\rho}(0)=1$, $\check{\Lc} = 0$ and $\es(\check{\rho}) = \rho\check{\rho}$, we have that the model $
\dot{\check{\rho}}(t) = \check{\Lc}[\check{\rho}(t)]$ and $\rho(t) = \es(\check{\rho}(t))$ with initial condition $\check{\rho}(0)=1$, provides the correct trajectory at all times, i.e. $\rho(t)=\rho$ for all $t\geq0$. Since the dimension of this model is 1, clearly it is minimal. 
\end{example}

The above explicitly shows that, although closing $\Rs$ to an algebra provides a valid reduced quantum model, this model is not guaranteed to be minimal. A strategy that avoids this issue is to seek the minimal {\em distorted algebra} \cite{blume2010information}, that is, a *-algebra with respect to a modified operator product, that contains $\Rs$ and is the image of the dual of a conditional expectation. The following subsections are dedicated to introducing these concepts and their use in the context of this work.

\subsection{Algebraic model reduction: Fundamentals}
\label{sec:algebraic_model_reduction}

\subsubsection{Finite-dimensional $*$-algebras and their representations}
\label{sec:algebraic_model_reduction_1}

We consider associative operator $*$-algebras $\As\subseteq\Bf(\Hc)$ which, in our finite-dimensional setting, are simply operator subspaces closed with respect to the standard matrix product and adjoint operation \cite{blackadar2006operator}.  

Given a set $\Sc\subseteq \Bf(\Hc)$ we define its {\em commutant}, denoted with $\Sc'$, as the set of operators that commute with every element of the set, that is, 
$$\Sc'=\{X\in\Bf(\Hc)|[X,S]=0,\, \forall S\in\Sc\}.$$ 
Note that $\Sc'$ is always a (not necessarily *-) algebra and if $\Sc$ is closed under the adjoint operation, it becomes a $*$-algebra. In particular, $\As'$ is a *-algebra if $\As$ is such.
Given an algebra $\As$, we also define its {\em center} as
$$\zentrum(\As) = \As \cap \As' = \{X\in\As|[X,A]=0, \,\forall A\in\As\}.$$ 
Again, $\zentrum(\As)$ is a commutative algebra and if $\As$ is a $*$-algebra, then $\zentrum(\As)$ is a commutative $*$-algebra.

A key result regarding the structure  of *-algebras, which has been extensively used in the theory of virtual quantum subsystems \cite{KLV, ViolaJPA,ZanardiVS} and information-preserving structures
 \cite{knill-encoding},\cite{blume2010information},\cite{ticozzi2010quantum},
is the {\em Wedderburn decomposition} \cite{wedderburn1908hypercomplex, arveson}. Given any associative *-algebra $\As \subseteq\Bf(\Hc)$, there exists a decomposition of the Hilbert space, say,
\begin{equation}
\Hc \equiv \bigoplus_{k} \big( \Hc_{F,k}\otimes\Hc_{G,k} \big) \oplus \Hc_R, 
\label{vs}
\end{equation}
and a unitary operator $U\in\Bf(\Hc)$ effecting a transformation to a basis in which the algebra and its commutant can be simultaneously decomposed as 
\begin{align}
    \As &= U\bigg(\bigoplus_k \Bf(\Hc_{F,k}) \otimes \one_{G,k} \oplus {\mathbb O}_R \bigg)U^\dag,
    \label{eq:wedderburn}\\
    \As' &= U\bigg(\bigoplus_k \one_{F,k} \otimes \Bf(\Hc_{G,k}) \oplus \Bf(\Hc_R) \bigg)U^\dag, 
    \label{eq:wedderburn_commutant}
\end{align}
whereby it also follows that
$$ \zentrum(\As) = U\bigg(\bigoplus_k \lambda_k \one_{F,k} \otimes \one_{G,k} \oplus {\mathbb O}_R\bigg)U^\dag.$$
When considering a matrix representation of these operator space, these decompositions amount to a simultaneous block-diagonalization  (numerical algorithms to find such a decomposition are provided in \cite{de2011numerical}). A property that directly follows from the above structural characterization and will be useful later is the following:
\begin{equation}
\label{eq:commperp}
\As\ominus \zentrum(\As)\,\perp\, \As',
\end{equation}
where orthogonality is meant here with respect to the standard Hilbert-Schmidt inner product. 

As it will become clearer in the next sections, to the purposes of analyzing the MR problems we are interested in, we can assume that $\As$ is {\em unital}, that is, it contains the identity operator. In this case, the summand $\Hc_R$ in Eq.\,\eqref{vs} is empty, and we will thus not further consider it in what follows. In case the considered algebra is not unital, we can focus only on its support, effectively removing the term $\zero_R$ \cite{tit2023, MyThesis}.

\subsubsection{Quantum conditional expectations and distorted algebras} 
\label{sec:algebraic_model_reduction_2}

Assume we are interested in describing and making prediction for a subalgebra of observables $\As$, knowing that the state of the system is $\rho.$ A reduced statistical description can then be obtained by suitably projecting the relevant operators in $\As$: the map enacting this ``coarse-graining'' is precisely the {\em dual of a conditional expectation} \cite{petz2007quantum}. In the following, we illustrate this idea and extend it to the reduction of dynamics.

\begin{definition}[Conditional expectation]
Given a unital $*$-algebra $\As \subseteq \Bf(\Hc)$, a {\em conditional expectation} onto $\As$, $\CE_\As:\Bf(\Hc)\to\As$, is a positive projector, that is,
$$\CE_\As(X) = X, \;\forall X\in\As, \quad \CE_\As(X^\dag X)\geq 0, \;\forall X\in\Bf(\Hc).$$ 
\end{definition}

While idempotent, a conditional expectation need {\em not} be self-adjoint and thus, in general, it is not an orthogonal projection. The adjoint operator of the conditional expectation with respect to the Hilbert-Schmidt inner product, denoted by $\SE_\As (\cdot) \equiv \CE_\As^\dag( \cdot)$, is known as the {\em state extension} \cite{petz2007quantum}. Some properties of the conditional expectation directly follow from its definition. In particular, since $\As$ is unital, the conditional expectation $\CE_\As (\cdot) $ is a positive, unital projector. Then its dual, $\SE_\As (\cdot)$ is a positive and TP projector. It is possible to show that, in fact, both maps are actually CP \cite{blackadar2006operator}.

The state extension is the map we need to reduce the description of the system in the Schr\"{o}dinger's picture. Let $X\in\As\subsetneq \Bf(\Hc)$ and $\rho\in \Df(\Hc).$ Then, in computing the expectation for $X$ according to $\rho$ we have:
$$\langle X \rangle = \tr(X\rho) =\tr[\CE_\As(X)\rho]=\tr[X\SE_\As(\rho)].$$
Therefore, a reduced description that preserving the expectations of a subalgebra of observables $\As$ with respect to any state can be obtained via the CPTP projector $\SE_\As.$
In doing so, we are reducing the set of density operators to the {\em image of $\SE_\As.$} In order to specify what kind of set this is, it is convenient to resort to the algebra structure and derive the induced form of the maps $\CE_\As,\SE_\As$.  Let the algebra $\As$ have a decomposition of the form \eqref{eq:wedderburn}. Then, there exists a set of full-rank density operators $\{\tau_k\in\Df(\mathcal{H}_{G,k})\}$ such that
\begin{widetext}
    \begin{align}
         \CE_\mathscr{A}(X) = U\bigg( \bigoplus_{k=0}^{K-1} \tr_{\Hc_{G,k}}\left[(W_k X W_k^\dag)(\one_{d_k}\otimes\tau_k)\right]\otimes\one_{G,k}\bigg) U^\dag , \qquad \forall X\in\mathcal{B(H)}, 
        \label{eqn:cond_exp_blocks}
    \end{align}
\end{widetext}
where $W_k$ is a linear operator from $\Hc$ onto $\Hc_{F,k}\otimes\Hc_{G,k}$ such that $W_k W_k^\dag=\one_{F,k}\otimes \one_{G,k}$ \cite{wolf2012quantum}. Under the same assumptions, the corresponding state extension onto $\mathscr{A}$ takes the form: 
\begin{align}
     \SE_\mathscr{A}(X) 
        &= U\bigg[ \bigoplus_{k=0}^{K-1} \tr_{\Hc_{G,k}}\left(W_k X W_k^\dag \right)\otimes \tau_k\bigg] U^\dag ,
    \label{eqn:state_ext_blocks}
\end{align}
whereby it is straightforward to see that
\begin{equation}
    \textrm{image}\SE_\mathscr{A}= U\bigg(\bigoplus_k \Bf(\Hc_{F,k}) \otimes \tau_{k}\bigg)U^\dag.
    \label{eq:distorted_Wedderburn}
\end{equation}
This set is also a *-closed algebra, but with respect to a {\em modified} operator composition. Define the {\em $\tau$-modified product} as  
$$\tau\equiv \bigoplus_k \one_{s,k}\otimes \tau_k, \quad  A{\circ}_\tau B \equiv A \tau^{-1}B.$$
It is easy to show that
$\As_\tau\equiv \textrm{image}(\SE_\As)$ is a linear operator subspace that is both *-closed and closed with respect to the $\circ_\tau$ product. Sets of this type are called {\em distorted algebras} \cite{blume2010information,wolf2012quantum} and are often used to characterize the set of fixed points of CPTP maps \cite{PhysRevA.66.022318,johnson2015general}. 

The main idea towards finding a quantum, CPTP MR is to extend the reachable subspace so that it is a distorted algebra. However, two key aspects must be taken into account: (i) we need the distorted algebra to be the image of a state extension; and (ii) we want the minimal such distorted algebra. Point (i) is nontrivial, and it is related to the existence of a conditional expectation whose dual preserves a state. While in the context of classical probability (which can be recovered within the present framework by considering Abelian algebras \cite{tac2023}) the existence of a conditional expectation fixing any state is always guaranteed, for quantum systems this need {\em not} be the case. Necessary and sufficient conditions for this to be the case are provided in Takesaki's modular theory and its specializations to the finite-dimensional case \cite{TAKESAKI1972306,petz2007quantum,tit2023}.

These issues have been addressed in \cite{tit2023} for discrete-time dynamics, and we recall here the main result we will need:

\begin{theorem}[Minimal compatible distorted algebra \cite{tit2023}]
\label{thm:minimalreachablealgebra}
Consider an operator subspace $\Vs\subseteq\Bf(\Hc)$ and a positive-definite operator $V\in\Vs$, $V=V^\dag>0$. Let us define 
\begin{equation}
\sigma \equiv \Pi_{\zentrum(\alg(\Vs))}\left(V\right)\in\zentrum(\alg(\Vs)),
\label{sigma}
\end{equation}
where $\Pi_{\zentrum(\alg(\Vs))}$ is the orthogonal projection onto the center ${\zentrum(\alg(\Vs))}$. Then, $\alg_\sigma(\Vs)$ is the minimal distorted algebra that both contains $\Vs$ and admits a state extension. 
\end{theorem}

\subsubsection{Reduced representation of states and observables  \\ 
on subalgebras}

As we anticipated in Sec.\,\ref{sub:structure}, the Wedderburn decomposition of $\As$ let us naturally reduce the size of the representation of a $*$-subalgebra by removing the repeated blocks. Explicitly, if $\As= U\left(\bigoplus_k \Bf(\Hc_{F,k}) \otimes \one_{G,k}\right)U^\dag,$ with $m\equiv \sum_k \dim(\Hc_{F,k})$, we can observe that it is isomorphic to $\check{\As} = \bigoplus_k \Bf(\Hc_{F,k}) \subseteq \Cb^{m\times m}$ (recall Fig.\,\ref{fig:block-diagonal_representation}). If we aim to reduce a model onto a $*$-sub-algebra $\As$, it then suffices to consider its compressed representation $\check \As$. As shown in \cite{tit2023}, such a reduction can be achieved by using CP unital or CPTP maps that factorize the conditional expectation or its dual, respectively:

\begin{theorem}[Factorization of conditional expectations  \cite{tit2023}]
\label{thm:factorizations}
Let $\As\subset\Bs$ be a unital $*$-subalgebra with a Wedderburn decomposition  as in \eqref{eq:wedderburn}, and let $\check \As\equiv \bigoplus_k \Bf(\Hc_{F,k}).$ Then, for any conditional expectation $\CE_\As$ and state extension $\SE_\As$, there exist (non-square) factorization maps
\begin{equation}
\label{eq:cond_exp_factorization}
\CE_\As = \eo \ro, \quad \SE_\As = \es \rs ,
\end{equation}
where  $\eo:\check\As\rightarrow \As,$ $\ro:\As\rightarrow \check\As$   are CP and unital, and $\es:\check\As\rightarrow \Dc_{\tau}(\As),$  $\rs:\As\rightarrow \check\As$  are CPTP, such that for any $X\in\Hf(\Hc)$, $\rho\in\Df(\Hc)$ we have 
\begin{equation}
\tr[\CE_\As(X)\rho]=\tr[\ro(X)\rs(\rho)]=\tr[X\SE_\As (\rho)].
\end{equation}
\end{theorem}

Explicitly, the block-structure of the CPTP  linear maps of interest is found to be as follows:
\begin{equation}
\label{eqn:reduction}
\Rc(X)=\bigoplus_k \tr_{\Hc_{G,k}}(W_k X W_k^\dag)=\bigoplus_k X_{F,k}=\check X ,
\end{equation}
and
\begin{equation}
\label{eqn:injection}
\es(\check X)=U\bigg(\bigoplus_k X_{F,k} \otimes\tau_{k}\bigg)U^\dag.
\end{equation}
We note that, in \cite{Kabernik}, $\Rc$ is referred to as a {\em quantum coarse-graining} map, while in \cite{petz2007quantum} this terminology is used for $\Jc$.

\begin{remark}
If the operator subspace $\Vs$ does not have full support, the resulting algebra $\alg(\Vs)$ is not unital. Consequently, $\Zc(\alg(\Vs))$ is also not unital and the projected operator $\sigma$ in \eqref{sigma} does not have full support. In such a case, we can restrict the model onto the support of $\Vs$ and then continue with the restriction of $\alg(\Vs)$ onto its support, which is now unital. Notice that by doing this we obtain a reduction that is not properly CPTP, but is only CPTP over the support of $\alg(\Vs)$. This does not pose a problem for our MR procedure, however, because in the subspace of $\Hc$ where the operator space $\Vs$ has no support we have no probability.
\end{remark}

\subsection{Reduction of continuous-time semigroup dynamics on \mbox{*-subalgebras}}
\label{sec:lindblad_reduction}

Before summarizing our MR procedure, we need to determine what happens to the restriction of a Lindblad generator $\Lc$ to a distorted algebra. While, for discrete-time dynamics, the reduction automatically preserves the CPTP character of the original map since $\Jc,\Rc$ are also CPTP, in the continuous-time case it is not {\em a priori} obvious that the reduction yields a valid Lindblad generator. In the following theorem, which is one of our main results, we prove that the reduction is indeed a CPTP semigroup generator that leaves the algebra invariant. While some supporting mathematical results are provided in Appendix \ref{appendix:supplementary_results}, the main idea of the proof is simple and we include it here due to the importance of the result.

\begin{theorem}[Reduced Lindblad dynamics]
\label{thm:Lindblad reduction}
Let $\mathscr{A}$ be a unital $*$-subalgebra of $\mathcal{B(H)}, $ and let $\rs$ and $\es$ denote the CPTP factorization of $\SE_\mathscr{A} = \es\rs$, as defined in Eq.\,\eqref{eq:cond_exp_factorization}. Then for any Lindblad generator $\Lc$, its reduction to $\As$, 
$$\check{\Lc}\equiv\rs\Lc\es, $$ 
is also a Lindblad generator, that is, $\check{\Lc}:\check\As\to\check\As$ and $\{e^{\check{\Lc}t}\}_{t\geq 0}$ is a quantum dynamical semigroup. 
\end{theorem}
\begin{proof}
Let us start by considering a Lindblad generator in the form $\Lc (\rho) \equiv {\Phi (\rho)} - K\rho -\rho K^\dag$ \cite{wolf2012quantum}, for some operator $K\in\Bf(\Hc)$ and for some CP superoperator $\Phi$ such that $\Phi^\dag (\one) = K + K^\dag$. We can then consider 
$$\check{\Lc}(\check{\rho}) = \rs\Phi\es(\check{\rho}) - \rs[K \es(\check{\rho})] - \rs[\es (\check{\rho}) K^\dag].$$ 
From Proposition \ref{prop:combining_reduction_injection} in Appendix \ref{appendix:supplementary_results}, we then have that $\check{\Lc}(\check{\rho}) = \rs\Phi\es (\check{\rho}) - \es^\dag (K) \check{\rho} - \check{\rho}\es^\dag (K^\dag)$, where $\es^\dag(K)\in\check{\As}$. Moreover, from Throrem \ref{thm:factorizations}, we know that both $\rs$ and $\es$ are CPTP and hence $\rs\circ\Phi\circ\es$ is CP. To conclude, we can verify that $\check{\Lc}^\dag (\one) = 0$ by observing that 
\begin{align*}
(\rs \Phi \es)^\dag (\one) & = \es^\dag \Phi^\dag \rs^\dag (\one) = \es^\dag  \Phi^\dag (\one) \\
&= \es^\dag (K + K^\dag) =  \es^\dag (K) +\es^\dag (K^\dag),
\end{align*}
where we have used the fact that $\rs^\dag$ is unital and linearity. It then follows from Theorem 7.1 in \cite{wolf2012quantum} that $\check{\Lc}$ is a generator of a quantum dynamical semigroup.
\end{proof} 

A limiting case of interest is closed-system dynamics, described by a Hamiltonian.

\begin{corollary}
\label{corollary:Hamiltonian}
Under the assumptions of theorem \ref{thm:Lindblad reduction}, assume that the Lindblad generator is non-dissipative, that is, $\Lc = -i\ad_H$, with Hamiltonian $H\in\Hf(\Hc)$. Then the reduced generator $\check{\Lc}\equiv \rs\Lc\es$ is also non-dissipative, $\check{\Lc}=  \ad_{\check{H}}$, with the reduced Hamiltonian \(\check{H} \equiv \es^\dag(H) \in\check{\mathscr{A}}\).
\end{corollary}

We now have all the ingredients to construct a reduced QSO that is able to exactly reproduce the state $\rho(t)$ of the original model when starting from $\rho_0\in\Sf$. The key steps may be summarized as follows:
\begin{enumerate}
\item Compute the reachable space $\Rs$ from $\Sf.$ We assume that $\Rs$ has full support: if this is not the case, we can reduce the model to the supporting subspace.
\item Find a positive self-adjoint element $V$ of $\Rs.$ 

\item Compute the commutant algebra $\As'\equiv\Rs',$ and orthogonally project $V$ onto $\As'.$ {Thanks to the property highlighted in Eq.\,\eqref{eq:commperp},} this is sufficient to obtain $\sigma$ as in Theorem \ref{thm:minimalreachablealgebra}. 

\item Compute the distorted reachable algebra $\Ds$. This can be done by using $\Ds=\sigma^{1/2}\alg(\sigma^{-1/2}\Rs \sigma^{-1/2})\sigma^{1/2},$ see \cite{tit2023} for detail. Note that, in cases where $\Rs = \text{alg}(\Rs)$, or $\one\in\Rs$, one can simply take $\Ds = \Rs$, as we have $\Rs \subseteq \Ds \subseteq \alg(\Rs)$. 

\item Compute the unitary change of basis $U$ that puts $\Ds$ into the Wedderburn-decomposed form \eqref{eq:distorted_Wedderburn} and $\Jc,\Rc$ as in Theorem \ref{thm:factorizations}. 

\item Define the reduced generator $\check \Lc \equiv \Jc\Lc\Rc$ on $\check\As$ and the initial condition $\check\rho_0\equiv \Rc(\rho_0).$ The output function for the reduced model in this case is simply $\Oc=\Jc.$
\end{enumerate}

Note that Step 2 can be completed by considering, for some $\epsilon>0$, the following operator:
$$V\equiv \sum_{\rho_0\in \Sf}
\int_{0}^\epsilon e^{{\cal L}t}(\rho_0)dt.$$
In fact, a continuous-time linear system spreads over its reachable set in arbitrarily small time \cite{marrobasile}. Integrating over time and summing over all possible initial condition yields a positive-semidefinite operator with maximal support. By Theorem \ref{thm:Lindblad reduction}, the obtained generator is still in Lindblad form, and by Theorem \ref{thm:linear_reachable_model_reduction}, it reproduced exactly the dynamics of interest. We have thus provided a solution to Problem \ref{prob:quantum_model_reduction_ct} for the case where the QSO has an identity output function.

\section{Observable-based model reduction}
\label{sec:obmr}

\subsection{Minimal observable linear models}
\label{sec:non-observable_ct}

In this section we tackle a MR problem that is dual to the one considered in Sec.\,\ref{sec:reachable}, namely, we assume no restriction on the initial conditions, $\Sf=\Df(\Hc)$. In this case, the reduction can only leverage the observables of interest:
intuitively, the key idea here is to discard from the model degrees of freedom that produce no output. 

More formally, let us first focus on the set of initial states that are \textit{indistinguishable} by having access only to the semigroup output for all times. Two states $\rho_0,\rho_1$ are said to be indistinguishable at the output (or output-indistinguishable for short) if $\Oc e^{\Lc t}(\rho_0) = \Oc e^{\Lc t}(\rho_1)$, for all $t\geq0$. By linearity, two states are output-indistinguishable if and only if $\rho_0-\rho_1\in\ker\Oc e^{\Lc t}$, for all $t\geq 0$. Generalizing this idea to operators yields the \textit{non-observable subspace}:

\begin{definition}[Non-observable subspace]
\label{def:nos}
    Given a QSO $(\Lc,\Oc,\Sf)$ as in Eq.\,\eqref{eqn:QHM_model_ct}, the {\em non-observable subspace  under $\Oc$} is the operator subspace 
    \begin{equation}
    \Ns \equiv \{X\in\As\,|\,\Oc e^{\Lc t}(X)=0, \;\forall t\geq0\} \subseteq \Bf(\Hc).
    \label{eqn:non_observable_space_ct}
    \end{equation}
\end{definition}

In analogy to Proposition \ref{prop:reachable_characterization}, an explicit characterization of the non-observable subspace is provided by the following: 

\begin{proposition}
    Given a 
    QSO $(\Lc,\Oc,\Sf)$ as in Eq.\,\eqref{eqn:QHM_model_ct}, defined on a sub-algebra of $\As\subseteq\Bc(\Hc)$, with $n=\dim(\Hc)$, the non-observable space is given by 
    \begin{equation}
     \Ns = \{X\in\As \,|\,\Oc\Lc^i(X)=0,\, \forall i=0,1,\dots,n^2-1\}.
    \end{equation}
    Moreover, $\Ns$ is the largest $\Lc$-invariant (and $e^{\Lc t}$-invariant) subspace contained in $\ker{\Oc}$. 
\end{proposition}

\noindent 
The proof of this proposition, up to vectorization of the matrices, is a standard linear system-theory result \cite{wonham,marrobasile,kalman1969topics}. 

An equivalent approach for characterizing the non-observable subspace is to consider its orthogonal complement in the Heisenberg picture. Consider a linear output map $\Oc$ of the general form given in \eqref{eq:omap}, $\Oc(\cdot) = \sum_i E_i \tr\left(O_i\cdot\right)$, where for convenience we now assume that $O_i$ represent physical observables, thus $\{O_i\} \subseteq\Hf(\Hc)$ (the procedure works also for general $O_i\in\Bf(\Hc)$, however). From Definition \ref{def:nos}, it follows that $X\in\Ns$ if $\Oc e^{t\Lc}(X) = \sum_i E_i \tr\left(O_i e^{t\Lc}[X]\right) =0 $, for all $t\geq 0$. Since the $E_i$ are linearly independent, this is equivalent to requiring that $\tr[O_i e^{t\mathcal{A}}(X)] = \tr[e^{t\Lc^\dag}(O_i)X]  =0$, for all $t\geq 0$ and all $i$. In other words, the observable $X$ is indistinguishable from the zero observable (it is unobservable) if and only if $X\perp \Ns^\perp$, with 
\begin{equation}
  \Ns^\perp\equiv\Span\{e^{t\Lc^\dag}(O_i), \; \forall t\geq0, \forall i\} \subseteq \Bf(\Hc),
  \label{eqn:non_obs_Heisenberg}
\end{equation}
or, equivalently, as presented in \cite{d2021introduction},
\begin{equation}
\Ns^\perp = \Span\{ \Lc^{\dag j}(O_i), \; \forall i, \forall j=0,\dots, n^2-1\}.
    \label{eqn:non_obs_Heisenberg_2}
\end{equation}
Similar to the reachable space $\Rs$, the observable space $\Ns^\perp$ is also a Krylov operator subspace.
The counterpart to Theorem \ref{thm:linear_reachable_model_reduction} can then be stated as follows:

\begin{theorem}
\label{thm:linear_observable_model_reduction}
Consider a QSO $(\Lc,\Oc, \Sf)$, defined on a sub-algebra $\As\subseteq\Bf(\Hc)$, and its non-observable subspace $\Ns$. Let $\Vs$ be an operator subspace such that $\Ns^\perp\subseteq\Vs$, with $\Pi_\Vs$ denoting a (non-necessarily orthogonal) projection superoperator on $\Vs$. Let $\rs$ and $\es$ be two factors of $\Pi_\Vs$ such that $\es\rs = \Pi_\Vs$ and $\rs\es = \Ic_{\Vs}$, and define $\Lc_L \equiv \rs\mathcal{L}\es$ and $\Oc_L \equiv \Oc\es$. We then have 
\begin{equation}
    \Oc e^{\Lc t}(\rho_0) = \Oc_L e^{\Lc_L t}\rs(\rho_0), \quad \forall t \geq 0, \forall \rho_0\in\Df(\Hc).
    \label{eq:non_obs_th}
\end{equation}  
Moreover, $\Vs = \Ns^\perp$ is an operator subspace of {\em minimal} dimension for which Eq.\,\eqref{eq:non_obs_th} holds.
\end{theorem}

\noindent 
The proof follows the same lines of the one of Theorem \ref{thm:linear_reachable_model_reduction}, only switching to the dual (Heisenberg) picture. Again, this provides a valid and minimal linear reduction, however there is no guarantee that the reduced generator $\Lc_L$ is a valid semigroup generator. We next aim to extend $\Ns^\perp$ so it can support CPTP dynamics.

\subsection{Extension to quantum models}
\label{sec:observablered}

In order to leverage Theorem \ref{thm:linear_observable_model_reduction} to construct a CPTP model, we seek the smallest distorted algebra that admits a CPTP projection and contains the complement of $\Ns$, so that we can apply Theorem \ref{thm:linear_reachable_model_reduction}.  A general result, formally proved in the discrete-time setting, is provided by the following theorem:

\begin{theorem}[\bf Reduction on the output algebra \cite{tit2023}]
\label{thm:non_observable-projection}
Consider a QSO $(\Lc,\Oc, \Sf)$, defined on a sub-algebra $\As\subseteq\Bf(\Hc)$, and its non-observable subspace $\Ns$. The {\em output algebra} $\Os \equiv \alg(\Ns^\perp)$ is a ({distorted}) algebra of minimal dimension that contains $\Ns^\perp$ and allows for a state extension, in this case $\SE_\Os=\CE_\Os.$  
\end{theorem}

It should be noted that the choice of $\Os$ as a distorted algebra is highly \emph{non-unique} in general. In fact, any distortion of $\Os$ with a full-rank state with compatible block structure yields the same results.  We now have all the ingredients to construct a QSO, whose output $Y(t)=\Oc (t)$ is identical to that of the original model, for arbitrary initial conditions in $\Sf$.  The key steps are outlined in the following:

\begin{enumerate}
    \item Compute the orthogonal to the non-observable subspace, $\Ns^\perp$, from $\{O_i\}$ as in Eq.\,\eqref{eqn:non_obs_Heisenberg_2}. 
    We assume $\Ns^\perp$ has full support: if this is not the case, we can immediately reduce the model to the supporting subspace.
    
    \item Compute\footnote{It is worth noting that, to the purpose of obtaining the relevant  Wedderburn decomposition, computing the commutant $(\Ns^\perp)'$ is equivalent, since the Wedderburn decomposition of the commutant automatically gives the one of the algebra, see e.g.,  Eqs.\,\eqref{eq:wedderburn} and \eqref{eq:wedderburn_commutant}. Interestingly, in implementing this algorithm for the examples of Sec.\,\ref{sec:examples}, we observed that the numerical calculations of the Wedderburn decomposition were more efficient and numerically stable through the calculation of the commutant instead of the algebra itself.} 
    the {output algebra} $\Os=\alg(\Ns^\perp).$ This can be done as a double commutant $\Os=(\Ns^\perp)''.$ Given $\Os$, find the unitary change of base $U$ that brings $\Ns^\perp$ and $(\Ns^\perp)'$ to their canonical Wedderburn decomposition.
    \item Consider the orthogonal projection $\SE_\Os=\CE_\Os.$ Compute $\Jc,\Rc$ as in Theorem \ref{thm:factorizations}.   
    \item Define the reduced generator $\check \Lc \equiv \Jc\Lc\Rc$ on $\check\As$ and the output function for the reduced model  $\check\Oc \equiv \Oc \Jc.$
\end{enumerate}

By Theorem \ref{thm:Lindblad reduction}, the generator $\check \Lc$ is still in Lindblad form, and by Theorem \ref{thm:linear_observable_model_reduction}, it reproduces exactly the output of the system of interest. We have thus provided a solution to Problem \ref{prob:quantum_model_reduction_ct} for QSOs with unrestricted initial condition. 

\begin{remark}
Notice that in the above procedure for observable-based MR, differently from the reachable-based case, we close $\Ns^\perp$ to a {\em standard} algebra, rather than a distorted one. One may wonder whether, by using a distorted algebra, one could obtain a smaller reduced model also in the observable MR; however, this is not the case. Intuitively, one can understand this as follows: algebras are images of conditional expectations, which are maps acting on observables; as such, they form the natural spaces to describe observables. Distorted algebras are images of state extensions, which act on states; hence they are the natural structure to describe states. As shown in Eq.\,\eqref{eqn:non_obs_Heisenberg}, $\Ns^\perp$ is a space of Heisenberg-evolved observables and it is thus natural to close this space to an algebra, not to a distorted one. Formally proving that there is no advantage in closing $\Ns^\perp$ to a distorted algebra is rather involved, and we do not further elaborate on that here (see \cite[Theorem 6]{tit2023} for technical details).
\end{remark}

\subsection{Joint reachable- and observable-based model reduction}
\label{sec:jointred}

So far we have considered the problem of constructing MR descriptions based on the knowledge of either a limited set of initial conditions or observables of interest. In scenarios where both are restricted, it is natural to ask whether the reductions can be combined. In this section, we illustrate how to do so in order to obtain linear or CPTP joint reductions.

If only {linear models} are sought (Problem \ref{prob:linear_model_red_ct}), it is possible to leverage existing results on minimal realizations of input-output linear systems to show that a {\em minimal realization} of the dynamics can be obtained by considering the so-called {\em effective subspace} \cite{ito1992identifiability,tac2023,tit2023}:
\begin{equation*}
    \mathscr{E} \equiv \Rs\ominus (\Ns\cap\Rs).
\end{equation*}
This subspace supports a description of the reachable operators that yield non-trivial output. The reduction can then be obtained by first using Theorem \ref{thm:linear_reachable_model_reduction} to reduce to the dynamics to $\Rs$, and next by applying Theorem \ref{thm:linear_observable_model_reduction} to the already reduced model to  obtain a reduction to $\mathscr{E}.$ Importantly, in the linear case, one can prove that exchanging the order of the two MRs leads to a different representation of the {\em same dimension} \cite{marrobasile}.

If the goal is to obtain CPTP models that solve Problem \ref{prob:quantum_model_reduction_ct}, however, our method requires that the appropriate reachable or, respectively, observable sets are {\em enlarged} to algebras. In doing so, we lose the guarantee of minimality. The best we can do is to iterate the two types of MR until no further reductions are obtained for either. The key steps to obtain a reduced Lindblad model are outlined in Algorithm \ref{algo:composed}. Notice that one could obtain an alternative MR algorithm by applying the reduction on the output algebra {\em before} the reduction on \(\Ds\) at each iteration. Unlike in the linear case, however, the results of the two reductions are potentially different, and need {\em not} have the same dimension in general.

\begin{algorithm}[t]
    \caption{Iterative reduction}
    \label{algo:composed}
    \SetAlgoLined
    \Input{A QSO model $(\Lc, \Oc , \Sf)$ defined on $\As \equiv \As_0.$}
    Assign $(\Lc_0, \Oc_0 , \Sf_0) \equiv (\Lc, \Oc , \Sf).$ 
    
    Use the reachable reduction based on 
    $\Ds$ on model $(\Lc_0, \Oc_0 , \Sf_0)$ to
    compute  $\SE_\Ds=\Jc_1\Rc_1;$ define the model $(\Lc_1=\Rc_1\Lc_0\Jc_1, \Oc_1=\Oc_0\Jc_1 , \check\Sf_1=\Rc_1\Sf_0) $.
    
    Use the observable reduction based on $\Os$ on model $(\Lc_{1}, \Oc_1,  \Sf_{1})$ to compute $\SE_\Os=\Jc_2\Rc_2;$ define the model $(\Lc_2=\Rc_2\Lc_1\Jc_2;, \Oc_2=\Oc_1\Jc_2 , \check\Sf_2=\Rc_2\Sf_1)$. 
    
    \If{$\dim(\As_0) \neq \dim(\Os)$}{
    Assign $\As_{0}\equiv \Os$ and $(\Lc_0, \Oc_0 , \Sf_0) \equiv (\Lc_2, \Oc_2 , \Sf_2).$ 
    Go back to step {\bf2}.}
    \Else{
    \Output{$(\Lc_2, \Oc_2 , \Sf_2)$ defined on $\Os.$}}
\end{algorithm}

\section{Symmetry-based model reduction}
\label{sec:symred}

Symmetries are arguably the most natural property one can leverage in seeking to reduce the ``complexity'' of a given model. Since the notion of symmetry is more nuanced for open quantum systems than it is for the limiting case of Hamiltonian dynamics, we first recall some basic notions as relevant to the class of Markovian semigroups we consider. We then discuss how known approaches that exploit symmetries for simplifying the dynamical problem may be subsumed by our approach, in an appropriate sense -- as long as exact solutions are sought for the dynamical quantities of interest.
Crucially, we further argue that, even in situations where symmetries are known, our system-theoretic framework is more powerful, and can in fact result in a larger degree of MR, as effectively induced by a new type of {\em observable-dependent symmetries}. 

\subsection{Symmetries for quantum dynamical semigroups}

For Hamiltonian dynamics, Wigner's theorem \cite{Ballentine} mandates that a symmetry transformation of a physical system is represented by a unitary or anti-unitary operator that commutes with the Hamiltonian. Similarly, a conserved quantity is an observable (a self-adjoint operator) that commutes with the Hamiltonian. The two notions are in direct correspondence with one another: On the one hand, for each continuous family of (necessarily) unitary symmetries, say, $S(\theta) = e^{i\theta Q}$, with $\theta\in\mathbb{R}$, the associated generator $Q$ is conserved. On the other hand, for each conserved quantity $Q=Q^\dagger$, $S(\theta) = e^{i\theta Q}$ yields a continuous family of symmetries. This may be taken as a realization of Noether's theorem in (non-relativistic) quantum mechanics. 

In extending the above picture to open quantum systems, two distinct notions of symmetry naturally arise \cite{Viktor, Buca_2012, VincentPRB}. Let us focus specifically on Markovian semigroup dynamics and {\em unitary} symmetries. We then have the following:

\begin{definition}[Symmetries of semigroup dynamics] 
\label{def:osym}
Let $\{\Tc_t\}$ be a CPTP continuous semigroup with Lindblad generator $\Lc$.

(i) A {\em weak symmetry} is a unitary operator $S$ that leaves the dynamics invariant, 
\begin{align}
\label{eq:sym}
\mathcal{T}_t(S\rho S^\dag) = S \mathcal{T}_t(\rho) S^\dag, \quad \forall t, \forall \rho\in \Df (\Hc),
\end{align}
or, equivalently, $[\Sc,\Tc_t]=0$, $\forall t$, in terms of the super-operator  $\mathcal{S}(\cdot) \equiv S\cdot S^\dag$.

(ii) A {\em strong symmetry} is a unitary operator $S$ that commutes with the Hamiltonian and all the noise operators in $\Lc$,
\begin{align}
[H,S] = 0,\quad [L_u,S]=0, \quad\forall u, 
\label{eqn:strong_symmetry}
\end{align}
where $\Lc \sim (H, \{L_u\})$ is an arbitrary representation of the semigroup generator as in Eq.\,\eqref{eqn:Lindblad_generator}. Equivalently, $\Sc(H) = H$ and $\Sc(L_u)=L_u$, $\forall u$. 
\end{definition}

Thanks to the continuity of $\mathcal{T}_t,$ it is immediate to see that property \eqref{eq:sym} is equivalent to $[\mathcal{S},\mathcal{L}] = 0,$ or also $[\Sc,\Lc^\dag] =0$.
It is also immediate to verify that every strong symmetry is a weak symmetry ({\em not} vice-versa) and that strong symmetries are independent of the chosen Lindblad representation, as assumed in the definition. 

Clearly, a Markovian system may have more than one weak (or strong) symmetry. In that case, we can construct a set $\Gs$ of unitary operators such that Eq.\,\eqref{eq:sym} (or, respectively, Eq.\,\eqref{eqn:strong_symmetry}) is obeyed for all $S\in\Gs$. In both cases, the set $\Gs$ is actually a group because, given $S_1,S_2\in\Gs$, it also follows that $S_1S_2\in\Gs.$ Such a group is often referred to as the {\em symmetry group} of the dynamical system. 

Assume that $\Gs$ is a continuous group of unitary operators -- say, a one-parameter group in the simplest instance, whereby we may write $\Gs \equiv \{S(\theta) \equiv e^{i\theta G}\}$, with $\theta \in {\mathbb R}$ and $G=G^\dag$. Then each $S(\theta)\in\Gs$ is a weak symmetry if and only if 
\begin{align*}
\mathcal{L}([G,\rho]) = [G,\mathcal{L}(\rho)],\quad \forall \rho \in \Df (\Hc).
\end{align*}
That is, $G$ generates a group $\Gs$ of weak symmetries (is a {\em symmetry generator}) if and only if the adjoint action of $G$, namely, $[G,\cdot]$, commutes with $\mathcal{L}$. The equivalent condition in the Heisenberg picture is 
\begin{align*}
\mathcal{L}^\dag([G,X]) = [G,\mathcal{L}^\dag(X)], \quad \forall X\in \Bf (\Hc).
\end{align*}
\noindent 
Similarly, $G$ is the generator of a continuous group of strong symmetries if and only if $[G,H] = [G,L_u] =0$, for all $u$. 

As in the case of Hamiltonian dynamics, a {\em conserved quantity} under Lindblad dynamics is a self-adjoint operator $Q$ that is invariant (is a constant of motion) in the Heisenberg picture, namely, it satisfies $\mathcal{L}^\dag(Q) = 0. $ This implies that arbitrary expectations are preserved,
$$\expect{Q}\!(t)  = \tr[Q e^{\Lc t}(\rho_0)] = \tr[e^{\Lc^\dag t}(Q) \rho_0] = \expect{Q}\!(0),$$ for any state. 
Despite these similarities, the connection between conserved quantities and symmetries is no longer straightforward, however \cite{Viktor}. In particular, while there is a one-to-one correspondence between generators of strong
symmetries and conserved quantities $Q$ whose moments are all conserved, there is in general no direct correspondence between generators of weak symmetries and conserved quantities. This is referred to as a {\em breakdown of Noether's theorem} in Markovian systems \cite{VincentPRB}.

\subsection{Symmetries and invariant subspaces}

In the context of MR, the relation between symmetries and invariant subspaces plays a key role. For any {\em weak} symmetry operator $S$, the eigendecomposition of the superoperator $\Sc$ provides a decomposition of the operator space, $\Bf(\Hc) = \bigoplus_\nu \Bs_\nu$, where $\Bs_\nu$ are operator-eigenspaces associated to distinct eigenvalues $\nu$, i.e., $\Sc(X) = \nu X$ for all $X\in\Bs_\nu$. Since $\Sc$ and $\Lc$ commute, it follows that each subspace $\Bs_\nu$ is $\Lc$-invariant \cite{Buca_2012,McDonald}: 
if $X$ is a $\nu$-eigenoperator of $\Sc$, with $\Lc(X) \equiv Y$, we have \(\Sc(Y)=\Sc\Lc(X)=\Lc\Sc(X) = \nu\Lc(X) = \nu Y \), hence $Y$ is also a $\nu$-eigenoperator of $\Sc$.

If we further assume that $S$ is a {\em strong} symmetry for $\Lc$, it is also known (Theorem A.1 in \cite{Buca_2012}) that the Lindbladian $\Lc$ can be block-diagonalized: formally, each operator-eigenspace $\Bs_\nu$ can be further decomposed into 
\[\Bs_\nu = \bigoplus_{
{j,k \,|\, u_j\bar{u}_k=\nu}} \Bs_{j,k}, \] 
with $\Bs_{j,k} \equiv \Span\{\ketbra{\phi}{\psi}, U\ket{\phi}=u_j\ket{\phi},  U\ket{\psi}=u_k\ket{\psi}\}$ and {\em each} \(\Bs_{j,k}\) $\Lc$-invariant. 

The identification of operator subspaces that are invariant under the dynamics generated by $\Lc$ immediately provides a mechanism for MR: if an operator subspace, say $\Vs$, is known to be $\Lc$-invariant, we can restrict our attention to evolution inside of $\Vs$ for arbitrary initial conditions $\rho_0\in\Vs$. 
Within our MR framework, we can then construct two factors $\Jc$ and $\Rc$, with $\Pi_\Vs = \Jc\Rc$, such that $\xi_0 = \Rc(\rho_0)$, $\dot{\xi}(t)= \Fc_L[\xi(t)]$, $\Fc_L = \Jc\Lc\Rc$ and $\rho(t)=\Jc[\xi(t)]$. 

While symmetries thus naturally lead to MR and allow, in fact, for considerable flexibility in constructing invariant subspaces, such subspaces constructed need {\em not} be minimal whenever initial conditions (compatible with the symmetry constraint) is specified from the outset. In contrast, the approach proposed in Theorem \ref{thm:linear_reachable_model_reduction} provably returns the smallest invariant subspace possible, for a fixed set of initializations. Likewise, for a fixed set of observables of interest, the approach in Theorem \ref{thm:linear_observable_model_reduction} provably yields the smallest model that is capable of reproducing the output for all possible initial conditions.
As we stressed, however, restricting the dynamics to operator subspaces does not ensure that the reduced dynamics is consistent with quantum CPTP constraints in general.  In order to exploit symmetries to obtain reduced {\em quantum} models, it is additionally necessary to relate symmetries to operator algebras. A natural connection that has been extensively studied \cite{Dygen,PhysRevA.63.012301,Kabernik} may be established as follows.

Let $\Gs$ be a group of weak symmetries for the dynamics generated by $\Lc$. The {\em group algebra} of $\Gs$ and its commutant in $\Bf(\Hc)$ are then, respectively, defined by
\begin{align*}
    \Cb \Gs & \equiv\Span_{\Cb} \{\Gs\}, \quad \text{dim}(\Cb \Gs) \leq |\Gs|, \\
    \Cb\Gs'& \equiv\{X\in\Bf(\Hc)|\,[S,X]=0, \forall S\in\Gs \}.   
\end{align*}
Both $\Cb \Gs$ and $\Cb \Gs'$ are unital, *-algebras and, in particular, $\Cb \Gs'$
is an $\Lc$-invariant algebra\footnote{ The fact that $\Cb \Gs'$ is $\Lc$-invariant follows from the fact that $\Cb\Gs'$ can be seen as the intersection of all the $1$-eigenspaces of $\Sc$ for all $S\in\Gs$. The invariance property is even more evident in the case of strong symmetries: since $\{H,L_u\}\subseteq \Cb\Gs'$ and, by definition of algebra, products and sums of operators in $\Cb\Gs'$ remain in $\Cb\Gs'$, it follows that $\Lc(\Cb\Gs')\subseteq\Cb\Gs'$.}.   
Assuming that the initial condition respects the symmetry of the dynamics, $\rho_0\in\Cb\Gs'$, we can then find a reduced quantum model that correctly reproduce the output state at all times. In particular, we can compute two CPTP maps $\Jc$ and $\Rc$ as in Theorem \ref{thm:factorizations} that are factors of the self-adjoint state extension $\SE_{\Cb\Gs'}=\Jc\Rc$. A reduced model may be constructed by letting 
$$\check{\rho}_0 = \Rc(\rho_0), \;\check{\Lc}=\Rc\Lc\Jc, \;\dot{\check{\rho}}(t) = \check{\Lc}[\check{\rho}(t)], \;\rho(t) = \Jc[\check{\rho}(t)].$$ 
By Theorem \ref{thm:Lindblad reduction}, we know that $\check{\Lc}$ and $\check{\rho}(t)$ are a valid Markovian semigroup and density operator, respectively. While the above is a legitimate symmetry-induced MR, a drawback is that it need not be optimal; in particular, the use of distorted algebras may in principle allow for a larger reduction (a {\em smaller}-dimensional model) to be obtained. 

In the next section, we argue that our general quantum MR approach genuinely extends symmetry-based MR. Specifically, by focusing on {\em observable-based} MR, we show that not only can our approach consistently recover all the conclusions that may be reached on the basis of both weak and strong symmetries, but it can lead to a more substantial reduction than what is possible using standard symmetry notions alone. 

\begin{remark}
As we mentioned in the Introduction, symmetries also provide the appropriate mathematical framework for characterizing Hilbert space fragmentation \cite{fragmentation}. The main idea is to consider a parametrized {\em family} of Hamiltonians, $H = \sum_k h_k H_k$, with $h_k \in {\mathbb R}$ and $H_k$ a few-body Hamiltonian, and construct the associative ``bond algebra'' $\Xs\equiv\alg(\{H_k\})$, generated by arbitrary linear combinations of
arbitrary products of the $H_k$'s. Fragmentation may then be defined in terms of a volume scaling of the dimension of the commutant algebra (equivalently, the number of dynamically disconnected Krylov subspaces) with system size, log[dim$(\Xs')]\sim L^d$, where $d$ is the spatial dimension. Several other features of fragmentation (e.g., its ``classical'' vs.\,``quantum'' nature) are further informed by the Wedderburn decomposition of $\Xs$ on $\Hc$ \cite{fragmentation}. Lindblad dynamics of fragmented systems have also been considered \cite{li2023hilbert}, the relevant ``open bond algebra'' being generalized to $\Xs\equiv\alg(\{H_k\})\cup \{L_j\})$, with $\{L_j\}$ being the noise operators in $\Lc$. 
Because $\Xs$ is an algebra, quantum MR is possible in principle, as we have focused on. 
In their current form, however, these analyses do not provide explicit reduced descriptions of the dynamics (for which the use of distorted $*$-algebras is needed in general), and only applies to strong symmetries -- in contrast with the MR procedures we propose\footnote{Interestingly, an approximate form of quantum MR for a purely dissipative class of Lindbladians exhibiting ``operator-space fragmentation'' has been reported in \cite{Essler2020}. In this case, $\Bf (\Hc)$ splits into exponentially many operator-Krylov subspaces, and the projected evolution on each subspace is governed by an integrable Hamiltonian.}. 
\end{remark}

\subsection{Symmetries and observable-based model reduction}
\label{SymmRed}

Let us now consider a Lindblad generator $\Lc$, with associated symmetry group $\Gs$, and let $\{O_i\}$ be the set of output observables of interest. By standard group-representation theory, the action of $\Gs$ (and $\Cb \Gs$) on $\Hc$ can be described in terms of its direct-sum decomposition into irreducible representations (irreps henceforth) \cite{Fulton}, $\mathcal{H} = \bigoplus_{k} \mathcal{H}_k$, where $k$ labels inequivalent irreps of $\Gs$. $\Cb \Gs$ acts irreducibly on $\Hc$ if (and only if) its commutant is trivial, in the sense that $\Cb \Gs ' = \Cb \one = {\{\lambda \one\}},$ $\lambda \in \Cb$. In the case where symmetries are present, the action of $\Cb \Gs$ is reducible, and the $k$-th irrep, with dimension, say, $d_k$, appears in general with a multiplicity $n_k$. By an appropriate unitary change of basis, $U\in \Bf(\Hc)$, $\Cb \Gs$ and $\Cb \Gs'$ can then be brought to their Wedderburn block-diagonal form \cite{KLV,Dygen,ZanardiVS},
\begin{align}
    \label{groupalg}
    \mathbb{C}\Gs &=  U\bigg(\bigoplus_{k}  \one_{F,k}\otimes \Bf(\Hc_{G,k})\bigg)  U^\dag, \\
    \label{commutant}
    {{\mathbb C}\Gs'}&= U\bigg(\bigoplus_{k} \Bf(\Hc_{F,k})\otimes\one_{G,k}\bigg)U^\dag, 
\end{align}
with an associated virtual-subsystems decomposition as in Eq.\,\eqref{vs}, subject to $\dim(\Hc_{F,k})= n_k$,  $\dim(\Hc_{G,k})= d_k$,  and $\sum_k n_k d_k = \dim(\Hc)=n$.

The following proposition establishes a direct connection between the symmetry group of the generator $\Lc$ and the observable space $\Ns^\perp$:
 
\begin{proposition}
\label{strongcond} 
Let $\{O_i\} \subset \Hf (\Hc)$ be a set of observables, and let $\Lc$ be a Lindblad generator with an associated group $\Gs$ of (weak) symmetries, that is, $[\Lc,\Sc]=0$, for all $S\in\Gs$. If $\{O_i\}\subseteq \mathbb{C}\Gs'$, we have \(\Ns^\perp \subseteq \mathbb{C}\Gs'.\)
\end{proposition}
\begin{proof}
Observe that $\Sc\Lc^{\dag k}(O_i) = \Lc^{\dag k}\Sc(O_i) = \Lc^{\dag k}(O_i)$, $\forall k\geq 0, \forall i$ where, in the second equality, we have used the fact that, since $O_i \in \mathbb{C}\Gs'$, we may write
\begin{equation}
O_i=U\bigg(\bigoplus_{k} C_{ik}\otimes \one_{G,k}\bigg)U^\dag, \quad C_{ik}\in\Bf(\Hc_{F,k}).
\label{obs}
\end{equation}
It follows that, for all the generators of $\Ns^\perp = \Span\{\Lc^{k \dag}(O_i),\, \forall k\geq 0,\, \forall i \}$, we have that $\Lc^{\dag k}(O_i)\in\Cb\Gs'$. Since all the generators of $\Ns^\perp$ are contained in $\Cb\Gs'$, any of their linear combination $X\in\Ns^\perp$ is also in $\Cb\Gs'$.
\end{proof}

Several remarks are worth making. First, notice that the above proposition implies the following chain of inclusions: 
\begin{equation}
\Ns^\perp\subseteq \alg\{\Ns^\perp\} =\Os  \subseteq \Cb\Gs' , 
\label{inclusions}
\end{equation}
since $\alg\{\Ns^\perp\}$ is the smallest $*$-algebra that contains $\Ns^\perp$ by definition. This allows us to conclude that: 
(i) The algebra $\Cb\Gs'$, via Theorem \ref{thm:non_observable-projection} and the procedure described in Sec.\,\ref{sec:observablered}, may be used to obtain a valid quantum MR, with $\Os \equiv \Cb\Gs'$. (ii) The resulting observable-based MR is (at least) as good as a reduction that directly exploits the knowledge of symmetries, with the associated irrep block-decomposition of Eqs.\,\eqref{groupalg}-\eqref{commutant} being computed through known algorithms \cite{de2011numerical,Wang,Kribs,NumFragmentation}. Note that knowing that a model exhibits some symmetries may help in computing $\Ns^\perp$, as one could perform the symmetry-based MR first, then compute $\Ns^\perp$ on the resulting reduced model more efficiently and verify whether further reduction is possible or not.

Second, the assumption $O_i\in\Cb\Gs'$ implies that the observables of interest have the block-structure given in Eq.\,\eqref{obs}, leading to \(\check{O}_i=\bigoplus_{k} C_{ik}\). Thus, we also have $\check{\rho}=\Rc(\rho) \equiv \bigoplus_k \rho_k$. In other words, the only parts of the reduced state that are necessary to reproduce the expectation values of the observables $O_i$ are those in $\Cb\Gs'$, with the rest of the information contained in $\rho$ being irrelevant for this purpose. 

Third, since the reduced generator $\check{\Lc}$ leaves $\check{\As}$ invariant by construction, the reduced Hamiltonian and noise operators need to obey appropriate conditions for this to holds. A general characterization for a semigroup generator to leave a $*$-subalgebra invariant have been recently developed in \cite{hasenohrl2023generators}. If $\Gs$ is a group of strong symmetries, the commutation of each symmetry operator in $\Gs$ with $H, L_u$ further leads to a MR of a particularly simple form. Since we have 
\[ \left \{ \begin{array}{l}
H=U\bigg(\bigoplus_{k} H_k\otimes \one_{G,k}\bigg)U^\dag, \\
L_u=U\bigg(\bigoplus_{k} L_{u,k}\otimes \one_{G,k}\bigg)U^\dag, \quad \forall u, 
\end{array} \right . \]
the corresponding reduced operators read 
\begin{align}
\label{reducedObHamLind}
\check{H}=\bigoplus_{k} H_k, \quad 
\check{L}_{u}=\bigoplus_{k}
 L_{u,k},
\end{align}
all belonging to the reduced operator algebra $\check{\As} \equiv \bigoplus_k \Bf(\Hc_{F,k})$. One can easily verify that having Hamiltonian and noise operators of the form given in Eq.\,\eqref{reducedObHamLind} suffices for $\check{\As}$ to be $\check{\Lc}$-invariant.

\subsection{Observable-dependent symmetries}
\label{sec:ods}

While the above discussion shows how our approach recovers standard symmetry-based approaches, we now show how a generalization of the notion of weak symmetry may lead to more general reductions than achievable by symmetry alone. In doing so, we also highlight the origin of the advantage that our systematic observable-based MR technique affords.

\begin{definition}[Observable-dependent symmetries]
\label{property1}
    Given a Lindblad generator $\Lc$, a set of observables $\{O_i\}$ and a unitary operator $S\in \Bf(\Hc)$, with associated 
    super-operator $\Sc(\cdot) = S \cdot S^\dag$, we say that $S$ is a $\{O_i\}$-{\em (observable-)dependent symmetry} (ODS) if
    \begin{equation}
    \label{defODS}
    \Sc \Lc^{\dag \,n}(O_i) = \Lc^{\dag \,n}(O_i), \quad \forall n \in {\mathbb N}, \forall i.
    \end{equation}
\end{definition}

Essentially, in contrast with the definition of a weak symmetry in Eq.\,\eqref{eq:sym}, the commutation between the unitary superoperator and the dynamics need \emph{not} hold in general, but only when the action of the superoperators is restricted to a designated set of observables, which is also required to be invariant under the symmetry transformations. In fact, Eq.\,\eqref{defODS} is a compact way to require both of these conditions:
   \[ \left \{ \begin{array}{l}
    \Sc\Lc^{\dag n}(O_i) = \Lc^{\dag n}\Sc(O_i), \quad \forall n\geq 0, \, \forall i, \\  
    \Sc(O_i)=O_i,\quad  \forall i.
	\end{array} \right .    \]
Equivalently, the ODS property in Eq.\,\eqref{defODS} may be restated as
\[\Sc e^{\Lc^\dag t}(O_i) = e^{\Lc^\dag t}(O_i), \quad\forall t, \, \forall i.\] 
\noindent 
It is immediate to verify that the set of ODS operators for a fixed set of observables also forms a group. 

Given the above definitions, a generalization of Proposition \ref{strongcond} follows:
\begin{corollary}
\label{generalcon}
Let $\{O_i\} \subset \Hf (\Hc)$ be a set of observables, and let $\Lc$ be a Lindblad generator with an associated group $\Gs$ of ODS operators, that is, Eq.\,\eqref{defODS} holds for all $S\in\Gs$.  
Then $\Ns^\perp\subseteq \Cb\Gs'$.
\end{corollary}
The proof follows identical steps to those in Proposition \ref{strongcond}. Remarkably, as we formally prove in Appendix \ref{appendix:supplementary_results}, reducing the dynamics to the commutant of the ODS group of interest turns out to be equivalent to constructing the reduction to the output algebra $\Os$:

\begin{theorem}
\label{maxsymm}
Let $\{O_i\} \subset \Hf (\Hc)$ be a set of observables, and let $\Lc$ be a Lindblad generator. Then $\alg\{\Ns^\perp\}\subsetneq\Bf(\Hc)$ if and only if there exists a non-trivial $\{O_i\}$-ODS for $\Lc$. Furthermore, we have 
$$\Os=\alg\{\Ns^\perp\}=\Cb\bar\Gs',$$ 
where $\bar\Gs$ is the largest group of ODS for the system.
\end{theorem}

Altogether, we have shown that knowing that a system admits a set of weak or observable-dependent symmetries allows us to directly identify a valid algebra for the reduction, via the commutant $\Cb\Gs'$. In what follows, we will illustrate the applicability and usefulness of the general framework and the concepts we have developed in paradigmatic examples motivated by quantum many-body physics and quantum computation.

\section{Illustrative applications}
\label{sec:examples}

In all the settings under consideration, we shall work with a finite number $N$ of spin-1/2 (qubits), with associated Hilbert and operator space
$$\Hc = \otimes_{i=1}^{N} \Hc_i, \quad \Hc_i\simeq\Cb^2,\quad \Bf(\Hc) \simeq (\Cb^4)^{\otimes N}. $$ 
As usual, $\sigma_q,$ with $q\in\{0,x,y,z\}$ and $\sigma_0\equiv \one_2$, denote Pauli spin-1/2 matrices, while we will use the shorthand notation $\sigma_q^{(k)} \equiv \otimes_{i=1}^{k-1} \sigma_0 \otimes \sigma_q \otimes_{i=k+1}^{N}\sigma_0$ to denote their multi-spin extension, with $\sigma_q^{(k)}\in\Bf(\Hc)$, acting non-trivially only on the $k$-th spin when $q\equiv u \in \{x,y,z\}$. Similarly, we will define local ladder (raising and lowering) operators acting non-trivially on spin $k$ as $\sigma_\pm ^{(k)} \equiv \frac{1}{2}(\sigma_x^{(k)}\pm i\sigma_y^{(k)})$.

\subsection{Dissipative central spin model}
\label{sec:examples_central_spin}

We consider a central spin system, $S$, coupled to a ``structured quantum environment,'' namely, an interacting spin bath, $B$, responsible for generally non-Markovian dynamics on $S$, and also sometimes referred to as a ``non-Markovian core'' \cite{Tamascelli}, along with a bath inducing {Markovian dissipation on $B$} alone. Models of this form have long been studied from the point of view of their exact solvability \cite{Gaudin1976,Gerardo}, as well as their relevance to both condensed-matter physics and quantum information. In particular, central spin models have been used to describe the decoherence experienced by a spin qubit coupled to a spin bath, e.g., an electron in a semiconductor quantum dot or a nitrogen-vacancy center in diamond \cite{Glazman,Witzel,Slava,Coish,onizhuk2023understanding,Lili,Hollenberg}. Likewise, they have been employed to represent networks of qubits connected in a spin-star topology in the context of state transfer and entanglement generation \cite{Burgarth2004,Hutton2004,Yung2011} and for studies of dissipative phase transitions \cite{kessler2012dissipative,Carollo}. 

Explicitly, in what follows we label the central spin by $1$, while the remaining $N-1\equiv N_B$ spins correspond to bath spins, whereby $\Hc_B\simeq\Cb^{2^{N_B}}$. The full dynamics for the joint system-bath state $\rho(t) \in \Df(\Hc_S \otimes \Hc_B)$ is determined by a Lindblad master equation of the following form: 
\begin{eqnarray}
\dot{\rho}= {\cal L}(\rho)= -i \big [H_{SB} , \rho \big]  
+\sum_{L_B} \Dc_{L_B}(\rho), \;\;
\label{CSL}
\end{eqnarray}
where each of the Markovian dissipative generators ${\cal D}_{L_B}$ acts non-trivially only on the bath spins, and the joint system-bath Hamiltonian reads
\begin{equation}
H_{SB} = H_S + H_{B} + H_{\text{int}}.
\label{HSB}
\end{equation}
Physically, $H_S$ and $H_B$ account for the bare evolution of the central spin and the bath, respectively, whereas $H_{\text{int}}$ is responsible for coupling $S$ to $B$, hence for the ensuing non-Markovian spin-bath decoherence. For the purpose of our discussion, it is convenient to write $H_B\equiv H_{B,0}+H_{B, \text{int}}$, with $H_{B, \text{int}}$ specifically describing intra-bath spin couplings. In applications, the goal is to reproduce the trajectory of the reduced state on the central spin, that is, the relevant output $Y(t)\equiv \rho_S(t) =\tr_B(\rho(t))$. Equivalently, we can choose the set of Pauli observables for the central spin, 
$$\{O_u\}\equiv\{\sigma_u^{(1)}, u \in\{x, y,z \} \},$$ 
and identify $Y(t)$ with the Bloch vector of their time-dependent expectation values, or with any functional of interest determined by these expectations, for instance, the Von Neumann entropy of the central spin, given by $ \Hs(\rho_S(t)) = - \tr[\rho_S(t) \ln \rho_S(t) ].$

In the remaining of this subsection, we examine dynamical generators of increasing complexity. A fixed choice for the central-spin and interaction Hamiltonians, $H_S$ and $H_{\text{int}}$ in $H_{SB}$, will serve as a backbone for all these generators, while different choices of $H_B$ and $L_B$ are used to assess the achievable observable-based MR. Specifically, in Eq.\,\eqref{HSB} we assume that 
\begin{equation}
\label{csbare}
H_S= \frac{1}{2} \Big( \omega_1 \sigma_z^{(1)} + \eta  \sigma_x^{(1)} \Big), 
\quad \omega_1, \eta \in \Rb. 
\end{equation}
Furthermore, we assume that the central spin also couples to the bath spin via a possibly anisotropic but {\em collective} interaction, as described by an XYZ Hamiltonian of the form 
\begin{align}
H_{\text{int}} \equiv {\frac{1}{2}} \Big( A_x\sigma_x^{(1)} J_x +  A_y\sigma_y^{(1)} J_y + A_z\sigma_z^{(1)}J_z \Big),
\label{Hint}
\end{align}
with $J_u\equiv \tfrac{1}{2} \sum_{k=2}^N \sigma_u^{(k)}$ denoting total bath-spin angular momentum operators, and $A_u \in {\mathbb R}$ being uniform strength parameters, $A_u^{(i)} \equiv A_u, \forall i>1$.
Physically, this corresponds to requiring that no inhomogeneity or local disorder is present in the couplings. In the limit of spatially isotropic couplings ($A_u\equiv A, \forall u$), the above further reduces to the Heisenberg (or XXX) Hamiltonian, $H_{\text{int}}^{\text{Heis}} \equiv \frac{1}{2}A\, \vec{\sigma}^{(1)} \cdot \vec{J}$; in this case, the interaction has full rotational symmetry, $[H_{\text{int}}^{\text{Heis}}, J_u]=0$, and the resulting central-spin dynamics is known to be exactly solvable for arbitrary \emph{factorized} initial conditions \cite{Boris,Slava}. 
Another notable limiting case arises when the system-bath coupling takes a single-axis Ising form (e.g., $A_x=A_y=0$), for which, in addition, $[H_\text{int}^{\text{Ising}}, H_B]=0$; a purely dephasing, analytically solvable spin-bath model is then obtained which, despite its simplicity, has provided useful insight into the emergence of pointer states \cite{zurek1982environment} and the role of intra-bath entanglement \cite{Dawson} (see also Appendix\,\ref{appendix:zurek}).

Even in their most general form involving non-collective and anisotropic couplings, central-spin model Hamiltonians remain {\em formally} exactly solvable in the sense of admitting an exact algebraic diagonalization via Bethe Ansatz techniques, as long as appropriate equations are obeyed \cite{Gaudin1976,Gerardo} and factorized initial conditions are assumed. In practice, however, extracting the desired output dynamics from the Bethe Ansatz solution remains very hard, and various (e.g., mean-field) approximations and (or) numerical techniques must be employed \cite{Anders,Uhrig}.  Although the disorder-free case of uniform couplings we consider is highly idealized, it provides a useful playground for showing how observable-based MR may be carried out by using the general approach of Sec.\,\ref{SymmRed}. Having achieved that, in Sec.\,\ref{sec:examples_central_spin_reachable} we further demonstrate how additional MR reduction can arise, in the sense of Sec.\,\ref{sec:reachable}, depending on the choice of the (reduced) initial condition.

\subsubsection{Non-dissipative central spin with non-dynamical bath}
\label{section:basecase}

Consider a non-dissipative setting where the dynamics of the bath spins may be ignored on the time scales of interest -- a so-called ``non-dynamical'' bath, as it is can be the case, for instance, for nuclear spin baths \cite{Slava}. That is, we set $H_B\equiv 0$, in addition to ${\cal D}_{L_B}\equiv 0$, in Eq.\,\eqref{CSL}.

Thanks to the assumed collective nature of the couplings $A_u$, the permutation group $\mathcal{S}_{N_B}$ acting on the $N_B$ bath spins furnishes a (non-Abelian) unitary symmetry of $H_{SB}$:
\begin{equation*}
[G,H_{SB}]=0,\qquad\forall G\in\Gs_N \equiv \one_S\otimes\mathcal{S}_{N_B}.
\end{equation*}
Furthermore, since the observables of interest act non-trivially only on $S$, it also holds that
\begin{equation*}
[G,O_u]=0,\qquad \forall G\in\Gs_N, \quad \forall u=x,y,z.
\end{equation*}
Thus, the two conditions in Proposition \ref{strongcond} are satisfied and we can carry out an observable-MR procedure based on $\Gs_N$. To do so, we first need to find the unitary change of basis $U$ that decomposes $\Hc$ according to irreps of $\Gs_N$, and brings $\mathbb{C}\Gs_N'$ into its block-diagonal structure. The result follows from a well-known consequence of the Schur-Weyl duality \cite{Fulton}, which establishes a correspondence between irreps of the permutation group ($\Sc_{N_B}$ in our case) and those of the unitary group ({SU($2^{N_B})$} in our case). Let us first decompose ${\cal H}$ into invariant subspaces corresponding to a fixed angular momentum eigenvalue, $\Hc=\bigoplus_j\Hc_j,$ where $j(j+1)$ is the eigenvalue of the total angular momentum operator $J^2 \equiv J_x^2+J_y^2+J_z^2$ and $j=0,1,\ldots,\frac{N_B}{2}$ ($j=\frac{1}{2},\frac{3}{2},\ldots,\frac{N_B}{2}$) depending on whether $N_B$ is even (odd). Each invariant subspace $\Hc_j$ can then be factorized into two virtual subsystems, $\Hc_j= \Hc_{F,j}\otimes\Hc_{G,j}$ \cite{ViolaJPA,ZanardiVS,kempe2001theory}, namely,
\begin{align}
\nonumber
\Hc_{F,j}&\equiv\Span\{\ket{a}\otimes\ket{j,m}|a=0,1,\,m=-j,\ldots,j\}, \\
\Hc_{G,j}&\equiv\Span\{\ket{j;\alpha}|\alpha=1, \ldots,d_j\}.
\label{angbasis}
\end{align}
Here, $\ket{j,m; \alpha}$ are simultaneous eigenstates of $J^2$ and $J_z$, and $\alpha$ labels the multiplicity of the $j$-th angular momentum irrep. By Shur-Weyl duality, the latter also determines the dimension of the corresponding $\Sc_{N_B}$-irrep, and is given by \cite{kempe2001theory,Arenz_2014}
\begin{align}
\label{multiplicity}
d_j = \frac{(2j+1)N_B!}{ \big( \tfrac{N_B}{2} -j\big)! \big(\tfrac{N_B}{2}+j+1 \big)!    }. 
\end{align}
Thus, we may write $\Hc_S = \Span\{\ket{a}|\, a=0,1\}$ and 
$\Hc_B = \Span\{\ket{j,m;\alpha}|\, j,m,\alpha\}$, meaning that the desired unitary change of basis $U$ effects a transformation from the local to the total angular-momentum basis (the so-called ``Shur transform'' \cite{Fulton,EffcientQAlg}) on $\Hc_B$.

In the new basis, both $\mathbb{C}\Gs_N$ and $\mathbb{C}\Gs_N'$ have a block-diagonal structure. In particular, by using Eq.\,\eqref{commutant}, it follows that 
\begin{align}
\dim(\mathbb{C}\Gs_N')&=\sum_{j} \dim(\Hc_{F,j})^2 =4\sum_{j} (2j+1)^2 \notag \\  
&=\tfrac{2}{3} N(N+1)(N+2) .
\label{eq:DimScaling}
\end{align}
Thanks to Proposition \ref{strongcond} and Eq.\eqref{inclusions}, we can obtain a valid reduced quantum model by projecting onto $\mathbb{C}\Gs_N'$, with the guarantee that the resulting model will correctly reproduce the evolution of $\rho_S(t)$ at all times $t\geq0$, for arbitrary initial conditions $\rho_0$.
Numerical calculations of the operators spaces $\Ns^\perp$ and $\Os$ show that, for generic values of the coefficients, $\Os=\mathbb{C}\Gs_N'$ (see Fig.\,\ref{fig:central_spin_dim}), thus implying that the observable-based MR onto $\Cb\Gs_N'$ is generically {\em optimal}. The scaling of the dimension of the relevant operator space is thus reduced {\em from exponential to polynomial} (cubical) in system size $N$, after leveraging the permutation symmetry on the bath spins. 
As Fig.\,\ref{fig:central_spin_dim} also shows, further reduction may be achieved if only a minimal linear (not necessarily quantum) model is sought, in the sense of Theorem \ref{thm:linear_observable_model_reduction}, in which case projecting onto the non-observable subspace $\Ns^\perp$ suffices. The dimension of the latter subspace is upper-bounded by the dimension of the largest block in the Wedderburn decomposition of $\mathbb{C}\Gs_N'$ in Eq.\,\eqref{eq:DimScaling},
$$ \dim(\Ns^\perp) \leq \max_j \dim(\Bf(\Hc_{F,j})) = 4N^2,$$
which implies a {\em quadratic} scaling\footnote{As one may show, a {\em linear} scaling of both $\dim(\Ns^\perp)$ {\em and} $\dim \Os$ is found for the class of analytically solvable central-spin models with single-axis coupling; see Appendix \ref{appendix:zurek}.}.  
As we noted, projecting onto $\Ns^\perp$ may be fully appropriate for the purposes of simulating the dynamics on a classical computer and numerically computing the output quantities of interest. 

\begin{figure}[t]
\centering
\includegraphics[width=0.95\columnwidth]{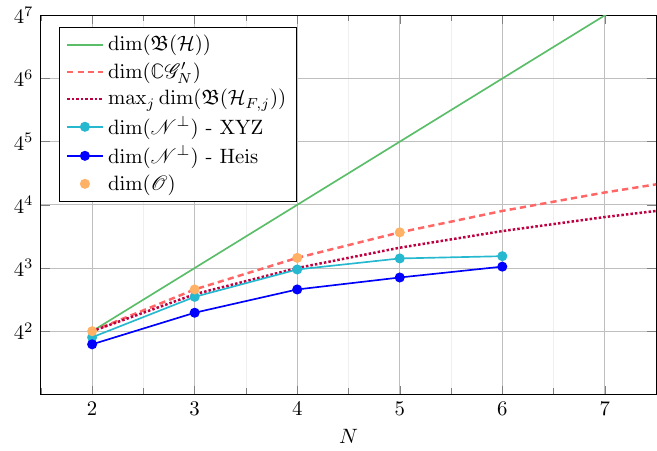}
\vspace*{-2mm}
\caption{\normalfont Dimensions of the observable space $\Ns^\perp$ and output algebra $\Os$, computed numerically for a non-dissipative central-spin model with collective couplings, Eq.\,\eqref{Hint}, compared to the dimensions of the full operator space $\Bf(\Hc)$ and group algebra $\Cb\Gs_N'$. We consider both a generic XYZ form for $H_{\text{int}}$, with arbitrarily chosen coefficients $A_x, A_y, A_z$, and $H_\text{int}=H_{\text{int}}^{\text{Heis}}$, with $A_x = A_y = A_z$. We set  $\omega_1=1$ (arbitrary units), and $\eta =H_B=0$ in both cases. Note that $\dim(\Ns^\perp)$ is smaller in the Heisenberg case, whereas $\dim(\Os)$ is identical for both models. Thus, while the extra symmetry in $H_{\text{int}}^{\text{Heis}}$ allows for a more efficient {\em linear} MR, the complexity is the same in terms of {\em quantum} MR.}
\label{fig:central_spin_dim}
\end{figure}

In order to explicitly carry out the observable-based quantum MR procedure presented in Sec.\,\ref{SymmRed}, the starting point is to define $\Hc_{\check{B},j}\equiv\Span\{\ket{j,m}$, with $m=-j,\dots,j\}$, so that $\Hc_S\otimes\Hc_{\check{B},j}=\Hc_{F,j}$. We then proceed to derive the reduced state, Hamiltonian, and observables accordingly. While the details are included in Appendix \ref{appendix:bipartition}, the final result may be summarized as follows. In the virtual bipartition above, let $W_j^\dag$ be the (non-square) isometry that maps the spin-$j$ invariant subspace $\Hc_j=\Hc_{F,j}\otimes\Hc_{G,j}$ onto $\Hc$. Then, the reduced state on $\Dc(\check{\Hc})$ can be evaluated as
\begin{equation}
\label{RedStateCentral}
\check{\rho}= \bigoplus_{j} \tr_{{\Hc}_{G,j}} \!\big( W_j\rho W_j^\dag \big)  \equiv \bigoplus_j \rho_{j} .
\end{equation}
Correspondingly, the reduced Hamiltonian is obtained by replacing the physical, collective spin-$1/2$ bath operators in 
Eq.\,\eqref{Hint} with new bath operators $\check{J}_{u,j}\in\Bf(\Hc_{\check{B},j})$, which are SU$(2j+1)$ generators, described by generalized Gell-Mann matrices \cite{Bertlmann} (see also Appendix\,\ref{appendix:bipartition}). That is, we have $\check{H}_{SB}\equiv \bigoplus_{j}\check{H}_{j}$, where the $j$-th block Hamiltonian 
\begin{eqnarray}
\label{RedObHamCentral}
   \check{H}_j &\equiv & \check{H}_{S,j} +\check{H}_{\text{int}} 
=  \frac{1}{2} \Big( \omega_1 \sigma_z + \eta \sigma_x \Big) \otimes\one_{\check{B},j}   \\
&+ &\frac{1}{2} \Big( A_x\sigma_x\otimes \check{J}_{x,j} +  A_y\sigma_y\otimes \check{J}_{y,j} + A_z\sigma_z\otimes \check{J}_{z,j} \Big), \nonumber
\end{eqnarray}
Similarly, the observables of interest are reduced according to the following:
\begin{equation}
\check{O}_u = 
\bigoplus_{j} \check{O}_{u,j},
\quad \check{O}_{u,j}\equiv \sigma_{u}\otimes \one_{\check{B},j}, 
\quad u\in \{x,y,z\}.
\label{modal1}
\end{equation}

Thanks to the fact that the reduced state and Hamiltonian share the same block-diagonal structure, the reduced model based on the symmetry $\Gs_N$ alone thus finally reads
\begin{align}
    \label{BaseRedModelCentral}
    \begin{cases}
        \dot{\check{\rho}}_j(t)&=-i[\check{H}_{j},\check{\rho}_{j}(t)] ,\\
        \check{\rho}_j(0)&=\rho_{0j} , 
    \end{cases} 
\end{align}
with the reduced state on the central spin given by
\begin{equation}
\rho_S(t) = \sum_j \tr_{\Hc_{\check{B},j}}[\check{\rho}_j(t)], \quad \forall t \geq 0.
\label{eq:reduced_output_map_central}
\end{equation} 
Since {\em no} assumption is made on the initial system-bath state $\rho(0)=\rho_0$ for Eq.\,\eqref{RedStateCentral} to hold, we stress that the above reduced model correctly reproduces the evolution of the desired observables even when starting from a general, {\em non-factorized} initial condition, without the need of specialized tools such as a $B^+$ decomposition \cite{PhysRevA.100.042120}. In situations where the initial joint state is factorized, say, $\rho(0)\equiv \rho_S(0)\otimes \rho_B$, we note that the map $\rho_S(0) \mapsto \rho_S(t)$ is not divisible in general, hence the reduced system dynamics is not Markovian.

Another noteworthy aspect of the reduced model in Eq.\,\eqref{BaseRedModelCentral} is the fact that each block $\check{\rho}_j$ of the reduced state evolves independently of the others. Accordingly, the expectation value of an observable of interest may be computed as 
\begin{equation}
\langle \sigma_u^{(1)}(t) \rangle = \sum_j\tr[ \check{O}_{u,j}\,\check{\rho}_j(t)] .
\label{modal2}
\end{equation}
This property can be useful both to gain qualitative insight into time evolution and to simplify the simulation of the model, as it is possible to analyze and simulate each block {\em in parallel}, and then sum the different contributions to obtain the full trajectory; see also Remark \ref{remark:computation_Advantage} in Sec.\,\,\ref{sec:examples_central_spin_dissipative}.

As anticipated before, we will now add more terms to this basic Hamiltonian model, while preserving $\Gs_N$, in the most general case, as an ODS.

\begin{remark}
To the best of our knowledge, the existence of an efficient classical algorithm for constructing the required Shur transform is yet to be determined (interestingly, an efficient quantum algorithm exists \cite{EffcientQAlg}). Therefore, although we managed to reduce the scaling of the operator-space dimension to only cubical in the number of degrees of freedom through symmetry considerations, the computational cost needed to arrive at this reduction should be accounted for separately.  
\end{remark}

\subsubsection{Effect of unitary bath dynamics}
\label{sec:examples_intrabath couplings}

Understanding the way in which intra-bath interactions alter the decoherence behavior of the central spin is both of fundamental interest and directly relevant to qubit implementations, especially in the solid state. Here, we show how the MR we carried out for the $H_B=0$ case can be used as a starting point for such a study in settings of physical relevance, where we let $H_B\equiv H_{B,0}+H_{B, \text{int}}$ while still keeping ${\cal D}_{L_B}\equiv 0$, and the bare bath Hamiltonian is also taken to have the following collective form:
\begin{equation}
\label{hbbare}
H_{B,0}=  \frac{1}{2} \Big ( \overline{\omega} J_z + \mu J_x \Big) , \quad \overline{\omega}, \mu \in {\mathbb R}.
\end{equation}

\begin{figure*}[t]
\includegraphics[width=0.99\textwidth]{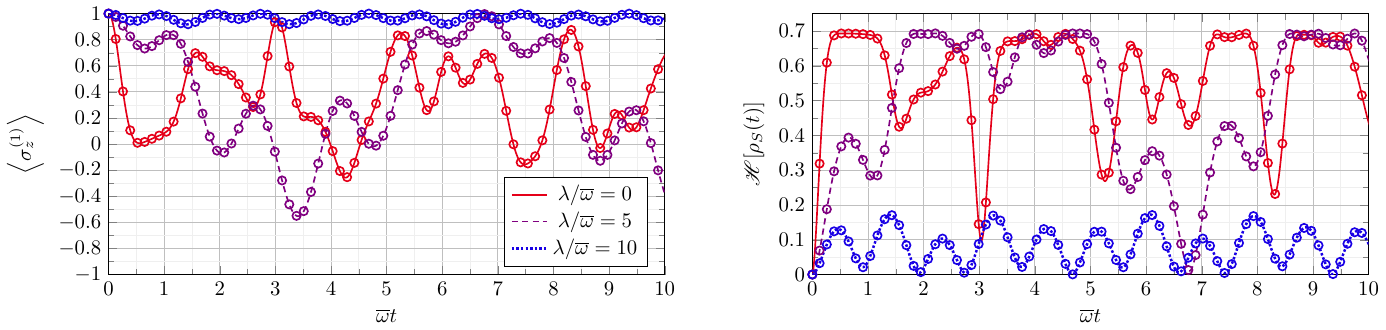}
\caption{\normalfont
Comparison between full vs.\,reduced dynamics for a central-spin model that couples via a general XYZ collective Hamiltonian, Eq.\,\eqref{Hint}, to a bath with Ising intra-bath interactions, Eqs.\,\eqref{hbbare} and \eqref{IsingInt}, for $N=7$. Left: Time-dependent central-spin polarization, $\langle \sigma_z^{(1)}(t) \rangle$. Right: Time-dependent central-spin entropy, $\Hs[\rho_S(t)]$. In both cases, with $\hbar =k_B=1$, time is in units of inverse temperature, and the parameters are: $\omega_1=0.83$, $\eta = 0.01$, $\overline{\omega}=1$, $\mu=0.02$, $A_x=4$, $A_y=0.1$, $A_z=-0.5$. Three different bath-interaction strengths $\lambda$ are used in both the full (dots) and the reduced (solid and dashed lines) models. The initial joint state is $\rho_0=\ket{0}\bra{0}\otimes \rho_B$, with $\rho_B$ being the thermal state at inverse temperature $\beta=50$.} 
\label{fig:Ising}
\end{figure*}

$\bullet$ {\em Rotationally invariant bath interactions.--} First, our symmetry analysis makes it clear that there are non-trivial interaction Hamiltonians $H_{B, \text{int}}$, which have {\em no} effect on the central spin dynamics, $\rho_S(t)$, or on central-spin expectation values, $\langle O_u(t)\rangle$. This is the case if $H_{B, \text{int}} \in \mathbb{C}\Gs_N$, in which case
$$[H_{SB},H_{B, \text{int}}]=[O_u,H_{B, \text{int}}]=0,\quad \forall O_u. $$ 
A prominent example is given by a bath whose spins interact via an Heisenberg Hamiltonian:
\[H_{B, \text{int}}^{\text{Heis}} = \sum_{2\leq i<k} B_{i k} \, \vec{\sigma}^{(i)}\cdot  \vec{\sigma}^{(k)}, \quad B_{i k}  \in {\mathbb R}. \]
It is well-known that the Heisenberg coupling is related to the exchange interaction, since $\vec{\sigma}^{(i)}\cdot  \vec{\sigma}^{(k)} = 2S_{ik}-\one,$ where $S_{ik}$ is the unitary operator that swaps the states of the $i$-th and $k$-th bath spins. Thus, $H_B^{\text{Heis}}\in \mathbb{C}\Gs_N$ and the above commutation property follows. Explicitly, it is immediate to see that 
\begin{eqnarray*}
   \langle \sigma^{(1)}_u(t)\rangle &= &
   \tr(\sigma^{(1)}_u e^{-i(H_{SB}+H_{B,\text{int}}^{\text{Heis}}) t}\rho_0 e^{i(H_{SB}+ H_{B,\text{int}}^{\text{Heis}} t)} \\
 & =&\tr(\sigma^{(1)}_u e^{-iH_{B,\text{int}}^{\text{Heis}}  t}e^{-iH_{SB}t} \rho_0 e^{iH_{SB}t}e^{iH_{B,\text{int}}^{\text{Heis}} t})\\
& =&\tr(\sigma^{(1)}_u e^{-iH_{SB}t} \rho_0e^{iH_{SB}t}), \quad \forall t.
\end{eqnarray*}
However, since the permutation group is non-Abelian, we have $[\Gs_N, H^{\text{Heis}}_B] \ne 0$, except in the uniform-coupling limit where $B_{ik}\equiv B$. Therefore, $\Gs_N$ remains an ODS despite {\em not} being a symmetry, and we can use its commutant algebra to carry out MR. Since 
$\check{H}_{B,\text{int}}=0$ for ${H}_{B,\text{int}}\in \mathbb{C}\Gs_N$,
the resulting reduced model is then still given by Eqs.\,\eqref{BaseRedModelCentral}-\eqref{eq:reduced_output_map_central}, upon redefining $\check{H}_{j}$: if we now denote with $\check{H}_{0,j}$ the Hamiltonian in Eq.\,\eqref{RedObHamCentral}, we have 
\begin{eqnarray}
\check{H}_{j}^{\text{Heis}} & \equiv &\check{H}_{0,j} + \check{H}_{B,0;j} \notag \\
& = & \check{H}_{0,j} + \one_{S} \otimes \frac{1}{2} \Big(  \overline{\omega} \check{J}_{z,j}  + \mu \check{J}_{x,j}  \Big). 
\label{redHj}
\end{eqnarray}

$\bullet$ {\em Permutationally invariant bath interactions.--} Another relevant scenario arises in the situation where $H_{B, {\text{int}}}\in \mathbb{C}\Gs_N'$. As a representative example, consider a uniform Ising interaction, of the form 
\begin{equation}
\label{IsingInt}
H_{B, \text{int}}^{\text{Ising}}=\frac{\lambda}{4}\sum_{2\leq i<k}\sigma_x^{(i)}\sigma_x^{(k)}, \quad \lambda \in {\mathbb R},
\end{equation}
where $\lambda$ is an overall strength parameter and the form of the coupling is inspired by a spin-bath model employed in Ref.\,\cite{Tessieri_2003} as a simplified model to study the effect of anharmonic phonon-phonon interaction on the decoherence of an impurity spin. As long as the couplings are uniform, $H_{B,\text{int}}^{\text{Ising}}$ preserves the permutation invariance, as one may see explicitly by rewriting 
\begin{equation}
H_{B,\text{int}}^{\text{Ising}}=\frac{\lambda}{4}\Big( 2 J_x^2- \frac{N_B}{2}\one_{2^{N}}\Big) \in {\mathbb{C}}\Gs_N' . 
\label{IsingInt2}
\end{equation}
Since both terms are manifestly permutation-invariant, we may write down the intra-bath Hamiltonian that enters the reduced model compactly:
\begin{equation*}
\check{H}_{B,\text{int}}^{\text{Ising}}=\bigoplus_{j} \frac{\lambda}{4}\Big(2\check{J}_{x,j}^2- \frac{N_B}{2}\one_{\Hc_{F,j}}\Big)
\equiv \bigoplus_{j} \check{H}_{B, \text{int}; j}^{\text{Ising}} .
\end{equation*} 
Accordingly, the reduced model is still governed by Eqs.\,\eqref{BaseRedModelCentral}-\eqref{eq:reduced_output_map_central}, where now 
$\check{H}_{j}^{\text{Ising}} =\check{H}_{j}^{\text{Heis}} 
+\check{H}_{B,j}^{\text{Ising}}, $
with $\check{H}_{j}^{\text{Heis}}$ explicitly given in Eq.\,\eqref{redHj}. 

\smallskip

\begin{figure*}[th]
    \centering
     \begin{subfigure}[c]{0.47\textwidth}
     \centering
         \includegraphics[height=\textwidth]{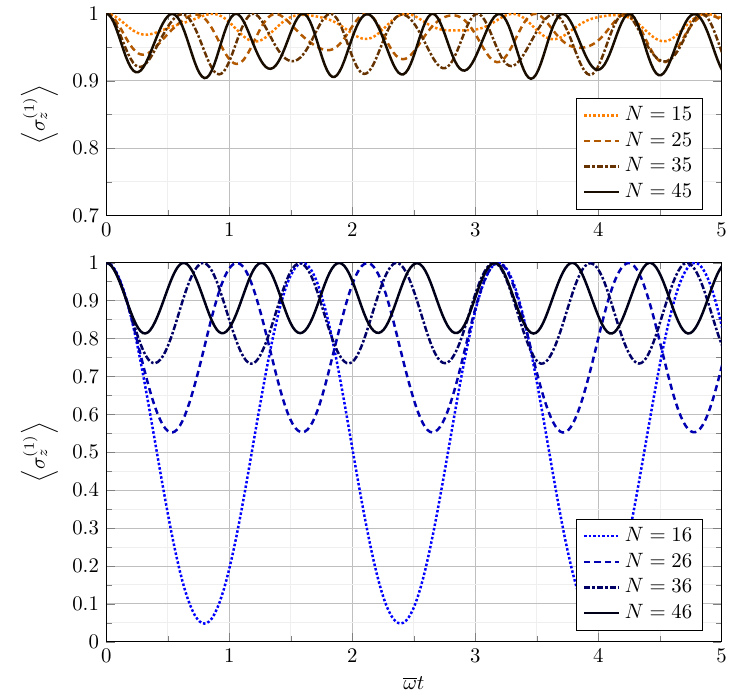}
     \end{subfigure}
     \hspace{20pt}
     \begin{subfigure}[c]{0.47\textwidth}
     \centering
         \includegraphics[height=\textwidth]{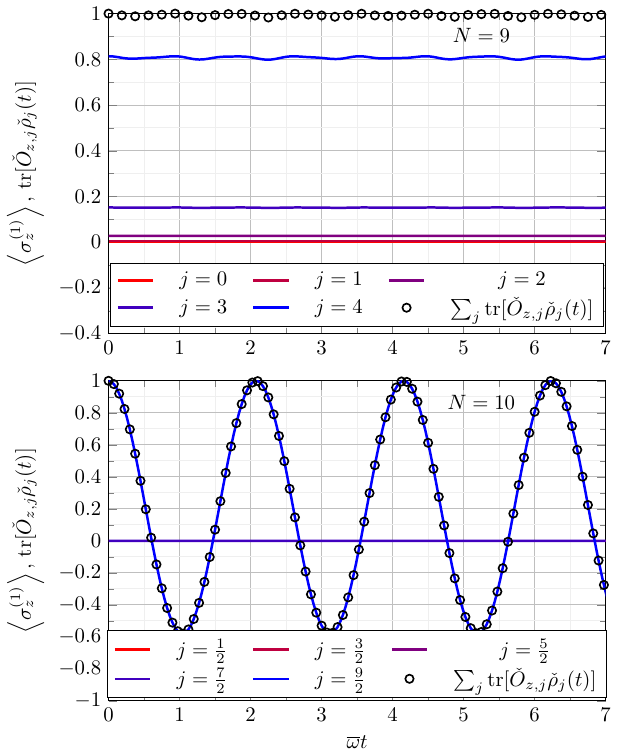}
     \end{subfigure}
    \caption{ \normalfont
    Left panels:
    Time-dependent central spin polarization, $\langle \sigma_z^{(1)}(t) \rangle$ for a central-spin model that couples via an Ising collective Hamiltonian, Eq.\,\eqref{Hint} with $A_y=A_z=0$, to a bath with Ising intra-bath interactions, Eqs.\,\eqref{hbbare} and \eqref{IsingInt}. Increasing even (bottom) vs. odd (top) are considered, for fixed intra-bath coupling strength $\lambda/\overline{\omega} =20$, $\overline{\omega}=0.75$ and $\mu=\eta=0.01$. The remaining parameters are as in Fig.\,\ref{fig:Ising}.
    Right panels: Comparison between the total central-spin polarization, $\langle \sigma_z^{(1)}(t)\rangle$, and the contribution given by each fixed bath-angular momentum block for $N=9,10$. 
    }    
    \label{fig:tessieri}
\end{figure*}

To quantitatively validate our analysis, we compare numerical simulations of dynamical quantities of interest for the full vs. the reduced model in a setting with a general XYZ collective system-bath interaction Hamiltonian and Ising intra-bath interaction Hamiltonian. While, as we mentioned, our MR procedure is applicable to arbitrary system-bath initial conditions, the choice of an initial product state is natural in this setting, namely, $\rho_0 = \ket{0}\bra{0}\otimes \rho_B$, with $\one_S\otimes\rho_B= e^{-\beta H_B} /\tr ( e^{-\beta H_B})$ 
being the thermal state at inverse temperature $\beta \equiv 1/k_B T$. Importantly, this initial state can easily be written in its reduced form, thanks to the fact that $H_{B, {\text{int}}}\in \mathbb{C}\Gs_N'$, Eq.\,\eqref{IsingInt2}; namely, we have 
\begin{equation}
\Rc (e^{-\beta H_B} ) = \bigoplus_j   e^{ -\beta ( \check{H}_{B,0;j}  + \check{H}_{B, \text{int};j}^{\text{Ising}} ) } \,  d_j, 
\label{RedIC}
\end{equation}
with $\check{H}_{B,0;j}$ defined in Eq.\,\eqref{redHj} and the multiplicity factors $d_j$ given in Eq.\,\eqref{multiplicity}. This property allows us to {\em directly} implement the reduced version of the initial state, $\check{\rho}_0$, bypassing the computational effort of reducing the initial state $\rho_0$. Thanks to that, it becomes possible to simulate the (collective) central spin model for a relatively large number of spins, up to $N_B\approx 45$ in only a few tens of seconds on a laptop. 

\begin{figure*}[ht]
    \centering
    \includegraphics[width=0.99\textwidth]{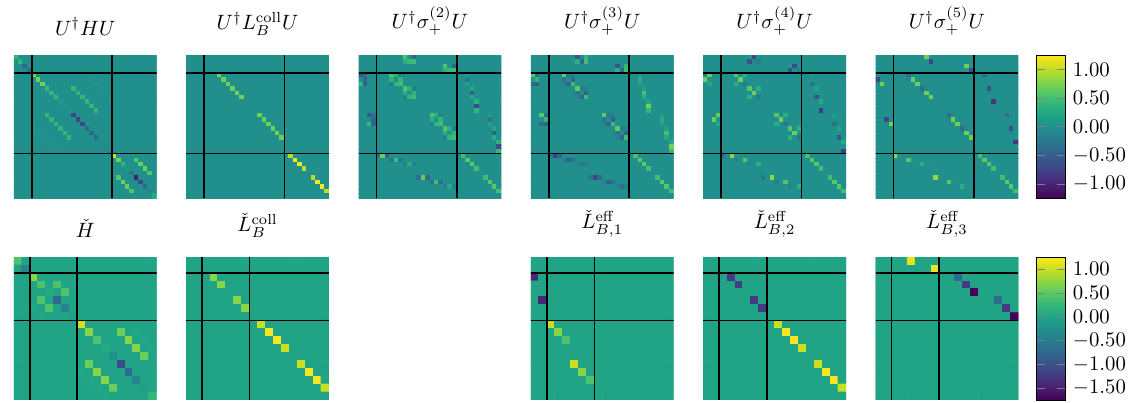}
    \vspace*{2mm}
    \caption{\normalfont Graphical representation of the Hamiltonian and noise operators in  both the original (top row) and reduced (bottom row) models for a central-spin system with Markovian dissipation on $N_B=4$ bath spins. The first and second columns corresponds to the Hamiltonian and the collective pumping dynamics, with $\Lambda=0.5$ in Eq.\,\eqref{RedCollective}, whereas the remaining columns corresponds to local dissipation, with strength $\delta=1$ and effective Lindblad operators given by Eq.\,\eqref{LBeff}. In both cases, we set $A_x=0.6$, $A_y=0.55$, $A_z=0.65$, $\omega_1=0.8$, $H_B=0$. 
    The operators of the original models belong to $\Cb^{32\times 32}$ and are represented in the base that block-diagonalizes the output algebra, while the operators for the reduced models belong to $\Cb^{18\times 18}$ and are represented in the standard basis. The black lines highlight the block structure of the output algebra.}
    \label{fig:central_spin_blocks}
\end{figure*}

Representative results of these simulations are shown in Fig.\,\ref{fig:Ising}, for both the time-dependent expectation value of the central-spin $z$-polarization, $\langle\sigma_z^{(1)} (t)\rangle$ (left column), and the time-dependent central-spin von Neumann entropy, $\Hs[\rho_S(t)]$ (right column). These results demonstrate how the reduced model quantitatively reproduces the full dynamics for the target observable and possibly non-linear functional thereof, as expected. Therefore, the reduced model may be reliably used to explore features related, for example, to the non-trivial interplay between the bath free dynamics and the intra-bath couplings with $H_{\text{int}}$. 

As a concrete illustrative setting, in order to make some contact with the problem investigated in \cite{Tessieri_2003}, we choose similar parameter values\footnote{In the model of \cite{Tessieri_2003}, a non-uniform on-site Zeeman splitting is assumed for the bath spins, $\sum_{i=2}^{N} \omega_i \sigma_z^{(i)}$, rather than our permutationally-invariant $\overline{\omega} J_z$. In our simulations, we have chosen the value of $\overline{\omega}$ to match the expectation value of the frequencies $\omega_i$, which are distributed according to a Debye probability density.} for $\overline{\omega}, \mu$ in $H_{B,0}$ and $\omega_1,\eta$ in $H_S$, and consider an Ising system-bath coupling, i.e., $A_y=A_z=0$ in Eq.\,\eqref{Hint}. One may then assess the extent to which ``self-decoupling'' and decoherence suppression occur for strong intra-bath interactions, as reported by \cite{Tessieri_2003} based on numerical results with up to $N_B=14$ spins in the less symmetric case with randomly distributed $\omega_i$.  Representative results are shown in Fig.\,\ref{fig:tessieri} (left), where the expectation values $\langle\sigma_z^{(1)}\rangle$ for odd and even number of spins are separately plotted against time. Our results suggest that a self-decoupling effect still occurs for $\lambda \gg \overline{\omega}$, as manifested by the fact that the spin polarization approximately oscillates periodically or freezes out for even or odd $N$, respectively. For larger number of spins, however, the difference between odd and even $N$ tends to diminish. A qualitative explanation of this behavior may be obtained by analyzing the contribution that each block $\check{\rho}_j$ gives to the target expectation value, i.e., by using Eq.\,\eqref{modal2}, for $u=z$. As illustrated in Fig.\,\ref{fig:tessieri} (right), it turns out that, for sufficiently large values of the bath frequency $\overline{\omega}$, the probability distribution between blocks, tr$[\check{\rho}_j]$ is heavily skewed towards the blocks associated to higher $j$. Thus, to obtain a quantitative understanding of the evolution of the central spin in the case depicted in Fig.\,\ref{fig:tessieri}, it suffices to consider the blocks associated to the highest or second-highest value of $j$, depending on whether $N$ is even or odd, respectively. As $N\gtrsim 30,$ only the highest weight block is in fact contributing.

\subsubsection{Effect of Markovian dissipation on bath spins}
\label{sec:examples_central_spin_dissipative}

We next illustrate how the approach may be extended to a central spin model as in Eq.\,\eqref{CSL}, where both coherent and dissipative dynamics are present, i.e., ${\cal D}_{L_B} \ne 0$. Specifically, we still assume a collective system-bath Hamiltonian of the form given in Eq.\,\eqref{HSB}-\eqref{Hint} and, to isolate the effect of the dissipative bath interactions, we return to the case of a non-dynamical bath, $H_B\equiv 0.$ We consider two types of Markovian dissipation, leading to models that can be reduced to the \emph{same output algebra}; however, in the first case the reduction will emerge from a strong symmetry of the Lindbladian, whereas in the second case it will correspond to a (strictly) weak one. 

\smallskip

$\bullet$ {\em Collective pumping.--} A dissipative pumping mechanics that acts collectively on the bath spins may be modeled by introducing a Lindblad operator
\begin{equation*}
L_B^{\text{coll}} \equiv \Lambda J_+, \quad \Lambda >0, \qquad \Dc_{L_B} \equiv \Dc_{L_B^\text{coll}}, 
\end{equation*}
with $J_+= J_x + iJ_y \in\Bf(\Hc)$ being the raising operator associated to the total-spin bath angular momentum. In this case, one may verify that the collective property of both $H$ and ${L_B^{\text{coll}}}$ makes $\Gs_N = \mathcal{S}_{N_B}$ a {\em strong symmetry} for the overall dynamics. Therefore, the result in Eq.\,\eqref{reducedObHamLind} applies and the reduced Lindblad operator is found to be 
$$\check{L}_B^{\text{coll}}=\Lambda \bigoplus_{j} \check{J}_{+,j},$$
with $\check{J}_{\pm,j} \equiv \check{J}_{x,j}\pm i\check{J}_{y,j} \in\Bf(\Hc_{\check{B},j})$ being the raising operator constructed out of SU$(2j+1)$ spin angular momentum generators (see also Appendix \ref{appendix:bipartition}). Comparing to Eq.\,\eqref{BaseRedModelCentral}, in each $j$-block the reduced dynamic has now an extra term:
\begin{equation}
\label{RedCollective}
\dot{\check{\rho}}_j \equiv \check{\Lc_j}(\check{\rho_j}) = 
\Lc_j^{\text{coll}}(\check{\rho}_j)=-i[\check{H}_{j},\rho_{j}]+\Lambda \Dc_{\check{J}_{+,j}}(\rho_j),
\end{equation}
where $\Lc_j^{\text{coll}}$ can be interpreted as the reduced Lindbladian in the $j$-th subspace, and $\check{H}_j$ has the expression in Eq.\,\eqref{RedObHamCentral}. A graphical representation of the Lindbladian operators' structure for both the original and reduced model is given in Fig.\,\ref{fig:central_spin_blocks}.

\begin{figure*}[t]
    \hfill
    \begin{subfigure}[c]{0.49\textwidth}
    \centering
        \includegraphics[width=\textwidth]{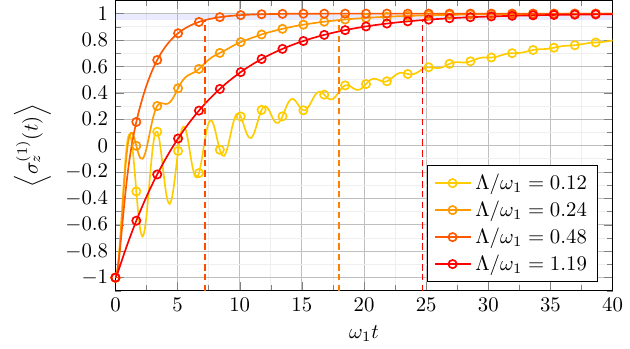}
    \end{subfigure}
    \hfill
    \begin{subfigure}[c]{0.46\textwidth}
    \centering
        \includegraphics[width=\textwidth]{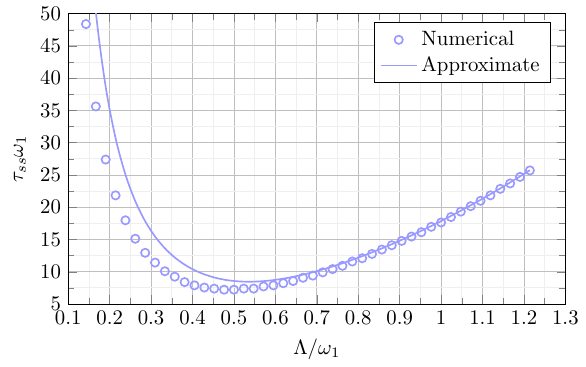}
    \end{subfigure}
    \hfill
    \vspace*{-3mm}
\caption{ \normalfont
Left: Time-dependent central spin polarization, $\langle{\sigma_z^{(1)}(t)}\rangle$, under collective bath dissipation of variable strength $\Lambda$, in both the full (dots) and reduced (solid curve) models, for $N=6$. As in \cite{Coish}, we choose an XXZ system-bath coupling Hamiltonian, with $A_x=A_y=4,A_z=2$, $\omega_1=4.2$, $\eta=0$, $H_B=0$, while the bath is at zero temperature, hence $\rho_0=\ketbra{1}{1}\otimes\ketbra{00\dots0}{00\dots0}$. 
The vertical dashed lines indicate the time-to-SS, $\tau_{\text{ss}}$, for the considered values of $\Lambda$, with the shaded blue region showing a $5\%$ deviation from the asymptotic value. Right: Time-to-SS vs. dissipation strength. The dots represent the values obtained from the numerical solution of the reduced model, while the solid curve represents the approximated estimate obtained in Sec.\,\ref{sec:examples_central_spin_reachable}.}
\label{fig:collectivepumping}
\end{figure*}

To validate our procedure, we again compare the numerical solution of the full vs. reduced models. As a representative example, in Fig.\,\ref{fig:collectivepumping} (left) we compare the central spin's polarization for the full (dotted line) and reduced model (solid line), again resulting in exact agreement, as expected. Interestingly, as the strength $\Lambda$ of the dissipation increases, we observe a transition from a regime where the trajectory reaches equilibrium slowly, with oscillation, to one where the equilibrium is reached more rapidly, and with no oscillation. 
Furthermore, Fig.\,\ref{fig:collectivepumping} (right) shows a {\em non-monotonic} behavior of the convergence time, which we quantity in terms of the time taken for $\langle{\sigma_z^{(1)}(t)}\rangle$ to remain confined within $5\%$ of its asymptotic value. Starting from its infinite value at $\Lambda=0$ (when no dissipation is present), such a time-to-SS, say $\tau_{\text{ss}}$, decreases and reaches an {\em optimal} value, then diverges to infinity afterwards, see Fig.\,\ref{fig:collectivepumping} (right). Similar non-monotonicity has been reported in recent works \cite{Coish,onizhuk2023understanding}, {with the ``optimal'' value being interpreted in terms of a synchronization effect between the dissipation rate and the characteristic time scales of the coherent dynamics, and the increase for strong-dissipation bearing similarities with a quantum-Zeno limit.}
In particular, \cite{Coish} studies a model which is, in a way, dual to the one considered here: the pumping noise acts only on the central spin, and the $J_z$-observable is computed for the bath spins. Treating the bath as a collective spin, it is natural, albeit still interesting, to recover the non-monotonicity upon swapping the roles of the central and the bath spins.

\smallskip

$\bullet$ {\em Local pumping.--} In a scenario where each bath spin undergoes local pumping with uniform strength, the dissipator may be described by a set of $N_B$ Lindblad operators
\begin{eqnarray*}
L_{B,k}^{\text{loc}} \equiv \delta \sigma_+^{(k)}, \quad \delta >0, 
\qquad \sum_{L_B} \Dc_{L_B}= \delta^2 \sum_{k=2}^N\Dc_{\sigma_+^{(k)}}.    
\end{eqnarray*}
Because $\sigma_+^{(k)}$ breaks permutation symmetry, $\Gs_N$ is no longer a strong symmetry; however, the form of the dissipator makes it clear that it remains a {\em weak} symmetry of the dynamics. This implies that the output algebra onto which we reduce the model remains the same, that is,  $\Os=\Cb\Gs'_N$. While the Hamiltonian retains the block-diagonal structure of   Eq.\,\eqref{RedObHamCentral}, the reduced noise operators, say, $\{ \check{L}_{B,h}^{\text{eff}}\}$ are no longer block-diagonal in the basis defined by $U$, hence Eq.\,\eqref{reducedObHamLind} (right) does not hold in this case. The resulting structure is illustrated in Fig.\,\ref{fig:central_spin_blocks}, top row. Albeit no closed formula is available to compute the reduced operators when off-diagonal terms are present, the reduced noise operators can still be computed numerically by first computing the super-operators $\Rc$ and $\Jc$ as in Eq.\,\eqref{eqn:reduction} and \eqref{eqn:injection}, and then by determining effective noise operators $\check{L}^{\text{eff}}_{B,h}\in\Bf(\check{\Hc})$, such that 
\begin{equation}
\sum_{h=1}^{\check{N}} \Dc_{\check{L}^{\text{eff}}_{B,h}} (\rho) = \sum_{k=2}^N \Rc \,\Dc_{L_{B,k}^{\text{loc}}} \Jc(\rho).
\label{LBeff}
\end{equation} 
The reduced noise operators $\check{L}_h$ obtained in this manner are shown in the bottom row of Fig.\,\ref{fig:central_spin_blocks}. This makes it clear that the effects induced by local pumping noise, as opposed to collective one, is to introduce off-diagonal terms that create cross-talk between adjacent $j$-blocks. Notice that, in general, when reducing a single dissipative term $\Dc_L$, multiple reduced operators $\sum_k \Dc_{\check{L}_k}$ may arise in principle. This is due to the fact that the composition of CP maps may, in general, increase the Kraus rank. In this particular case, however, we numerically verified (up to $N=7$) that $\check{N}=3$  operators are sufficient to define the reduced dissipative term, independently of $N$. The matrix representations of such noise operators differ from zero only in the main, the first lower, and the first higher diagonal blocks, respectively.

\begin{remark}
\label{remark:computation_Advantage} 
From the above reduction analysis and the block structure of the resulting generators, it is clear that having a strong symmetry leads to a computational advantage with respect to the weak-symmetry case. Since, in the strong-symmetric case, both the reduced Hamiltonian and noise operators are block-diagonal, {\em each block of $\Cb\Gs'$ is an invariant subalgebra} in itself. Thus, one can project the initial state on the algebra and simulate {\em independently} each block's dynamics, allowing for sequential or parallel simulations. Practically, the capabilities of a given hardware only limit the size of the largest simulatable block, rather than that of the full algebra. In contrast, in the weak-symmetric case, the dynamics of different block are linked, thus all the blocks of the reduced state need to be computed at the same time. With reference to  Fig.\,\ref{fig:central_spin_dim}, for the collective central-spin model the dependence on $N$ of the dimension of algebra is cubic, while it is quadratic for the largest block size.
\end{remark}

\subsubsection{Reachable analysis: W-state initialization}
\label{sec:examples_central_spin_reachable}

In the analysis carried out so far, we have only exploited the possibility to reduce the model of interest based on specified observables. In this section, we exemplify the potential of the reachable-based MR presented in Sec.\,\ref{sec:reachable}. 

For the sake of illustrating the procedure through an analytical example, let us first assume that the dynamics exhibit a strong symmetry: specifically, we consider a central-spin model subject to collective dissipation as in Sec.\,\ref{sec:examples_central_spin_dissipative}, with 
$$H_S=\tfrac{1}{2} \omega_1 \sigma_z^{(1)}, \,\,\, H_B=0, \,\,\, H_{\text{int}}= H_\text{XXZ}, \,\,\, L_B^{\text{coll}} = \Lambda J_+.$$
Suppose we are interested in studying the evolution when the initial state is $\rho_0 = \ketbra{0}{0} \otimes \ketbra{W}{W}$, where $\ket{W}$ is the single-excitation, many-body entangled W-state on the bath \cite{durW}, 
$$\ket{W} \equiv \frac{1}{\sqrt{N_B}}\big(\ket{10\dots0}+\ket{01\dots0}+\dots+\ket{00\dots1} \big).$$ 
The state $\ket{W}$ is manifestly permutation-invariant, and is an eigenvector of $J^2$ and $J_z$ with eigenvalues $j=N_B/2 \equiv j^*$ and $m=j^*-1 \equiv m^*$, respectively. Thus, in the basis given by Eq.\,\eqref{angbasis}, the initial condition reads  $\rho_0 = \ketbra{0}{0}\otimes\ketbra{j^*,m^*}{j^*,m^*}$, since the $N_B/2$-angular momentum irrep has multiplicity $1$. Thanks to the presence of a strong symmetry, the evolution of $\rho_0$ will remain confined in the block associated with $j^*$, $\Hc_{j^*} = \Hc_{F,j^*} \otimes \Hc_{G,j^*}$.
Accordingly, we can restrict our attention to the subspace $\Hc_{F,j^*}$ and to the evolution of $\check{\rho}_{j^*}(t)$, which is generated by the Lindblad master equation in Eq.\,\eqref{RedCollective}, subject to the initial condition $\check{\rho}_{0j^*}= \ketbra{0}{0}\otimes\ketbra{m^*}{m^*}$. This alreadys allow for a reduction of the dimension of the supporting algebra to a quadratic scaling, $\dim(\Bf(\Hc_{F,j^*}))=4N^2,$ in contrast to the cubic scaling of the output algebra dimension, Eq.\,\eqref{eq:DimScaling}. As we show next, further MR can be achieved, since the  evolution originating from $\rho_0$ does not explore the entire block $\Bf(\Hc_{F,j^*})$. 

Explicitly, under the above assumptions the reduced Hamiltonian reads 
\begin{align*}
\check{H}_{0,j^*}&=\frac{\omega_1}{2}\sigma_z\otimes\one_{\Hc_{\check{B},j^*}}+A\left(\sigma_-\otimes\check{J}_{+,j^*}+\sigma_+\otimes\check{J}_{-,j^*}\right)\\
&+\frac{A_z}{2}\sigma_z\otimes\check{J}_{z,j^*}, \quad A\equiv A_x=A_y, 
\end{align*}
while the reduced noise operator $\check{L}_B = \Lambda\check{J}_{+,j^*}$. For generic values $\omega_1\neq0$ and $\Lambda\neq0$, direct computation of the reachable space yields 
\begin{equation}
\begin{split}
    \Rs = \Span&\big\{ \ketbra{0}{0}\otimes\ketbra{j^*}{j^*}, \\   
    &\hspace*{2mm}\ketbra{0}{0}\otimes\ketbra{m^*}{m^*}, \ketbra{1}{1}\otimes\ketbra{j^*}{j^*}, \\
    &\hspace*{1.7mm} \ketbra{0}{1}\otimes\ketbra{m^*}{j^*}, \ketbra{1}{0}\otimes\ketbra{j^*}{m^*} \big\},
\end{split}
\label{reach}
\end{equation}
with $\dim(\Rs)=5$. One can also verify analytically that $\Rs$ is an algebra of dimension $5$ (irrespective of the bath's size), has only partial support in the block \(\Bf(\Hc_{F,j^*})\), and has a Wedderburn decomposition equal to alg$(\Rs)\simeq\Cb\oplus\Cb^{2\times 2}\oplus\zero_R$, where $\dim(\Hc_R)=2N-3$. Following the discussion given in Sec.\,\ref{sec:algebraic_model_reduction_1}, we can restrict our attention to the support of the algebra, so that the reduced model obtained after observable-based and reachable-based MR is going to be defined over $\Cb\oplus\Cb^{2\times 2}\subset\Cb^{3\times3}$. Since, in addition, we have that $\Rs=\alg(\Rs)$, we can take $\Ds=\Rs$ in the procedure described in Sec.\,\ref{sec:lindblad_reduction} and proceed to compute the maps $\Rc$ and $\Jc$ over the support of $\Ds$ to obtain the reduced model.

The reduced state on $\Ds$, $\check{\check{\rho}}(t) \equiv \Rc[\check{\rho}(t)]$ can then be expressed in the form
\[\check{\check{\rho}}(t) = p(t) \oplus \tau(t),  \]
where $p(t) \equiv \tr[\ketbra{0}{0}\otimes\ketbra{j^*}{j^*} \check{\rho}(t)]\in[0,1]$ is a scalar and 
$$ \tau(t) \equiv \sum_{k,l=0}^1\ketbra{k}{l} \tr\left[\ketbra{k}{l}\otimes\ketbra{j^*-f(k)}{j^*-f(l)} \check{\rho}(t)\right]$$
is a positive-semidefinite operator, with $f(k) \equiv \tfrac{1}{2}[1-(-1)^k]$ and $\check{\rho}(t) =\check{\rho}_{j^*}(t)$. Because $\check{\check{\rho}}(t)$ is a valid state, we have $p(t)+\tr[\tau(t)]=1$ for all $t$. Furthermore, the system evolves according to the following set of equations:
\begin{equation}
    \begin{cases}
        \dot{p}(t) = L_R \tau(t) L_R^\dag , \\
        \dot{\tau}(t) = -K\tau(t)-\tau(t)K^\dag,
    \end{cases}
    \label{eq:reduced_reachable_model}
\end{equation}
where $L_R=\begin{bmatrix}
    \Lambda\sqrt{N_B} & 0
\end{bmatrix}$ 
and the effective non-Hermitian Hamiltonian $K \equiv i\check{\check{H}} + L_R^\dag L_R/2$, with 
$$\check{\check{H}} = \frac{A\sqrt{N_B}}{2}\sigma_x + \Big[\frac{A_z}{4}(1-N_B) - \frac{\omega_1}{2}\Big]\sigma_z. $$ 
Although it need not be obvious that Eq.\,\eqref{eq:reduced_reachable_model} represents a Lindblad generator, one can verify that this is in fact the case by representing the state $\check{\check{\rho}}(t)$ as a matrix in $\Cb^{3\times3}$, and rewriting the generator in this form. Interestingly, Eq.\,\eqref{eq:reduced_reachable_model} represents a non-trivial dynamical connection between $p(t)$ and $\tau(t)$, effectively pumping probability towards $p(t)$ -- similar to what was observed in the previous section in the case of local noise and a weak symmetry.
Finally, the reduced central-spin observables of interest, $\check{\check{O}}_q = \Jc^\dag(\check{O}_q)$, take the form $$\check{\check{O}}_0 = 1\oplus\one_2,\quad \check{\check{O}}_x = \check{\check{O}}_y = \zero_3, \quad \check{\check{O}}_z = 1\oplus \sigma_z,$$ 
showing that, when starting from the assumed initial condition, coherences on the central spin can never be observed.

As a concrete application of the reduced model we just derived, we are in a position to derive an approximate description of the non-monotonic behavior observed for the time-to-steady state (SS) in Fig.\,\ref{fig:collectivepumping} (right). Since the initial condition $\ketbra{1}{1}\otimes\ketbra{00\dots0}{00\dots0}$ used in the simulation therein belongs to $\Rs$, the reduced model in Eq.\,\eqref{eq:reduced_reachable_model} applies. One can easily observe that such a model converges to the SS $\check{\check{\rho}} = 1\oplus \zero_2$, which is associated to the convergence of the 
full model to the pure state $\ketbra{0}{0}\otimes\ketbra{00\dots0}{00\dots0}$, as expected. The asymptotic rate of convergence to the SS can be characterized in terms of the spectral gap of the dynamical generator, that is, the eigenvalue of $-K\tau - \tau K^\dag$ with the smallest (strictly non-zero) absolute value, say, $\lambda_\text{dom} (\Lambda)$. By letting $\tau_{\text{ss}} \approx c / \lambda_{\text{dom}}(\Lambda)$, and computing numerically the eigenvalues as a function of $\Lambda$, we find $c\approx \ln(0.05)$ and $\lambda_{\text{dom}}(\Lambda) \approx - 2.0 \Lambda^{2}$. 
From Fig.\,\ref{fig:collectivepumping}, we see that the above estimate for $\tau_{\text{ss}}$ correctly captures the observed non-monotonic behavior, to first-order approximation. We stress that this approximation does not come from the MR procedure (which is exact), but only from the approximate computation of $\tau_{\text{ss}}$.

\begin{figure}[t]
\centering
\includegraphics[width=\columnwidth]{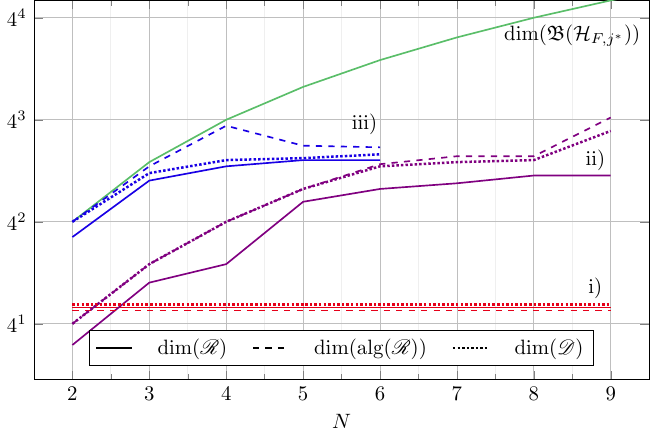}
\caption{\normalfont Comparison between the dimension of the operators spaces $\Rs$, $\alg(\Rs)$, the distorted algebra $\Ds$, and the size of the entire block associated with $j^*$, $\Bf(\Hc_{F,j^*})$, for different central-spin models undergoing collective dissipation of strength $\Lambda$ (here, $\Lambda= 1$). 
Specifically, we consider: 
i) XXZ coupling in $H_{\text{int}}$ (red curves), $A_x=A_y=1$, $A_z=2$, $\omega_1 = 1$, and $H_B=0$ (note that the represented dimensions have been shifted by 0.2 to improve readability).
ii) XXZ coupling (blue curves), $A_x = A_y = 1$, $A_z=2$, $\omega_1=1$, $H_{B,0}=0$ and $H_{B,\text{int}}$ an Ising interaction, with $\lambda=2$.
iii) XYZ coupling (purple curves), $A_x=1, A_y=3, A_z=2, \omega_1=1$, $H_{B,0} \ne 0$ with $\mu=1$, and $H_{B,\text{int}}=0$. We also set $\eta=0$ in all cases. }  
\label{fig:reachable_dim}
\end{figure}

While in the discussion above we have chosen a setting amenable to analytic treatment, for more general dissipative central-spin models the reachable-based MR can still be computed numerically. In Fig.\,\ref{fig:reachable_dim} we compare the dimensions of the operator spaces of interest with the dimension of the full block $\Hc_{F,j^*}$, for the same initial condition but different choices of system-bath Hamiltonian, under which a strong symmetry is still maintained. We can observe that even relaxing the assumptions made before (e.g., $A_x = A_y$), the entire block is not explored when starting from the considered initial condition $\rho_0,$ and MR beyond the one afforded by only considering the relevant observables is possible\footnote{Due to  numerical errors, it is possible that the algorithms we used to construct the algebras might consider some operators as new linearly independent generators when they should not. In this uncommon occurrence, the algebra dimensions shown in Fig.\,\ref{fig:reachable_dim} might be slightly overestimated, with the optimal reduction actually offering even better performance.}. The figure also demonstrates the advantage of using modified products in the construction of the reachable distorted algebra $\Ds$: in cases ii) and iii), the curves displaying the dimension of the computed $\Ds$ are lower than the corresponding ones for standard, undistorted algebras. 

\begin{remark}
The choice of the initial state of the bath leading to significant reduction is not limited to $\ket{W}.$ In fact, the reduced model in Eq.\,\eqref{eq:reduced_reachable_model} also holds for any initial condition inside $\Rs$ in Eq.\,\eqref{reach}. This includes, for instance, the W-state on both system and bath, i.e. $\ket{W'}=\frac{1}{N} \ketbra{1}{1}\otimes\ketbra{j^*}{j^*} + \frac{NB}{N}\ketbra{0}{0}\otimes\ketbra{m^*}{m^*}$. More generally, we may consider states in $\Df(\Hc_S\otimes\Hc_W)$, where $\Hc_W\equiv\Span\{\ket{W_v}\}$ and $\ket{W_v}$ are qubit Dicke states \cite{Dicke}, that is, uniform superposition of states with precisely $v$ excited spins, e.g., $\ket{W_0} = \ket{00\dots0}$, $\ket{W_1}=\ket{W}$, $\ket{W_2} = \ket{110\dots0}+\ket{101\dots0}+\dots+\ket{00\dots11}$, and $\ket{W_{N_B}}=\ket{11\dots1}$. Clearly,  the W-states we have considered belong to such a set. 
\end{remark}

\subsection{Boundary-driven XXZ model}
\label{sec:examples_XXZ}

In this section we consider a paradigmatic example of many-body dissipative dynamics, provided by a boundary-driven XXZ spin-$1/2$ chain. This model has been extensively investigated in the context of non-equilibrium statistical mechanics, with a focus on understanding magnetization and heat transport \cite{Mahler2003,Buca_2012,mendoza-arenas_heat_2013, popkov2013manipulating, bertini_finite-temperature_2021, landi_nonequilibrium_2022}. In these works, primary emphasis has been given to characterizing the ``non-equilibrium SS'' of the dynamics, and the resulting SS spin current and magnetization profiles. Here, we are going to address the full time-dependent transient behavior and show how we can reduce the description of the model to one that exactly reproduces the expectation values of the target physical observables at any time, not just under SS conditions. 

We consider a finite, open-boundary chain of $N$ spin-$1/2$ on a line, interacting with each other via a nearest-neighbor XXZ coupling Hamiltonian, of the form
\begin{equation}
H = \Gamma \sum_{j=1}^{N-1} \big(\sigma_{x}^{(j)}\sigma_{x}^{(j+1)} + \sigma_{y}^{(j)}\sigma_{y}^{(j+1)} + \Delta\sigma_{z}^{(j)}\sigma_{z}^{(j+1)}\big) ,
\label{eq:XXZ_Hamiltonian}
\end{equation}
where $\Gamma >0$ and $\Delta >0$ are the exchange coupling constant and the anisotropy parameter, respectively. In addition, the boundary spins, $1$ and $N$, couple to a Markovian environment, via Lindblad operators $L_u$ that are assumed to be site-local. Specifically, as in \cite{Buca_2012}, we consider a set of Lindblad operators that include both loss and gain of spin-excitation at the boundary \footnote{If only loss processes are present, the  Lindbladian generator is exactly solvable in the sense that the full spectrum and the resulting dynamics may be obtained by a dissipative Bethe Ansatz approach \cite{Buca_2020}.}: 
\begin{align*}
    &L^{(1)}_+ = \sqrt{\alpha}\, \sigma_{+}^{(1)}, &L^{(N)}_+ = \sqrt{\beta}\, \sigma_{+}^{(N)}, \\
    &L^{(1)}_- = \sqrt{\beta} \,\sigma_{-}^{(1)}, &L^{(N)}_- = \sqrt{\alpha} \,\sigma_{-}^{(N)},
\end{align*}
where we defined $\alpha \equiv \kappa (1-\mu)$ and $\beta \equiv \kappa (1+\mu)$, in terms of a fixed coupling strength $\kappa >0$ and the parameter $0\leq \mu\leq 1$ quantifying the boundary driving .

As observables of interest, we consider both the local spin magnetizations, $\{ \sigma_z^{(j)} \}_{j=1}^{N}$, and the spin currents, described by two-body operators 
\[J_j\equiv \sigma_x^{(j)}\sigma_y^{(j+1)}-\sigma_y^{(j)}\sigma_x^{(j+1)}, \quad 
j=1,\dots, N-1. \]
Accordingly, we are interested in finding a reduced model that reproduces the time-dependent expectation values as outputs,  
$$Y_j(t) = \tr(O_j \rho(t)), \quad O_j \in\{\sigma_z^{(j)}\}\cup\{J_j\},$$ 
starting from an arbitrary initial condition $\rho_0\in\Df(\Hc)$. We thus perform the observable-based MR reduction, which means that the output algebra $\Os=\alg(\Ns^\perp)$ is the main object we need to determine. 

\subsubsection{Symmetries of the model}

Consider a permutation matrix $P$ that exchanges the spin in location $i$ with the spin in location $N-i$, that is, for a sequence of $N$ bits $b_1\dots b_{N-1}b_{N}$, we have $P\ket{b_1\dots b_{N-1}b_{N}} = \ket{b_{N}b_{N-1}\dots b_{2}b_{1}}$. Then, as noted in \cite{Buca_2012}, the unitary matrix $S \equiv P\prod_{j=1}^{N}\sigma_x^{(j)}$, is a weak ${\mathbb Z}_2$ symmetry for the dynamics. Since $S$ does not commute with the observables of interest, however, Proposition \ref{strongcond} cannot be used to relate this symmetry and the observable space $\Ns^\perp$.

Although not observed in \cite{Buca_2012}, the unitary group generated by the global magnetization operator, $\Gs_\varphi \equiv \{e^{-i\varphi M}\}$, with $\varphi\in\Rb$ and $M = J_z= \tfrac{1}{2}\sum_{j=1}^N  \sigma_z^{(j)}$, can be seen to provide a weak symmetry for the model (see Appendix \ref{sec_appendix_symm} for an explicit proof and further comments). Notably, this unitary group commutes with the observables of interest, whereby it follows, by applying Proposition \ref{strongcond}, that $\Os=\alg(\Ns^{\perp}) \subseteq \Cb\Gs_\varphi'$.

\begin{figure}[t]
\centering
\includegraphics[width=0.9\columnwidth]{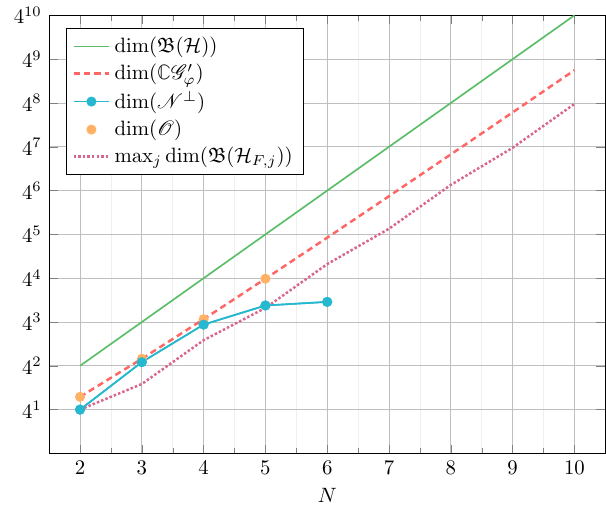}
\vspace*{-2mm}
\caption{\normalfont 
Dimensions of the observable space $\Ns^\perp$ and output algebra $\Os$ compared to the dimensions of the full operator space $\Bf(\Hc)$ for the boundary-driven XXZ model.} While the dimension of $\Cb\Gs_\varphi'$ is given by the analytical expression in Eq.\,\eqref{dimphi}, the dimensions of $\Ns^\perp$ and $\Os$ have been computed numerically using the procedure described in Sec.\,\ref{sec:observablered}. The model parameters are $\Gamma =1$, $\Delta=0.5$, $\kappa=1.2$, and $\mu=0.1$. 
\label{fig:xxz_dim}
\end{figure}

By construction, the elements of $\Cb\Gs_\varphi'$ must share the same eigenspaces as $M$. The number of distinct eigenvalues of $M$ is $N+1$ and the dimensions of the eigenspaces of $M$ follow the Pascal/Tartaglia triangle rule, i.e., $\binom{N}{k}$ for $k = 1,\dots,N$. With this, and using the Chu-Vandermonde identity and Stirling's approximation, one finds the dimension of $\Cb\Gs_\varphi'$ as
\begin{equation}
\dim(\Cb\Gs_\varphi') = \sum_{k=0}^{N} \binom{N}{k}^{\!2} = \binom{2N}{N} \approx \frac{4^N}{\sqrt{\pi N}} ,
\label{dimphi}
\end{equation} 
where the last approximate equality holds for large $N$. The above provides an upper bound for the dimension of the output algebra $\Os$. For later discussion, it is also worth noting that the dimension of the largest block in the Wedderburn decomposition of $\Cb\Gs_\varphi'$ scales as 
\[ \max_k \dim(\Bf(\Hc_{F,k}))= \binom{N}{\lfloor{{N}/{2}}\rfloor}^2\approx 2 \,\frac{4^N}{\pi N}  ,  \]
that is, the size of this block is reduced by a factor $\sqrt{N}$, which may be significant for large $N$.

To verify whether one may achieve a larger reduction than the one provided by the weak symmetry group $\Cb\Gs_\varphi'$, the observable algebra $\Os=\alg(\Ns^\perp)$ was numerically determined by implementing the algorithm described in Sec.\,\ref{sec:observablered}. Numerical calculations of the observable space $\Ns^\perp$, the output algebra $\Os$, and its Wedderburn decomposition were carried out for chains with up to $N=5$ spins; the results  are summarized in Table \ref{tab:xxz_dim} and Fig.\,\ref{fig:xxz_dim}. As one may observe from Fig.\,\ref{fig:xxz_dim}, the dimension $\Os$ coincides with the one of $\Cb\Gs_\varphi'$ for up to $N=5$. Moreover, from Table \ref{tab:xxz_dim}, one can observe that the Wedderburn decomposition of $\Os$ has precisely $N+1$ blocks of the same dimension as the eigenspaces of $M$. Further numerical testing confirmed that, up to $N=5$ qubits, $\Os = \Cb\Gs_\varphi'$. We conjecture this to be true for \emph{arbitrary} $N$. 

\newcolumntype{C}[1]{>{\centering\let\newline\\\arraybackslash\hspace{0pt}}m{#1}}
\begin{table}[t!]
    \centering
    \begin{tabular}{|C{8pt}|C{18pt}|C{18pt}|c|}
    \hline
        & \multicolumn{2}{c|}{$\dim(\cdot)$} & Output algebra structure \\[2pt]
         $N$& $\Ns^\perp$ & $\Os$ & $\Os\simeq$ \\[2pt]\hline\hline
         &&&\\[-5pt]
         2& 4 & 6  & $2\Cb\oplus\Cb^{2\times 2}$\\[5pt]
         3& 18 & 20 &   $2\Cb \oplus 2\Cb^{3\times 3}$\\[5pt]
         4& 59 & 70 &  $2\Cb \oplus 2\Cb^{4\times 4} \oplus \Cb^{6\times 6}$ \\[5pt]
        5 & 108 & 252 & $2\Cb \oplus 2\Cb^{5\times 5} \oplus  2\Cb^{10\times 10} $ \\ [5pt]
         6 & 121 & 924 & 
         $2\Cb \oplus 2\Cb^{6\times 6} \oplus  2\Cb^{15\times 15} 
         \oplus \Cb^{20\times 20} $ \\ 
         [2pt]\hline
    \end{tabular}
    \caption{\normalfont Numerically determined dimensions of the observable space $\Ns^\perp$ and output algebra $\Os$ for the boundary-driven XXZ model with the same parameters used in Fig.\,\ref{fig:xxz_dim}.}
    \label{tab:xxz_dim}
\end{table}

\begin{figure*}[ht]
\centering
\includegraphics[width=0.98\textwidth]{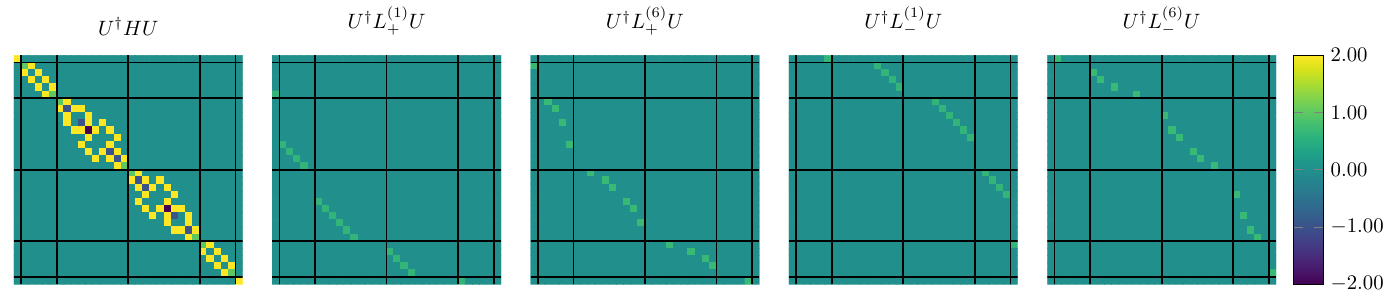}
\caption{\normalfont Graphical representation of the Hamiltonian and noise operators for the boundary-driven XXZ chain with $N=5$ spins and parameters as in Fig.\,\ref{fig:xxz_dim}. All matrices are represented in the base that block-diagonalizes the output algebra, namely, $U^\dag \Os U=\Cb \oplus \Cb^{5\times 5} \oplus \Cb^{10\times 10} \oplus \Cb^{10\times 10} \oplus \Cb^{5\times 5} \oplus  \Cb$, with the black lines highlighting the block structure.} 
    \label{fig:xxz_blocks}
\end{figure*}

\begin{figure*}
    \begin{subfigure}[t]{0.46\textwidth}
        \centering
        \includegraphics[width=\textwidth]{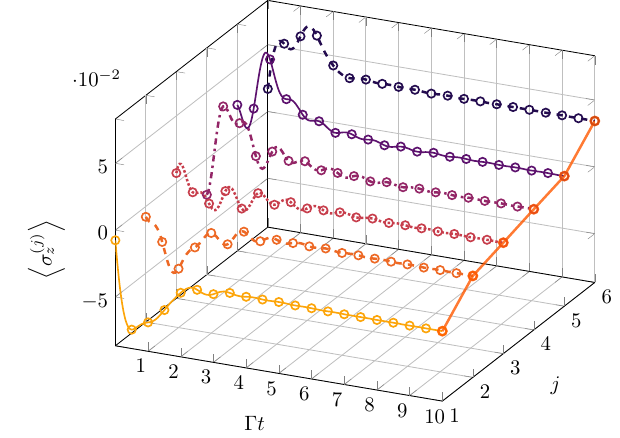}
    \end{subfigure}
    \hfill
    \begin{subfigure}[t]{0.49\textwidth}
        \centering
        \includegraphics[width=\textwidth]{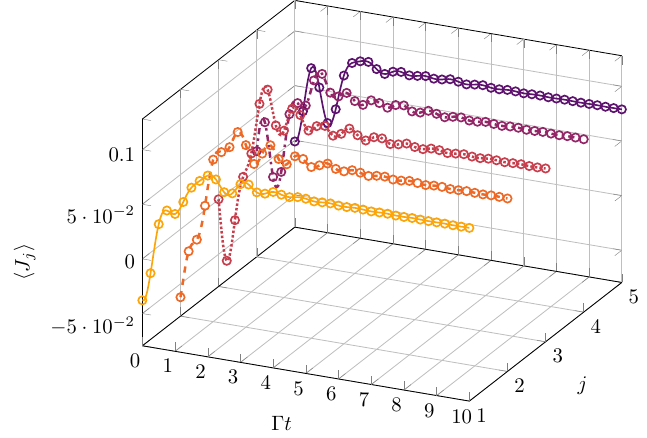}
    \end{subfigure}
    \vspace*{-2mm}
\caption{\normalfont Numerical results for the time-dependent expectation values of the local magnetization, $\langle\sigma_z^{(j)}(t)\rangle$ (left), and the spin currents, $\expect{J_j(t)}$ (right), for the boundary-driven XXZ chain with $N=6$ and parameters as in Fig.\,\ref{fig:xxz_dim}, 
starting from a random initial state $\rho_0$. In both graphs, the continuous line (resp. circles) represents the evolution computed from the reduced (resp. full) model. The dark orange curve, which interpolates between the final values of the local magnetization, closely resembles the behavior reported in Figure 2.a of \cite{Buca_2012}, which was obtained at SS  from similar parameters.} 
\label{fig:xxz_magnetization}
\end{figure*}

\subsubsection{Numerical model reduction and simulations}

Having obtained the relevant output algebra $\Os$ , let us denote with $U$ the unitary change of basis that decomposes it into its Wedderburn decomposition. In particular, $U$ can be taken as the permutation matrix that orders the eigenvalues of $M$ and groups them together. Thanks to their commutation with $\Gs_\varphi$, both the system's Hamiltonian and the observables of interest assume a block-diagonal structure in this new basis, with the same block dimensions of the Wedderburn decomposition of $\Os$. The four noise operators, instead, assume an upper or lower block-diagonal structure in this basis, as depicted in Fig.\,\ref{fig:xxz_blocks}. As in the case of the central spin model under local dissipation we discussed in Sec.\,\ref{sec:examples_central_spin_dissipative}, the off-diagonal terms connect different blocks of $\Os$, effectively creating a probability pumping from one block to the other. 

A comparison between results obtained for the full vs. the reduced model is provided in Fig.\,\ref{fig:xxz_magnetization}, where a simulation from a random initial condition is performed and the expectation values of the local magnetization, $\sigma_z^{(j)}$, and spin currents, $J_j$, are compared as a function of time. As one may see, the outputs of the full and reduced models exactly coincide, again demonstrating how the proposed algorithm correctly reduces the open quantum dynamics of interest. 

\begin{remark}
By removing the dissipation effects, or by considering Lindblad operators that commute with the total magnetization $M$, the weak symmetry becomes a strong one. This, in turn, means that each block of the state $\check{\rho}_k$ evolves independently of the others and thus, as we have also pointed out in Remark \ref{remark:computation_Advantage}, we are no longer limited in size by the dimension of the full operator algebra, but only by the dimension of the largest block in the Wedderburn decomposition. This difference can be seen in Fig.\,\ref{fig:xxz_dim} when comparing the scaling between $\dim(\Cb\Gs')$ (orange curve) and $\max_k \dim(\Bf(\Hc_{F,j}))$ (purple curve). 
\end{remark}

\subsection{Model reduction to study encoded dynamics under non-ideal error models}
\label{sec:examples_encoding}

While the MR scenarios we have discussed so far are concerned with Lindblad dynamics on {\em physical} degrees of freedom, our approach is also applicable to studying the dynamics of quantum information encoded in \emph{logical} degrees of freedom \cite{knill-encoding}. For continuous-time dynamics as we consider here, a relevant application arises in the context of obtaining simpler descriptions of the noisy dynamics that an ``infinite-distance'' quantum code, corresponding to a {\em decoherence-free subspace} or a {\em noiseless subsystem} \cite{PhysRevA.63.012301}, \cite{KLV}, \cite{Lidar_Brun_2013}, \cite{ticozzi2010quantum}, undergoes when the physical generator deviates from the ideal error model the code is designed to protect against: both unwanted unitary evolution within the code space or leakage out of the code space may then occur, resulting in logical and leakage errors, respectively.

To present the idea in its simplest setting, we focus on a scenario where a {\em single} logical qubit is encoded in $N$ physical qubits. The information is then stored in the expectation values of logical qubit observables, which can be chosen as the output operators of interest for our observable-based MR approach. In terms of the observable spaces we previously introduced, one can see that the dynamics preserves the encoded information {\em only if} the associated output algebra {\em coincides} with the logical algebra generated by the observables in which the information was initially encoded. When this is not the case, the observable subspace ${\cal N}^\perp$ and the output algebra $\Os$ allow us to determine the (operator) subspace on which the information spreads and, in turns, our observable-based MR procedure allows us to compute the reduced model needed to simulate the resulting faulty dynamics more efficiently than in the physical $4^N$-dimensional operator space. 

Let $\Hc_\Cc = \Cb^2 \equiv \text{span}\{|0\rangle_L, |1\rangle_L \}$ denote the code space, with associated logical basis states, and assume that, as long as the system's state is properly initialized in ${\cal D}({\cal H}_\Oc)$, the code is invariant under the evolution generated by the ideal Lindbladian $\Lc^{\text{id}}$. Unwanted dynamical behavior is introduced by an ``error generator'' $\Lc^{\text{err}}$, which may include both a Hamiltonian and a dissipative component, corresponding to $H^{\text{err}}$ and Lindblad operators, $\{ L^{\text{err}}_j\}$. The task is to reproduce the encoded information dynamics, as determined by the time-dependent expectation values of the logical qubit observables, say, $\{C_q , q=0,x,y,z\}$, which obey the appropriate qubit commutation and anti-commutation relationships \cite{ViolaJPA}. Thus, formally, the output $Y(t)=\tau(t)$, with $\tau(t) \in \Dc( \Hc_C)$ being the logical qubit state, and the dynamical model we wish to reduce may be written as 
\begin{equation}
\hspace*{-3mm}    \begin{cases}
        \dot{\rho}(t) = \Lc^{\text{id}}[\rho(t)] + \Lc^{\text{err}}[\rho(t)] \\
        \tau(t) = \tfrac{1}{2}\sum_{q=0,x,y,z} \sigma_q \tr[C_q\rho(t)] 
    \end{cases}\hspace*{-3mm}, 
    \;\rho(0)\in{\cal D}({\cal H}_\Oc).
    \label{encodedFull}
\end{equation}

\subsubsection{Error dynamics on a three-bit subsystem code} 

As a concrete illustrative example, we consider the smallest infinite-distance noiseless subsystem code proposed in \cite{KLV,DeFilippo,fortunato2003exploring} for protecting a qubit against arbitrary collective (permutation-invariant) noise using $N=3$ physical qubits. Let $s_{jk}\equiv \vec{\sigma}^{(j)}\cdot\vec{\sigma}^{(k)}= \sigma_x^{(j)}\sigma_x^{(k)} + \sigma_y^{(j)}\sigma_y^{(k)} + \sigma_z^{(j)}\sigma_z^{(k)}$, $j,k\in \{1,2,3\}$, denote the two-spin operators that are invariant under collective spin rotations, generated by collective angular momentum operators $\{J_u, u=x,y,z\}$. One may show \cite{fortunato2003exploring} that $P_{1/2}\equiv \frac{1}{2}\left[\one_8-\frac{1}{3}\left(s_{12}+s_{23}+s_{31}\right)\right]$ is the projector onto the 4-dimensional subspace $\Hc_{1/2} \subset \Hc=(\Cb^2)^{\otimes 3}$ of states with total spin angular momentum eigenvalue $j=1/2$. If, as in Sec.\,\ref{sec:examples_central_spin}, $\Sc_3$ denotes the permutation group on three objects, we have $s_{jk} \in \Cb\Sc_3$, while $J_u \in \Cb\Sc_3'$. As one may verify, the following four observables in $\Cb\Sc_3$ then obey the algebraic properties of valid qubit observables \cite{ViolaJPA}: 
\begin{align}
      O_0 &= P_{1/2}, \notag \\
      O_x &= \frac{1}{2}(\one+s_{12})P_{1/2}, \notag \\
      O_y &= \frac{\sqrt{3}}{6}(s_{23}-s_{31})P_{1/2}, \notag \\
      O_z &= -i C_x C_y.
      \label{nsobs}
\end{align}
These observables form a 4-dimensional algebra,
$$ \Es_\Oc\equiv\alg(\{O_q\})=\Cb^{2\times 2}\otimes \one_2 \oplus \zero_4 = \Bf(\Hc_\Cc) \otimes \one_2 \oplus \zero_4,$$ 
on which we can encode a logical qubit: that is, the code space $\Hc_\Cc$ is associated with a tensor-product factor in a virtual-subsystem decomposition $\Hc\equiv  (\Hc_\Cc\otimes \Hc_Z ) \oplus \Hc_{3/2}$ of $\Hc$. Explicitly, the logical basis states may be given by the following identification:
\begin{align}
    \ket{0}_L\otimes\ket{+1/2}_Z &= \frac{1}{\sqrt{3}} \left(\ket{001}+\omega \ket{010}+\omega^2\ket{100}\right), \notag \\ 
    \ket{1}_L\otimes\ket{+1/2}_Z &= \frac{1}{\sqrt{3}}\left(\ket{001}+\omega^2\ket{010}+\omega\ket{100}\right),
    \label{ns}
\end{align}
with the two additional basis states $\{\ket{i}_L\otimes \ket{-1/2}_Z, i=0,1\}$ corresponding to total $J_z$-angular momentum eigenvalue $m=-1/2$ being obtained by flipping every spin, and  $\omega\equiv e^{i 2\pi/3}$.

By construction, if the ideal Lindbladian $\Lc^{\text{id}}$ comprises permutationally-invariant Hamiltonian and Linblad operators, $\Es_C$ is left invariant and information encoded in $\Hc_\Cc$ is unaffected, since $A=\one_2 \otimes \Cb^{2\times 2} \oplus \zero_4 $, for all $A \in \Cb\Sc_3'$: in such a case, $\Oc$ is a {\em fixed}, noiseless code \cite{blume2010information,ticozzi2010quantum}. 
In what follows, let us assume that $\Lc^{\text{id}}$ comprises Hamiltonian dynamics described by a Heisenberg exchange Hamiltonian, $H^{\text{id}}=\Gamma(s_{12}+s_{23}+s_{13}) \in \Cb\Sc_3 \cap \Cb\Sc_3'$, $\Gamma >0$, along with collective noise, described by Lindblad operators $L^{\text{id}}_u \in \{\kappa_u J_u, \kappa_u >0, u=x,y,z \}$\footnote{In the terminology of \cite{kempe2001theory}, the case where noise couples along arbitrary axes, $\kappa_u >0$, for all $u$, is also referred to as ``strong collective decoherence'' -- as opposed to ``weak collective decoherence'', in which case only collective dephasing is present, $\kappa_x=\kappa_y=0$.}. Imagine now that unwanted error behavior results from both breaking permutation symmetry ($\Gamma_{jk} \ne \Gamma$ for some pair) or (and) rotational symmetry ($\Delta \ne 1$) in the Hamiltonian, which becomes an XXZ Hamiltonian of the form 
$$ H = \sum_{j<k} \Gamma_{jk} [s_{jk}  
+ (\Delta-1)\sigma_z^{(j)}\sigma_z^{(k)}\big], \quad \Gamma_{jk} \geq 0, \Delta \geq 0,$$
and {\em local} Markovian dephasing, via Lindblad operators 
$$\{ L^{\text{err}}_j \equiv L_{z}^j = \gamma_{j}\sigma_z^{(j)}, j=1,2,3 \} , \quad \gamma_{j}>0.$$
By implementing the algorithm described in Sec.\,\ref{sec:observablered}, we may numerically determine the output algebra $\Os$ that results from the combined dynamics in Eq.\,\eqref{encodedFull}. Different possibilities arise, depending on the interplay between the nominal and error components (for illustration, see also Fig.\,\ref{fig:encoding_evolution}, where we consider evolution starting from the initial state $\rho_0 = \ketbra{\psi_0}{\psi_0}$, with $\ket{\psi_0} \equiv \ket{0}_L \otimes \ket{+1/2}_Z$ in Eq.\,\eqref{ns}):

\smallskip

(i) $\Delta=1$, $\Gamma_{13}=0$, in the presence of general collective noise and no local dephasing ($\gamma_j=0, \forall j$, and $\kappa_u >0, \forall u$). In this case, $H$ no longer belongs to $\Cb\Sc_3'$. However, since $H$ is still a linear combination of scalars under collective rotations, $\Es_\Oc$ is preserved and information encoded in $\Oc$ evolves unitarily, with the output algebra $\Os=\Es_\Oc$. Since unitary dynamics within the code space are reversible, $\Oc$ is {\em unitarily noiseless} \cite{blume2010information,ticozzi2010quantum} 
[Fig.\,\ref{fig:encoding_evolution} (a)].

(ii) $\Delta \ne 1$, $\Gamma_{13}=0$, in the presence of collective dephasing only  ($\kappa_x=\kappa_y = 0$) and again no local dephasing. In this case, $H$ also breaks collective rotational symmetry, $H\not \in \Cb \Sc_3$, implying that $\Es_\Oc$ is no longer preserved and information encoded in $\Oc$ no longer evolves unitarily. The resulting output algebra is now $\Os=\alg(C_q,\sigma_z^{(1)}\sigma_z^{(2)}P_{1/2})$, of dimension $\dim(\Os_D) = 9$, corresponding to a ``qutrit'' operator algebra [Fig.\,\ref{fig:encoding_evolution} (b)].

(iii) $\Delta \ne 1$, $\Gamma_{13}=0$, in the presence of both collective and local dephasing ($\kappa_x=\kappa_y = 0$, $\gamma_j \ne0$ for some $j$). Once again, in this case $\Es_C$ is not preserved; the output algebra generated by this evolution is found to be the same as above, $\Os=\alg(C_q,\sigma_z^{(1)}\sigma_z^{(2)}P_{1/2})$. The fact that, despite also breaking permutation symmetry, local dephasing causes no further growth of $\Os$ may be explained by noticing that the Eq.\,\eqref{ns} span a decoherence-free subspace in $\Hc_{1/2}$ and, as such, are invariant under collective $z$ errors \cite{fortunato2003exploring}; further to that, we have $[L_z^j, H]=0$ [Fig.\,\ref{fig:encoding_evolution} (c)-(d)].

(iv) $\Delta \ne 1$, $\Gamma_{13}=0$, in the presence of general collective noise and local dephasing ($\gamma_j>0$ for some $j$, and $\kappa_u >0, \forall u$). In this case, the output algebra becomes the full 16-dimensional algebra, $\Os=\Bf(\Hc_\Cc)$. Interestingly, however, similar to the previous example in Sec.\,\eqref{sec:examples_XXZ}, the observable space remains significantly smaller: here, we find $\dim(\Ns^\perp)=16$. Thus, although no quantum MR is possible in this case, linear MR remains viable if desired.

\subsubsection{Reduced encoded dynamics}

In order to compute the reduced dynamics onto the output algebra $\Os$ that is relevant to both case (i) and (ii) above, it is necessary to compute its Wedderburn decomposition. Using the numerical methods presented in \cite{de2011numerical} it is possible to find a unitary matrix $U$ that block-diagonalizes $\Os$. 
Specifically, the transformation we use here may described as effecting a change of basis relative to which the Hilbert space decomposes as $\Hc = (\Hc_{Q} \otimes \Hc_{F}) \oplus \Hc_{R}$, where $\Hc_{Q} \equiv \Span\{\ket{0}_Q, \ket{1}_Q, \ket{2}_Q\}$ is a 3-dimensional qutrit space, the co-subsystem  $\Hc_{F} \equiv \Span\{\ket{-1}_F, \ket{+1}_F\}$, and $\Hc_{R} \equiv \Span\{\ket{000},\ket{111}\}$. In this way, we have  
\begin{align}
    \ket{0}_Q\otimes\ket{+1}_F &= \frac{1}{2} \left(\ket{001}-\ket{010}-\ket{101}+\ket{110}\right),\notag\\ 
    \ket{0}_Q\otimes\ket{-1}_F &= \frac{1}{2}\left(-\ket{001}+\ket{010}-\ket{101}+\ket{110}\right),\notag \\
    \ket{1}_Q\otimes\ket{+1}_F &= \frac{1}{2}\left(-\ket{001}-\ket{010}-\ket{101}-\ket{110}\right),\notag \\
    \ket{1}_Q\otimes\ket{-1}_F &= \frac{1}{2}\left(\ket{001}+\ket{010}-\ket{101}-\ket{110}\right),\notag 
            \end{align}
    \begin{align}
    \ket{2}_Q\otimes\ket{\pm 1}_F &= \frac{1}{\sqrt{2}}\left(\ket{011}\pm \ket{100}\right), 
    \label{Qbasis}
\end{align}
where basis states in $\Hc_F$ have been labeled according to the eigenvalue of $\sigma_x^{(1)}\sigma_x^{(2)}\sigma_x^{(3)}$. With this transformation, the desired Wedderburn decomposition of $\Os$ reads 
\[U^\dag \Os U = (\Cb^{3\times 3} \otimes \one_2) \oplus \zero_2.\] 

We can then proceed to obtain the CP reduction and injection maps that yield MR onto the algebra $\check{\Os} = \Cb^{3\times 3}$, namely, 
\begin{align*}
    \Rc(X) = \tr_{\Hc_F}(W U^\dag X U W^\dag), \\
    \Jc\left(\check{X}\right) = U W^\dag \bigg(\check{X}\otimes \frac{\one_2}{2}\bigg) W U^\dag ,
\end{align*}
where $W\equiv \left[\begin{array}{c|c}
    \one_6 & \zero_{6\times 2}
\end{array}\right]\in\Cb^{6\times 8}$. 
Accordingly, the reduced state is given by 
\[\check{\rho} \equiv \Rc[\rho] \in {\cal D}(\Hc_Q) .\]
Note that the reduction map $\Rc$ is CP; however, it is TP only over the support of $\Os$, which is not full.

To obtain a representation of the reduced Lindblad generator that describes the error dynamics of Eq.\,\eqref{encodedFull}, we need to find a reduced version of the observables, as well as the Hamiltonian and noise operator entering the reduced dynamics. Representing the qubit observables of Eq.\,\eqref{nsobs} in the basis of Eq.\,\eqref{Qbasis}, the block-diagonal structure is manifest, as expected from the fact that $C_q\in\Os$ for all $q$:
\begin{align*}
    U^\dag C_q U &= ({\check{O}}_q\otimes \one_2)\oplus \zero_2,
\end{align*}
where the explicit form of ${\check{O}}_q$ is obtained as ${\check{O}}_q \equiv\Jc^\dag(C_q)$:
\begin{align*}
    {\check{O}}_0 &\equiv \frac{1}{3}\left[\!\!\begin{array}{ccc}
    3&0&0\\
    0&1&\sqrt{2}\\
    0&\sqrt{2}&2
    \end{array}\!\right], \\
    {\check{O}}_x &\equiv - \frac{1}{2}\left[\!\!\begin{array}{ccc}
    -1& 1& \sqrt{2}\\
    1& {1}/{3}& {\sqrt{2}}/{3}\\
    \sqrt{2}& {\sqrt{2}}/{3}&-{2}/{3}
    \end{array}\!\right],
\end{align*}
\begin{align*}
    {\check{O}}_y &\equiv \frac{\sqrt{3}}{6}\left[ \!\!\begin{array}{ccc}
    -3&-1&-\sqrt{2}\\
    -1&1&\sqrt{2}\\
    -\sqrt{2}&\sqrt{2}&2
    \end{array}\!\right],\\
    {\check{O}}_z &\equiv \frac{i\sqrt{3}}{3}\left[\!\!\begin{array}{ccc}
    0&1&\sqrt{2}\\
    -1&0&0\\
    -\sqrt{2}&0&0
    \end{array}\!\right]\!\!.
\end{align*}
These four observables $\{\check{O}_q\}$ obey su(2) commutation relationships and, in addition, $\check{O}_q^2 = \check{O}_0$, for all $q$, as demanded for valid logical-qubit observables \cite{ViolaJPA}. 

\begin{figure*}[ht]
\centering
\hspace*{-3mm}
     \begin{subfigure}{0.24\textwidth}
         \includegraphics[width=0.85\textwidth]{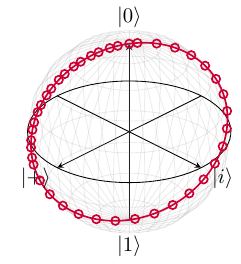}
         \caption{\centering$\Delta=1$, $\gamma_{j}=0$, $\forall j$}
     \end{subfigure}
    \begin{subfigure}{0.24\textwidth}
        \includegraphics[width=0.85\textwidth]{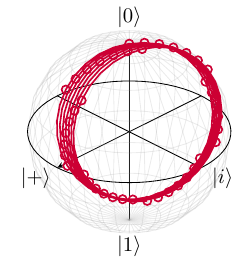}
        \caption{\centering$\Delta=10$, $\gamma_{j}=0$, $\forall j$}
    \end{subfigure}
    \begin{subfigure}{0.24\textwidth}
         \includegraphics[width=0.85\textwidth]{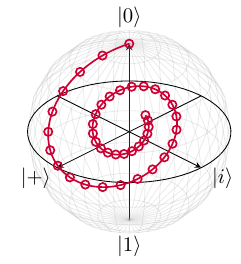}
         \caption{\centering$\Delta=2$, $\gamma_{1}=1$, $\gamma_2=\gamma_3=0$}
     \end{subfigure}
    \begin{subfigure}{0.24\textwidth}
        \includegraphics[width=0.85\textwidth]{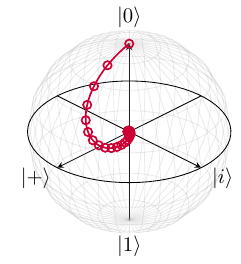}
        \caption{\centering$\Delta=2$, $\gamma_{1}= \gamma_2=\gamma_3=1$}
     \end{subfigure}
    \caption{\normalfont Time development of the logical qubit corresponding to encoded observables $\{C_q\}$ [Eq.\,\eqref{nsobs}], under error evolution due to permutation-symmetry-breaking and anisotropy in the Hamiltonian, plus local dephasing, $\gamma_j >0$, for the initial code state $\rho_0 = (\ketbra{0}{0})_L \otimes (\ketbra{+1/2}{+1/2})_Z$ [Eq.\,\eqref{ns}].
    In (a), only permutation symmetry is broken, and the code remains unitarily noiseless under general collective noise.
    In (b), anisotropy is added with only collective dephasing, while in (c) and (d) local dephasing is also present. The relevant output algebra grows to the one of a qutrit in the last three cases. The solid lines indicate the trajectory of the reduced model, Eq.\eqref{encodedFull}, while the empty dots denote the trajectory of the original one, Eq.\,\eqref{nsred}. The model parameters are $\Gamma_{12}=\Gamma_{23}=1$, $\Gamma_{13}=0$, and $\kappa_x=\kappa_y=\kappa_z=1$ in (a), while $\kappa_x=\kappa_y=0$, $\kappa_z=1$ in (b), (c) and (d).}
\label{fig:encoding_evolution}
\end{figure*}

Similarly, one can write down the reduced Hamiltonian in the basis defined by $U$, 
\begin{align}
    U^\dag H U &= ({\check{H}}\otimes\one_2)\oplus \bar{\Gamma}\Delta\one_2 ,
    \label{HamQF}
\end{align}
where $\bar{\Gamma} \equiv \sum_{j<k}\Gamma_{j,k}$
and $\check{H}$ may again be obtained as 
\begin{widetext}
\begin{eqnarray*}
{\check{H}} = \Jc^\dag(H) = 
 \left[\! \begin{array}{ccc}
(-2-\Delta)\Gamma_{23}&(\Gamma_{13}-\Gamma_{12})\Delta&(\Gamma_{13}-\Gamma_{12})\sqrt{2}\\
(\Gamma_{13}-\Gamma_{12})\Delta&(2-\Delta)\Gamma_{23}&-(\Gamma_{13}+\Gamma_{12})\sqrt{2}\\
(\Gamma_{13}-\Gamma_{12})\sqrt{2}&-(\Gamma_{13}+\Gamma_{12})\sqrt{2}&(\Gamma_{12}-\Gamma_{13}-\Gamma_{23})\Delta
\end{array}\! \right]\!. 
\end{eqnarray*}
\end{widetext}

Finally, in this basis the ideal (collective) and error (local) dephasing noise operators read:  
\begin{align}
    U^\dag L_{z}^{\text{id}} U &= (\check{J}_z \otimes\sigma_x) \oplus \frac{3\kappa_z}{2}\sigma_z,\notag\\
    U^\dag L_{z,j} U &= ({\check{L}}_{z,j}\otimes \sigma_x) \oplus \gamma_j\sigma_z, \label{NoiseQ}
\end{align}
with reduced noise operators given explicitly by 
\begin{align*}
    {\check{L}}_{z,1} &=\gamma_1\left[\begin{array}{ccc}
    -1 & 0 & 0 \\0 & -1 & 0\\0 & 0 & 1
    \end{array}\right],\\
    {\check{L}}_{z,2} &= \gamma_2\left[\begin{array}{ccc}
    0 & 1 & 0\\1 & 0 & 0\\0 & 0 & - 1
    \end{array}\right], \\
    {\check{L}}_{z,3} &= \gamma_3\left[\begin{array}{ccc}
    0 & -1 & 0\\ -1 & 0 & 0\\0 & 0 & -1
    \end{array}\right], \\
    \check{L}_{\text{id}} &= -\frac{\kappa_z}{2}\one_3.
\end{align*}
Notice that, in this case, the Hamiltonian and noise operators expressed in the basis defined by $U$ are all block-diagonal with respect to the algebra block structure. This implies that there is no flow of information from and into the block $\Hc_{F}$, which can then be deleted in the reduction. Thus finally, the reduced model is given by 
\begin{align}
\hspace*{-2mm}\begin{cases}
    \dot{\check{\rho}}(t) &= -i[\check{H},\check{\rho}(t)] + \sum_{\check{L}} \Dc_{\check{L}}[\check{\rho}(t)]\\
    \tau(t) &= \tfrac{1}{2}\sum_{q=0,x,y,z} \sigma_q\, \tr[{\check{O}}_q \check{\rho}(t)]
\end{cases},\;\,\check{\rho}_0 = \Rc(\rho_0).
\label{nsred}
\end{align}
In Fig.\,\ref{fig:encoding_evolution}, representative evolutions of the reduced and the original models are simulated and compared.
It can be seen that the trajectories of the two models perfectly overlap in all cases, as desired.

We conclude the analysis of this model by showing that it exhibits ODS operators that do {\em not} correspond to (weak) symmetries. From the form of the Hamiltonian in Eq.\,\eqref{HamQF}, we have that $[H,S]=0$ for any unitary matrix $S$ of the form $S = U [(\one_3\otimes\check{S})\oplus\one_2] U^\dag$, where $\check{S}\in\Cb^{2\times2}$ is also unitary. Hence, any such unitary is a symmetry for $H$. Nonetheless, due to the form of the noise operators in Eq.\,\eqref{NoiseQ}, unitaries $\check{S}$ that do not commute with $\sigma_x$ are {\em not} symmetries for the dissipative part of the dynamics: in fact, we have $S \Dc^\dag(X) S^\dag \neq \Dc^\dag( S X S^\dag),$ where $\Dc^\dag (X)\equiv \Dc_{L_z^{\text{id}}}^\dag(X)  +\sum_{j=1}^3 \Dc_{L_{z,j}}^\dag(X)$, unless $X\in\Os$. This shows that the operator $U[(\one_3\otimes\sigma_z)\oplus\zero_2]U^\dag$ is an ODS, but not a symmetry. 

\section{Conclusions and outlook}
\label{sec:end}

We have presented a comprehensive framework for {\em exact dimensional MR} of many-body open quantum dynamics described by a time-independent Lindblad master equation. Our approach leverages the fact that, in practical applications, the focus is on simulating only the evolution of a subset of observables of interest, starting from a known (linear) set of admissible initial conditions, and may be used to derive either 

(i) a {\em linear} model of provably minimal dimension, 

\noindent 
or 

(ii) a reduced {\em quantum} model of provably Lindblad form,

\noindent
when such reduced models exist. Mathematically, a key step is to lift the use of Krylov subspaces, widely employed in system theory to study reachability and observability properties of dynamical systems, to operator subspaces or associative algebras -- as demanded for reduction of general quantum statistical dynamics. 

For both the linear and the quantum settings, the achievable amount of reduction  
depends on many factors, in particular, the existence of strong or weak symmetries in the dynamics and the choice of the initial states. In addition to providing constructive procedures for numerically determining the different types of reductions that are possible, we have analyzed in depth a number of analytically tractable examples, in order to both gain some concrete insight on the implementation of the proposed methods and showcase their applicability and performance. These examples include both many-body systems (spin chains) subject to Markovian dissipation, and systems coupled to a structured environment (a central spin in contact with an interacting spin bath, subject in turn to Markovian dissipation). In this case, while by construction the MR at the joint system-bath level yield a Lindblad master equation, the reduced dynamics of the system (central spin) alone may include both Markovian and non-Markovian components in general. Importantly, unlike in standard approaches, non-factorized system-bath initial conditions may be easily accommodated if desired. As expected, linear reduced models may offer a larger degree of reduction, at the cost of sacrificing CPTP constraints. Again, however, we stress that preservation of these quantum constraints is a key and, to the best of our knowledge, {\em unique} advantage that our approach affords in comparison with existing methods for reducing the complexity of Lindbladian dynamics. Altogether, our examples convincingly demonstrate how the MR tools we provide may be used to investigate phenomena of interest in non-equilibrium many-body physics, ranging from decoherence behavior to quantum transport and encoded information dynamics in the presence of unintended error components. 

Special emphasis in our analysis has been given to elucidate the role and usefulness of symmetries in enabling quantum MR. On the one hand, we have made precise the sense in which a strong symmetry of the Lindbladian results in a larger computational advantage than a weak one, with the improvement in system-size scaling being exponentially larger in principle (as we explicitly saw in a boundary-driven vs.\,an isolated or purely-dephasing XXZ chain). On the other hand and most importantly, the introduction of a new class of {\em observable-dependent symmetries} has allowed us to go beyond standard symmetry-based reductions, and show that exact MR is possible even in the absence of weak symmetries, as long as observable-dependent symmetries emerge for the observables of interest.
It is natural to ask whether, in a dual manner, one could also introduce \textit{state-dependent symmetries}, as unitary conjugations that only commute with the generator $\Lc$ when applied to certain initial conditions. This cannot be enough, however, to account for the reachable-based MR introduced in Sec.\,\ref{sec:reachable}, as even state-dependent symmetries would lead to projections onto standard unital algebras corresponding to their commutators. It would then remain to explain, from a symmetry-based viewpoint, how distorted algebras emerge. This is a first problem we leave for future work. 

Several other related research questions are naturally prompted by the present investigation, some of which we highlight below.

(i) While we have considered the time-dependent expectations of the observable of interest as the quantities of interest for the simulation, our methods can be adapted to include the effect of the conditioning (state collapse) due to direct or indirect measurements, following the ideas presented for the discrete-time case in \cite{conditional}. This would be especially interesting for obtaining reduced models to describe {\em continuously monitored} many-body quantum systems, which are being intensively investigated in the context of measurement-induced entanglement phase transitions \cite{Fuji,Schiro}. Likewise, since general time-local TP master equations have recently been shown to admit an unraveling in terms of a Markov process \cite{Donvil},  this may also pave the way for extending MR to continuous-time systems described by genuinely {\em non-Markovian, time-local} master equations. 

(ii) Extensions to {\em parametrized families} of Lindbladian generators may be contemplated within our operator-algebraic framework, with appropriate modifications. Characterizing the level of MR robustly achievable for families of generators may enable a more complete understanding of Hilbert-space and operator-space fragmentation in open quantum systems, beyond the partial results available to date \cite{Essler2020,li2023hilbert}. Similarly, extensions to incorporate {\em time-dependent control} would be desirable, for instance to assess the effectiveness of methods such as dynamical decoupling in reducing decoherence. While combining Lindblad generators (e.g., adding an explicitly time-dependent Hamiltonian) is only possible under restrictive conditions \cite{Brask}, our approach would remain viable for structured environments, with Hamiltonian control applied on the system side at the level of the joint unitary dynamics, prior to MR. Such an extension to constructing exact reduced controlled models would complement, for instance, the (approximate) approach of \cite{goldschmidt2021bilinear}, which allows for quantum control but, as noted, does not retain the quantum structure in general, or approaches based on cluster-expansion techniques as in \cite{onizhuk2023understanding}, which are also approximate and purely numerical. 

(iii)  We have focused here on finite-dimensional systems, as they allow direct numerical implementation and validation of the proposed MR protocols using linear-algebraic tools. Nonetheless, the underlying ideas and methods would be well-worth extending to the {\em infinite-dimensional setting}, at the cost of added mathematical complexity \cite{Remi}. This would open up applications to the large body of open quantum systems involving bosonic modes, which are of central relevance, for instance, to circuit quantum electrodynamics \cite{Blais}.

(iv) While our reduced models are guaranteed to be in Lindblad form, they need not in general exhibit well-defined {\em locality} properties, as desirable on physical grounds. Adapting our framework so that locality constraints on observables or states (e.g., decay of correlations) are imposed in the MR procedure is a challenging problem well-worth investigating. In connection to the extension to state-dependent symmetries mentioned above, this may be possible to establish connections with operator-algebraic approaches to non-equilibrium many-body dynamics \cite{BucaPRX}, whereby exact solutions for the behavior of arbitrary local quantities of interest at all times are shown to be determined by knowledge of a smaller set of special ``pseudo-local'' quantities. Reminiscent of our state-dependent symmetries, such operators need {\em not} obey the requisite symmetry requirements for all initial states.

(v) While exact MR as we considered is particularly appealing for both rigorous analysis it may be too stringent for applications in many realistic problems. This calls for our proposed framework to be adapted to include different kinds of {\em approximate} MRs. In this venue, a first way of relaxing the exact setting would be to require that properties of interest are exactly reproduced only {\em asymptotically} in time, rather than along the full trajectory. This would suffice in applications where only SS properties of the system are of interest, as in many quantum transport or SS criticality problems \cite{BeiZ} or where non-trivial, non-stationary dynamics emerges within the system's center manifold in the long-time limit \cite{BucaNSD}. 
In moving to approximate MR, special care must be put in selecting an effective trade-off between accuracy and size of the target reduced model, while ensuring that it remains physically admissible. In this context, it would be interesting to determine whether modifications of the present methods may return a valid quantum model in settings where adiabatic elimination provably violates the CP property \cite{Cyril}.  We look forward to report on some of these issues in future studies. 

\section{Acknowledgments}
The authors are grateful to Augusto Ferrante, Vincent P. Flynn, and Michiel Burgelman for useful discussions, and to Miguel Casanova for independently verifying some of the numerical simulations. T.G. would also like to thank Alain Sarlette and Michael M. Wolf for useful comments on the general topics of this work. T.G. was supported in part by a scholarship provided by the Fondazione Ing. Aldo Gini and is thankful to Dartmouth College for the generous hospitality during the initial stages of this research. Work at Dartmouth was supported in part by the U.S. Army Research Office under grant No. W911NF2210004. F.T. was supported by the European Union through NextGenerationEU, within the National Center for High Performance Computing, Big Data and Quantum Computing under Projects CN00000013, CN 1, and Spoke 10. 


\clearpage

\onecolumn

\section*{Appendix}
\appendix

\section{Supplementary theory results}

\subsection{Technical proofs}
\label{appendix:supplementary_results}

\begin{proposition*}[\ref{prop:reachable_characterization}]
Given a QSO $(\Lc,\Oc,\Sf)$ as in Eq.\,\eqref{eqn:QHM_model_ct}, defined on a subalgebra $\As\subseteq \frak{B}(\Hc),$ with $n=\dim(\Hc)$, the reachable space can be computed as
    \begin{equation}
        \Rs = \Span\{\Lc^i(\rho_0),\, i=0,1,\dots,n^2-1,\,\rho_0\in\Sf\}.
    \end{equation}
    Moreover, $\Rs$ is the smallest $\Lc$-invariant (and $e^{\Lc t}$-invariant) subspace containing $\Span\{\Sf\}$.
\end{proposition*}

\begin{proof}
Let $\Ss\equiv\Span\{\Sf\}$, and let us also define another subspace, $\Rs_t$, which is the set of states that are reached at time $t$, i.e. $\Rs_t=\Span\{e^{\Lc t}(\rho_0), \rho_0\in\Sf\}$. We can then observe that the subspace $\Rs_t$ coincides with the image of the super-operator $e^{\Lc t}\Pi_{\Ss}$ at all times $t$, where $\Pi_\Ss$ is a projector onto $\Ss$. Recalling that ${\rm Im}(e^{\Lc t}\Pi_\Ss)=\ker(\Pi_\Ss^\dag e^{\Lc^\dag t})^\perp$, we can characterize $\Rs_t$ by characterizing $\ker(\Pi_\Ss^\dag e^{\Lc^\dag t})$. 
From the definition of the exponential map, $X\in\ker(\Pi_{\Ss}^\dag e^{\Lc^\dag t})$ if and only if $\sum_{i=0}^{+\infty}\Pi_\Ss^\dag \Lc^{\dag i}(X)\frac{t^i}{i!}= 0$, which is true if and only if $\Pi_{\Ss}^\dag\Lc^{\dag i}(X)=0$ for all $i=0,1,\dots$. From the Cayley-Hamilton theorem, we then have that $\Pi_{\Ss}^\dag\Lc^{\dag i}(X)=0$ for all $i=0,1,\dots,n^2-1$, if and only if it is equal to zero for all $i=0,1,\dots$. This directly implies that $X\in\Rs_t$ if and only if $X\in\Span\{\Lc^i(\rho_0), \, i=0,1,\dots,n^2-1, \, \rho_0\in\Sf\}$, which is independent of $t$. This means that $\Rs = \oplus_{t\geq0} \Rs_t$ is exactly $\Rs=\Span\{\Lc^i(\rho_0), \, i=0,1,\dots,n^2-1, \, \rho_0\in\Sf\}$. 

The fact that $\Rs$ is $\Lc$-invariant and that it contains $\Span\{\Sf\}$ follows trivially from this characterization of the reachable space and the Cayley-Hamilton theorem. To verify that $\Rs$ is also the smallest such subspace, consider a subspace $\Vs$, which is $\Lc$-invariant and that contains $\Ss$. Then $\Vs$ must also contain $\Lc^k\Ss$ for all $k\geq 0$, and thus must contain $\Rs$. Finally, to see that $\Rs$ is also $e^{\Lc t}$-invariant for all $t\geq 0$, it is sufficient to observe that 
$e^{\mathcal{L} t}(\Rs) = \sum_{k=0}^{\infty}\frac{t^k}{k!} 
\mathcal{L}^k(\Rs) 
\subseteq\Rs.$
\end{proof}

\begin{theorem*}[\ref{thm:linear_reachable_model_reduction}]
Consider a QSO $(\Lc,\Oc, \Sf)$, defined on a sub-algebra $\As\subseteq\Bf(\Hc)$, and its reachable subspace $\Rs$. Let $\Vs$ be an operator subspace that contains the reachable space, i.e., $\Rs\subseteq\Vs$, with $\Pi_\Vs$ denoting a (non-necessarily orthogonal) projector superoperator on $\Vs$. Let $\rs$ and $\es$ be two full-rank factors  of $\Pi_\Vs$, i.e., $\es\rs = \Pi_\Vs$, and $\rs\es = \Ic_{\Vs}$, and define $\Fc_L \equiv \rs\mathcal{L}\es,$ $\Oc_L=\Oc\Jc$ and $\Sf_L \equiv \rs(\Sf)$. We then have 
\begin{equation}
    e^{\Lc t}(\rho_0) = \es e^{\Fc_L t}\rs (\rho_0), \quad  \forall t\geq0, \forall \rho_0\in\Sf.
    \label{eq:Rreduction_app}
    \end{equation} 
Moreover, $\Vs = \Rs$ is an operator subspace of \emph{minimal} dimension for which Eq.\,\eqref{eq:Rreduction_app} holds.
\end{theorem*}
\begin{proof}
Let us start by observing that for all $\rho_0$ and for all $t\geq0$
\begin{align*}
    \es e^{\rs\Lc\es t}\rs(\rho_0) &= \sum_{k=0}^{+\infty}\frac{t^k}{k!} \es (\rs \Lc \es)^k\rs(\rho_0) 
    = \sum_{k=0}^{+\infty}\frac{t^k}{k!} (\Pi_\Vs \Lc \Pi_\Vs)^k\Pi_\Vs(\rho_0)
    = e^{\Pi_\Vs\Lc\Pi_\Vs t}\Pi_\Vs(\rho_0)
\end{align*}
where we used the definition of the exponential map and the property of the isometry $\es\rs = \Pi_\Vs$.
Let then proceed by considering $\Pi_\Rs$ to be a (non-necessarily orthogonal) projector onto $\Rs$ (which might be orthogonal for a different inner product than the one that makes $\Pi_\Vs$ orthogonal) and let us recall three key facts from Proposition \ref{prop:reachable_characterization} and from the hypothesis. First $\Sf\subseteq\Rs$ and hence $\Pi_\Rs(\rho_0) = \rho_0$ for all $\rho_0\in\Sf$. Second, $\Rs$ is $\Lc$-invariant and hence $\Lc\Pi_\Rs = \Pi_\Rs \Lc \Pi_\Rs$ and $\Lc^{k}\Pi_\Rs = (\Pi_\Rs\Lc\Pi_\Rs)^k\Pi_\Rs$ for all $k\in\Nb$. Third, $\Rs\subseteq\Vs$ and thus $\Pi_\Vs\Pi_\Rs = \Pi_\Rs$. Moreover, combining the last two observations we have $\Pi_\Vs\Lc\Pi_\Vs\Pi_\Rs = \Pi_\Vs\Lc\Pi_\Rs  = \Pi_\Vs\Pi_\Rs\Lc\Pi_\Rs= \Pi_\Rs\Lc\Pi_\Rs$ and thus $(\Pi_\Vs\Lc\Pi_\Vs)^{k}\Pi_\Rs = (\Pi_\Rs\Lc\Pi_\Rs)^k\Pi_\Rs$.

Let then consider the right-hand side of the statement, i.e., $e^{\Lc t}(\rho_0)$. Applying the first two facts we just recalled and the definition of the exponential super operator we obtain:
\begin{align*}
    e^{\Lc t}(\rho_0) &= e^{\Lc t}\Pi_\Rs(\rho_0) = \sum_{k=0}^{+\infty} \frac{t^k}{k!}\Lc^{k}\Pi_\Rs(\rho_0) 
    = \sum_{k=0}^{+\infty} \frac{t^k}{k!}(\Pi_\Rs\Lc\Pi_\Rs)^{k}\Pi_\Rs(\rho_0) 
    = e^{\Pi_\Rs\Lc\Pi_\Rs t}\Pi_\Rs(\rho_0),
\end{align*}
for all $t\geq0$ and $\rho_0\in\Rs$ and hence, in particular for all $\rho_0\in\Sf$.
Similarly, for the left-hand side of the statement, we obtain:
\begin{align*}
    e^{\Pi_\Vs\Lc\Pi_\Vs t}\Pi_\Vs(\rho_0) &= e^{\Pi_\Vs\Lc\Pi_\Vs t}\underbrace{\Pi_\Vs\Pi_\Rs}_{\Pi_\Rs}(\rho_0) 
    = \sum_{k=0}^{+\infty} \frac{t^k}{k!}(\Pi_\Vs\Lc\Pi_\Vs)^{k}\Pi_\Rs(\rho_0) \\
    &= \sum_{k=0}^{+\infty} \frac{t^k}{k!}(\Pi_\Rs\Lc\Pi_\Rs)^{k}\Pi_\Rs(\rho_0) 
    = e^{\Pi_\Rs\Lc\Pi_\Rs t}\Pi_\Rs(\rho_0) ,
\end{align*}
for all $t\geq0$ and for all $\rho_0\in\Rs$ and hence, in particular for all $\rho_0\in\Sf$. This proves the first statement.

At last, we will prove the minimality of choosing $\Vs=\Rs$ by contradiction. Assume that there exists a smaller subspace $\Vs^*\subset\Rs$ such that 
$e^{\Lc t}(\rho_0) = e^{\Pi_{\Vs^*}\Lc\Pi_{\Vs^*}t}\Pi_{\Vs^*}(\rho_0)$ for all $\rho_0\in\Sf$ and $t\geq0$.
Let us denote with $\Rs\ominus\Vs^*$ the subspace of $\Rs$ such that $\Rs = \Vs^* \oplus (\Rs\ominus\Vs^*)$.
Then, from the definition of $\Rs$, there exists $\rho_0$ and $t_1$ such that $e^{\Lc t_1}(\rho_0) = X+Y$ where $X\in\Vs^*$ and  $Y\in \Rs\ominus\Vs^*$, with $Y\neq0$. Moreover, there also exist a time $t_2$ such that $e^{\Lc t_2}(Y)=0$. 
On the other hand, from the assumption that both $\Vs^*$ and $\Rs$ lead to effective reductions, we have that 
$(e^{\Pi_{\Rs}\mathcal{L}\Pi_{\Rs} t_2} - e^{\Pi_{\Vs^*}\mathcal{L}\Pi_{\Vs^*}{t_2}} ) \circ e^{\mathcal{L}{t_1}}(\rho_0) = 0. $
Using $e^{\mathcal{L}{t_1}}(\rho_0)=X+ Y,$ and the fact that $\Pi_{\Rs}(X)=\Pi_{\Vs^*}(X)=X$ while $\Pi_{\Vs^*}(Y)=0$ and $\Pi_{\Rs}(Y)=Y$, we obtain
\begin{align*}
    0&=(e^{\Pi_{\Rs}\mathcal{L}\Pi_{\Rs}{t_2}} - e^{\Pi_{\Vs^*}\mathcal{L}\Pi_{\Vs^*}{t_2}}) (X+Y)
    =e^{\Pi_{\Vs}\mathcal{L}\Pi_{\Vs}{t_2}}(Y)=e^{\mathcal{L}{t_2}}(Y),
\end{align*}
where in the last expression we used the fact that $Y\in{\Rs}.$ So we proved that if such ${\Vs^*}\subsetneq \Rs$ existed, $e^{\mathcal{L}{t_2}}(Y)$ should be both equal and different from zero, reaching an absurd. 

The fact that the reachable subspace $\Rs$ provides the operator subspace of minimal dimension for which \eqref{eq:Rreduction} holds is a consequence of known results from control-system theory \cite{wonham, kalman1969topics}.
\end{proof}

\begin{lemma}
Consider a bipartite Hilbert space $\Hc =\Hc_S\otimes\Hc_F$ and a positive-definite operator $\tau\in\Hf_{>}(\Hc_F)$.
Then, for all $Q\in\Bf(\Hc)$ and $A\in\Bf(\Hc_S)$, the following properties hold:
\begin{align*}
    \tr_{F}[Q (A\otimes \tau)] = \tr_{F}[Q (I_S\otimes\tau)]A, \qquad 
    \tr_{F}[(A\otimes\tau) Q] = A \tr_{F}[(I_F\otimes\tau) Q].
\end{align*}
\end{lemma}
\begin{proof}
Let us start by observing that, since $\tau\in\Hf_>(\Hc_F)$, we can write $\tau = \sum_k w_k \ketbra{\phi_k}{\phi_k}$ with $w_k>0$ for all $k$ and $\{\ket{\phi_k}\}$ an orthonormal base for $\Hc_F$. Moreover, since $\tr_F(\cdot)$ is CP, it allows for a Kraus form, e.g. $\tr_F(\cdot) = \sum_l M_l \cdot M_l^\dag$ with $M_l = I_S\otimes \bra{\psi_l}$ where $\{\ket{\psi_l}\}$ can be any base for $\Hc_F$ and, in particular, we here pick $\ket{\psi_l}=\ket{\phi_l}$. Similarly, the operation $\cdot\otimes\tau:\Bf(\Hc_S)\to\Bf(\Hc)$ is also CP and hence it allows for a Kraus form $\cdot\otimes\tau = \sum_k S_k \cdot S_k^\dag$ with $S_k = \sqrt{w_k} I_S\otimes\ket{\phi_k}$.  
Now, substituting these two Kraus forms into $\tr_{F}[Q (\tau\otimes A)]$ we obtain:
\begin{align*}
    \tr_{F}[Q (\tau\otimes A)] 
    &= \sum_{l} M_l Q (\tau\otimes A)  M_l^\dag
    = \sum_{l,k} M_l Q S_k A S_k^{\dag}  M_l^{\dag}. 
\end{align*}
Now, we can observe that $S_k^\dag M_l^\dag = \sqrt{w_k}(I_S\otimes\ket{\phi_k})^\dag (I_S\otimes \bra{\phi_l})^\dag =\sqrt{w_k} (I_S\otimes\bra{\phi_k})(I_S\otimes\ket{\phi_l}) = \sqrt{w_k}I_S\braket{\phi_k}{\phi_l} = \sqrt{w_k} I_S \delta_{k,l}$ and thus $\tr_{F}[Q (\tau\otimes A)]= \sum_{l} \sqrt{w_l} M_l^{\dag} Q S_l A$.
On the other hand, using the same techniques one can verify that $\tr_{F}[Q (I_S\otimes\tau)]A =\sum_{l} \sqrt{w_l} M_l^{\dag} Q S_l A.$
The second equation of the statement is proved in the same manner.
\end{proof}

\begin{lemma}
Let us consider a Hilbert space comprising a direct sum, $\Hc=\bigoplus_k\Hc_k$. Consider the algebra $\As$ with structure $\As = \bigoplus_k \Bf(\Hc_k)$ and two operators in it, i.e. $A,B\in\As$. Then 
$$ AB = \sum_{k=0}^{K-1} V_k^\dag A_kB_k V_k = \bigoplus_k A_kB_k,$$
where $A_k,B_k\in\Bf(\Hc_k)$ and $V_k$ are such that $V_k:\Hc\to\Hc_k$ and $V_kV_k^\dag = I_k$.
\end{lemma}
\begin{proof}
The proof follows trivially from the block-diagonal structure of the operators $A$ and $B$.
\end{proof}

\begin{proposition}
\label{prop:combining_reduction_injection}
Let consider a $*$-subalgebra $\mathscr{A}$ of $\mathcal{B(H)}$ with Wedderburn decomposition $\As = U\left(\bigoplus_\ell \Bf(\Hc_{S,\ell})\otimes \one_{F,\ell} \right)U^\dag \simeq \bigoplus_\ell \Bf(\Hc_{S,\ell}) =: \check{\As}$. Let then $\rs$ and $\es$ be the CPTP factorization of the state extension $\SE_\mathscr{A} = \es\circ\rs$ as defined in Eq.\,\eqref{eqn:reduction} and \eqref{eqn:injection}. Then, for all $A\in\check{\As}$ and for all $X\in\Bf(\Hc)$, we have  
\begin{equation}
\rs[X \es(A)] = \es^\dag(X) A, \qquad 
\rs[\es(A) X] = A \es^\dag(X).
\end{equation}
\end{proposition}
\begin{proof}
Let us start by recalling the definitions of $\rs$ and $\es$ from Eq.\,\eqref{eqn:reduction} and \eqref{eqn:injection} for convenience:
\begin{align*}
    \es(A)=\sum_{k=0}^{K-1} W_k^\dag\left( V_kAV_k^\dag\otimes \tau_k\right) W_k, &&
\rs(X)=\sum_{k=0}^{K-1}V_k^\dag\tr_{F,k}[W_k X W_k^\dag]V_k
\end{align*}
for all $X\in\Bf(\Hc)$ and $A\in\check{\As}$. 
We can then substitute the definitions of $\rs$ and $\es$ into $\rs( X\es(A))$ to obtain:
\begin{align*}
\rs[X\es(A)] &= \sum_{k=0}^{K-1}  \mathcal{R}[X W_k^\dag (A_k\otimes\tau_k) W_k]
    = \sum_{l,k=0}^{K-1}  V_l^\dag\tr_{F,k} \bigg[W_l X W_k^\dag (A_k\otimes\tau_k) \underbrace{W_k W_l^\dag}_{\one_{n_k}\delta_{k,l}} \bigg]V_l\\
    &= \sum_{k=0}^{K-1}  V_k^\dag \tr_{F,k}\left[W_k X W_k^\dag (A_k\otimes\tau_k) \right] V_k 
    = \sum_{k=0}^{K-1}  V_k^\dag\tr_{F,k}\left[W_k X W_k^\dag(\one_{m_k}\otimes\tau_k)\right] A_k V_k\\
    &=   \underbrace{\sum_{k=0}^{K-1} V_k^\dag\tr_{F,k}[(\one_{m_k}\otimes\tau_k)W_k X W_k^\dag]V_k}_{\es^\dag(X)=\ro(X)} A  ,
\end{align*}
where we used the fact that $W_kW_j = \one_{n_k}\delta_{j,k}$ and the two lemmas above. The other equality can be proved in an analogous manner.
\end{proof}

\begin{proposition}
\label{prop:double_commutant}
Let us consider a set of operators $\Phi$ containing the identity, i.e. $\one\in\Phi$. Then $\Phi'' = \alg(\Phi)$. 
\end{proposition}
\begin{proof}
We will prove the proposition by proving the equivalent relation $\Phi' = (\alg(\Phi))'$, which follows from the Von Neumann bi-commutant theorem, see e.g., \cite[Sec.\,I.9.1]{blackadar2006operator}.  Then, from the definition of $\alg(\Phi)$, it follows that $\Phi\subseteq\alg(\Phi)$, which in turn implies that $\Phi' \supseteq (\alg(\Phi))'$.

It then remains to prove the other inclusion, i.e., $\Phi'\subseteq (\alg(\Phi))'$. This is equivalent to prove that any element that commutes with every element in $\Phi$ must also commute with every element in $\alg(\Phi)$. Take $X\in\Phi'$ we have, from the definition of commutant $[X,Y]=0$ for all $Y\in\Phi$. Then, since any element in $\alg(\Phi)$ can be written as a linear combination and products of elements in $\Phi$ and, since $[X,\alpha Y+\beta Z]=0$ and $[X,YZ] = XYZ - YZX = YXZ - YZX = YZX -YZX = 0$, for all $\alpha,\beta\in\Rb$ and for all $Y,Z\in\Phi$, we have that $[X,Y]=0$ for all $Y\in\alg(\Phi)$.
\end{proof}

\begin{theorem*}[\ref{maxsymm}]
Let $\{O_i\} \subset \Hf (\Hc)$ be a set of observables, and let $\Lc$ be a Lindblad generator. Then $\alg\{\Ns^\perp\}\subsetneq\Bf(\Hc)$ if and only if there exists a non-trivial $\{O_i\}$-ODS for $\Lc$. Furthermore, we have 
$$\Os=\alg\{\Ns^\perp\}=\Cb\bar\Gs',$$ 
where $\bar\Gs$ is the largest group of ODS for the system.
\end{theorem*}
\begin{proof}
If $\Gs$ is an ODS then, from Corollary \ref{generalcon}, we know that $\alg\{\Ns^\perp\}  \subseteq \Cb\Gs'$. Under the assumption that $\Gs$ is non-trivial, meaning it contains elements other than the identity, the commutant of $\Gs$ must be a proper subalgebra of $\Bf(\Hc)$; hence we have established the backward direction.

For the forward direction, assume  $\alg\{\Ns^\perp\}\subsetneq\Bf(\Hc)$. This implies that the observable algebra has a non-trivial Wedderburn decomposition, say:
\begin{equation*}
\alg(\Ns^\perp) = U\bigg(\bigoplus_k \Bf(\Hc_{F,k})\otimes \one_{G,k}\bigg)U^\dag ,
\end{equation*}
where $U$ is the unitary transformation of basis. In particular, generators of the algebra all have block structures of the form,
\begin{equation*}
\Lc^{\dag n}(O_i)=U\bigg(\bigoplus_k O_{n,i,k}\otimes \one_{G,k}\bigg)U^\dag,\quad O_{n,i,k}\in\Bf(\Hc_{F,k}).
\end{equation*}
Note that, since the operators $\Lc^{\dag n}(O_i)$ are the generators for $\alg(\Ns^\perp)$, it must hold that,
\begin{equation}
\label{complete}
    \alg\{O_{n,i,k}\}=\Bf(\Hc_{F,K}).
\end{equation}
Let us then denote by $\Uf(\Hc)$ the unitary group of operators in $\Bf(\Hc)$. Then, the group
\begin{equation}
\label{maximalgroup}
\bar\Gs \equiv U\bigg(\bigoplus_k \one_{F,k}\otimes 
\Uf(\Hc_{G,k})\bigg)U^\dag,
\end{equation}
is a group of ODS. Indeed, any unitary group that commutes with all $\Lc^{\dag n}(O_i)$ must be of the form,
\begin{equation*}
\Gs=T\bigg(\bigoplus_k \mathcal{U}_{F,k}\otimes \mathcal{U}_{G,k}\bigg)T^\dag,
\end{equation*}
where $\mathcal{U}_{F,k}$ and $\mathcal{U}_{G,k}$ are unitary subgroups of $\Uf(\Hc_{F,k})$ and $\Uf(\Hc_{F,k})$, respectively. Also, $[\mathcal{U}_{F,k},O_{ni,k}]=0$ for all $n,i$, which implies that $\mathcal{U}_{F,k}$ lies in the commutant to the algebra generated by $O_{ni,k}$. However, due to \eqref{complete} this commutant is trivial, therefore, $\mathcal{U}_{F,k}= \{\one_{F,k}\}$. We see that \eqref{maximalgroup} is indeed the maximal group of ODS.

Notice that in the above decomposition, it follows from Proposition\,\ref{prop:double_commutant} that ${\mathbb{C}}\bar\Gs'=(\Ns^\perp)''=\alg(\Ns^\perp)$. This implies that reducing to the commutant of the ODS $\Cb\bar\Gs'$ is equivalent to reducing to the output algebra $\Os$.
\end{proof}

\subsection{Algorithm complexity}
\label{app:complexity}

Here, we discuss the implementation and the complexity of the proposed MR algorithms. We focus first on the observable-based linear and quantum MR algorithms: the linear one is covered by the first three steps described below, while to obtain a quantum reduction one has to further proceed by computing minimal algebras and their block decompositions. As in the main text, let $\dim(\Hc)=n<\infty$ and consider a set of $m$ linearly independent observables of interest.

\vspace*{-1mm}

\begin{enumerate}

	 \item In order to compute the generators of $\Ns^\perp$, using Eq.\,\eqref{eqn:non_obs_Heisenberg_2}, we need to perform the product between a $n^2\times n^2$ matrix and a $n^2$ vector, a number of times equal to $n^2-1$ for each observable. This implies a worst-case complexity of $O(mn^8)$.

\vspace*{-2mm}

    \item In order to find an orthogonal basis for $\Ns^\perp,$ we can use Gram-Schmidt orthogonalization procedure on the $mn^2$ vectors of dimension $n^2$, giving a complexity for this step of $O(m^2n^6)$, \cite{golub2013matrix}.

    \item In practice, the two above steps can be combined to obtain a more efficient algorithm in the more relevant case $\Ns^\perp\subsetneq\Bf(\Hc).$ This follows from the fact that, if we define  
    $$\Ns^\perp_f \equiv \Span\{\Lc^{\dag,k}(O_j), k=0,\dots,f,\,\forall j\},$$ 
    we have $\Ns^\perp_{f-1}\subseteq \Ns^{\perp}_{f}\subseteq\Ns^\perp_{f+1}\subseteq {\Ns^\perp} $. 
    Furthermore, one can prove that whenever $\dim(\Ns^\perp_{f-1}) = \dim(\Ns^\perp_{f})$ for some $f$, then $\dim(\Ns^\perp_{f}) = \dim(\Ns^\perp_{q})$ for any $q\geq f$, \cite{kalman1969topics}. {This implies that, instead of computing $mn^2$ vectors and then performing the orthogonalization procedure on all of them one can, at each step, compute $m$ new vectors (which requires a complexity of $O(mn^6)$) and perform the orthogonalization only on these new vectors (which requires a complexity of $O(mn^2)$), discarding all the vectors that are linearly dependent while keeping the others. The algorithm stops when a new iteration does not add any new vector, thus repeating the process $f$ times. This reduces the complexity of these two steps from $O(mn^8)$ to $O(fmn^6)$. Clearly, $f$ cannot be known {\em a priori} and could be as large as $n^2$, but this proves to be a more efficient implementation in common cases, as one usually stops for $f<n^2$.} 

\vspace*{-2mm}
    \item For quantum MR, it is also necessary to compute (4a) the observable algebra $\Os$ and (4b) its Wedderburn decomposition. 

        \vspace*{-2mm}
\begin{itemize}
    \item[(4a)] Since the Wedderburn decomposition of $\Os$ and its commutant $\Os'$ are complementary, one can decide to compute one or the other. To compute the algebra $\Os$ from $\Ns^\perp$, one can adapt the standard algorithms for the closure of a Lie algebra, upon substituting the binary operation $[\cdot,\cdot]$ with the matrix product. Such an approach has a complexity of approximately $O(n^8)$ \cite{zeier2011symmetry}, although more efficient numerical algorithms are being actively investigated \cite{iiyama2025fast}. On the other hand, to compute $\Os'$, assuming that $\dim(\Ns^\perp) = d$, we have to compute the null space of a matrix of size $dn^2\times n^2$, which implies a complexity of $O(d n^6)$ using standard SVD techniques \cite{park2023fast}.

\vspace*{-2mm}
    \item[(4b)] The computation of the Wedderburn decomposition requires the computation of the SVD of generic operators in $\Os$ (or $\Os'$) and thus has a complexity of $O(n^3)$. We also refer the interested reader to \cite{EBERLY200435,murota2010numerical,de2011numerical} for more in-depth discussion.
    \end{itemize}
    
\end{enumerate}

    A complete algorithm for linear ME consists of first performing the linear observable-based MR, followed by the reachable-based MR -- {\em without} the need to iterate the process; see \cite{tit2023,kalman1969topics}. Accordingly, based on the above estimates (first three steps), the linear MR algorithm has an overall complexity of $O(n^8)$.

    For quantum MR, the overall complexity of the Algorithm \ref{algo:composed} still combines the complexities of the reachable-based and observable-based MR. Unlike the linear reduction, however, the (finite) number of of times that the two subroutines need to be repeated remains an open problem at this time; see \cite{tit2023}. Based on the above estimates, a single run of the algorithm proposed in Sec.\,\ref{sec:observablered} for observable-based quantum MR has a complexity of the order of $O(n^8)$. Since the reachable-based quantum MR algorithm proposed in Sec.\, \ref{sec:lindblad_reduction} is virtually identical from a computational complexity point of view, it also has a complexity of the order of $O(n^8)$. Regarding the required number of iterations, a worst-case upper bound is clearly provided by the dimension of the initial operator space ($n^2$). However, at each iteration, the size of the operators we manipulate may be reduced, making this bound very conservative. For instance, it is worth stressing that, in all the examples we considered, we never had to iterate more than 3 times.

    For all the examples included in this paper, we were able to test the proposed numerical MR algorithm on a commercial laptop for systems up to $N=6$ qubits (although the code can certainly be optimized further). Despite this being a modest number of qubits, it allowed us validate the minimality of the algebra we obtained with symmetry-based considerations, as shown in Fig.\,\ref{fig:central_spin_dim} and \ref{fig:xxz_dim}. Leveraging the {\em analytical description} of the algebra given by symmetry-based considerations, we could then extend the results and compute the reduced model for larger system size, as shown e.g., in Fig.\,\ref{fig:Ising} and \ref{fig:tessieri}.

\section{Observable-based model reduction in central-spin models with collective couplings}
\label{centralspin}

\subsection{Single-axis coupling Hamiltonian}
\label{appendix:zurek}

Before discussing the more general central-spin model of Sec.\,\ref{sec:examples_central_spin}, it is worth considering the simpler yet relevant case where the system-bath interaction Hamiltonian involves a single coupling operator, which also enters the free bath Hamiltonian, say,
\begin{equation}
H_{SB} = \frac{1}{2} \Big( \omega_1  \sigma_z^{(1)} + \eta \sigma_x^{(1)} \Big)  + 
\frac{{\mu}}{2} {J}_{x}  + \frac{1}{2} A_x\sigma_x^{(1)} J_x . 
\label{genZurek}
\end{equation}
Importantly, $[H_{\text{int}}, H_B]=0$. If, in addition, $[H_{\text{int}}, H_S]= 0$ (that is, if $\omega_1=0$ in Eq.\eqref{genZurek}), the above become a purely dephasing model, as considered for instance in \cite{zurek1982environment,Dawson}. In such a special case, a well-known analytical solution exists for factorized system-bath initial conditions, and general initial conditions have been studied by using a B$+$ decomposition technique \cite{PhysRevA.100.042120}.

Given the spectral decomposition $J_x = \one_S\otimes\sum_{\ell} \lambda_\ell \ketbra{\varphi_\ell}{\varphi_\ell}$ it is possible to re-write $H_{SB}$ as
\[H_{SB} = \sum_{\ell}H_{S,\ell}\otimes\ketbra{\varphi_\ell}{\varphi_\ell}, \qquad\text{with}\qquad H_{S,\ell}\equiv\frac{1}{2}\Big(\omega_1\sigma_z + \eta\sigma_x+\lambda_\ell\mu\one_S+\lambda_\ell A_x\sigma_x\Big),\]
since $\sum_\ell \ketbra{\varphi_\ell}{\varphi_\ell} = \one_B$. 
In order to practically compute the observable-based MR described in Sec.\,\ref{sec:observablered}, we need to first compute the generators of $\Ns^\perp$,  $\ad^n_{H_{SB}}(\sigma_q^{(1)})$, with $n=0,1,\dots,2^{2N}-1$. It is then easy to verify that $\ad^n_{H_{SB}}(\sigma_q^{(1)}) = \sum_\ell \ad^n_{H_{S,\ell}}(\sigma_q)\otimes\ketbra{\varphi_\ell}{\varphi_\ell}$. This fact is already sufficient to prove that the coherences of the bath are not necessary to reproduce the evolution of $\rho_S$ for any initial condition $\rho_0\in\Df(\Hc)$ and thus the bath can be reduced to a classical Markov model of size $2^{N_B}$. However, further reduction is possible in this case.

For simplicity, let us first consider the case $\omega_1=0$. Then, we have that $\ad^n_{H_{S,\ell}}(\sigma_0) = \ad^n_{H_{S,\ell}}(\sigma_x) = 0$, $\forall \ell$ and $n>1$, while $\ad_{H_{S,\ell}}(\sigma_y) = 2i(\eta+\lambda_\ell A_x)\sigma_z$ and $\ad_{H_{S,\ell}}(\sigma_z) = -2i(\eta+\lambda_\ell A_x)\sigma_y$, $\forall \ell$. This implies that all the operators $\ketbra{\varphi_\ell}{\varphi_\ell}$ associated to the {\em same} eigenvalue $\lambda_\ell$ always appear together, hence
\begin{align*}
    \Ns^\perp &= \Span\{\one_{SB},\sigma_x^{(1)},\sigma_y\otimes\Pi_m, \sigma_x\otimes\Pi_m;\, \forall m\in[-N_B/2,\dots,N_B/2]\},\\ 
    \Os = \alg(\Ns^\perp) &= \Span\{\sigma_q\otimes \Pi_m; \forall m\in[-N_B/2,\dots,N_B/2]\},
\end{align*}
where $\Pi_m = \sum_{\ell|\lambda_\ell = m} \ketbra{\varphi_\ell}{\varphi_\ell}$ are the eigenprojectors associated to the eigenvalues $m=-\frac{N_B}{2},\dots,\frac{N_B}{2}$ of $J_x$. From this, one can easily verify that $\dim(\Ns^\perp) = 2N+2$ and $\dim(\Os)=4N$. 
Note that, as long as $\omega_1=0$, the introduction of dissipative terms on the bath -- either collective $J_q$ or local $\sigma_q^{(j)}$, along axis $q=x,y,z$ -- does not change the observable space and algebra. Further details can be found in \cite{MyThesis}.

Similar calculations can be carried out for the case $\omega_1\neq0$. Specifically, in this case we have: 
$$\ad_{H_{S,\ell}}(\sigma_0) =0, \quad \ad_{H_{S,\ell}}(\sigma_x) = 2i\omega_1 \sigma_y, \quad \ad_{H_{S,\ell}}(\sigma_y)=-2i\omega_1\sigma_x + 2i (\eta+\lambda_\ell A_x)\sigma_z, \quad \ad_{H_{S,\ell}}(\sigma_z) = -2i(\eta+\lambda_\ell A_x)\sigma_y,$$ 
which again lead to the same observable space and same observable algebra. Thus, this central-spin model can be reduced to a qubit interacting with a {\em classical} Markov model of size $N$. Specifically, let us define the reduced Hilbert space $\check{\Hc} \equiv \Hc_S \otimes \Hc_A$, with $\Hc_A\simeq\Cb^N$, a new ``surrogate'' bath. Let then $\{\ket{k}\}_{k=1}^{N}$ be the standard basis for $\Hc_A$, and $\{\ket{m}\}_{m=-N_B/2}^{N_B/2}$ a convenient relabelling $m \equiv k-1-N_B/2$,  so that $\ket{m=-N_B/2}=\ket{k=1}$ and $\ket{m=N_B/2}=\ket{k=N}$. With this, we can define the reduced state $\check{\rho} = \sum_{m} \rho_{S,m}\otimes\ketbra{m}{m}$ and reduced Hamiltonian as 
\[\check{H} = \sum_{m=-N_B/2}^{N_B/2} H_{S,m} \otimes\ketbra{m}{m},\qquad\text{with}\qquad H_{S,m}\equiv \frac{1}{2}\big[\omega_1\sigma_z + \eta\sigma_x+m(\mu\one_S+A_x\sigma_x)\big] .\]
For any initial condition $\rho(0)=\rho_0\in\Df(\Hc)$, one can compute the reduced initial condition $\check{\rho}(0)$ by computing $\rho_{S,m}(0) = \tr_E[ \rho_0 (\one_S\otimes \Pi_m)]$, for all $m=-\frac{N_B}{2},\dots,\frac{N_B}{2}$. The reduced state $\check{\rho}(t)$ then evolves according to $\dot{\check{\rho}}(t) = -i[\check{H},\check{\rho}(t)]$ and the evolution of the central spin can be retrieved as $\rho_S(t) = \sum_m \rho_{S,m}(t)$. 

\smallskip

The reduced model shows that the coherences of the bath in the initial states are not necessary to reproduce the evolution of the central spin. Furthermore, because of the structure of $\check{\rho}$ and $\check{H}$, the evolved state will never develop bath coherences and will always be in the form $\check{\rho} = \sum_{m} \rho_{S,m}\otimes\ketbra{m}{m}$. 
One can further notice that the evolution of the reduced model is equivalent to the evolution of $N$ separate qubits, each initialized in $\rho_{S,m}(0)/\tr[\rho_{S,m}(0)]$ and evolving through the Hamiltonian $\check{H}_{S,m}$. To retrieve the evolution of $\rho_S(t)$, one simply computes the weighted average of the $N$ unitary evolutions, i.e. $\sum_m \rho_{S,m}(t)$. The reduced model thus behaves as a {\em probabilistic ensemble of $N$ qubits}, each with their own evolution and for which we take the expectation value. 
For this reason, we can say that the reduced model is effectively a {\em quantum-classical hybrid} \cite{barchielli2023markovian,Dammeier2023quantumclassical}, where the central spin is the quantum component while the bath behaves as a classical Markov model.

\subsection{Generic XYZ coupling Hamilltonian}
\label{appendix:bipartition}

We now derive the reduced model presented in Sec.\,\ref{section:basecase}, that is, we consider system-bath model Hamiltonians of the form:
\begin{equation}
H_{SB} = \frac{1}{2} \Big( \omega_1  \sigma_z^{(1)} + \eta \sigma_x^{(1)} \Big) + \frac{1}{2} \Big(\overline{\omega} {J}_{z}+ \mu {J}_{x}   \Big) + \frac{1}{2} \Big( A_x\sigma_x^{(1)} J_x +  A_y\sigma_y^{(1)} J_y + A_z\sigma_z^{(1)}J_z \Big). 
\label{fullH}
\end{equation}

Let us start by making explicit the unitary change of basis $U$ that puts $\Cb\Gs_N'$ in its block-diagonal structure. Let $V\in\Bf(\Hc_B)$ be a unitary matrix whose columns are composed of the states $\ket{j,m;\alpha}$, ordered by $j, m$ and $\alpha$. Then we define 
$$U \equiv \Big(\one_2\otimes V\Big) \,K_{2^{N-1},2} \Big(\bigoplus_j K_{2,(2j+1)d_j}\Big),$$ 
where $K_{n,m}$ are tensor swap permutation matrices, that is, such that $K_{n,m}(A\otimes B)K_{n,m}^\dag = B\otimes A$ for all $A\in\Cb^{n\times n}$ and $B\in\Cb^{m\times m}$. Note that the tensor swap matrices are necessary for technical reasons: specifically, they are included due to the fact that the tensor product is distributive with respect to the diagonal sum $\oplus$ only on the left term of the product, i.e., $\oplus_k (A \otimes B_k) \neq A\otimes (\oplus_k B_k)$, while $\oplus_k (A_k \otimes B) = (\oplus_k A_k)\otimes B$.
The resulting unitary $U$ is such that
\begin{align*}
    \mathbb{C}\Gs =  U\bigg(\bigoplus_{j}  \one_{F,j}\otimes \Bf(\Hc_{G,j})\bigg)  U^\dag, \qquad 
    {{\mathbb C}\Gs'}= U\bigg(\bigoplus_{j} \Bf(\Hc_{F,j})\otimes\one_{G,j}\bigg)U^\dag, 
\end{align*}
with $\dim(\Hc_{F,j}) = 4j+2$, and $\dim(\Hc_{G,j}) = d_j$ as defined in Eq.\,\eqref{multiplicity}.
The adjoint of the non-square isometries $W_j:\Hc\to\Hc_{F,j}\otimes\Hc_{G,j}$ are then composed as 
\[W_j = \left[\begin{array}{c|c|c}\zero_{\ell_j,s_j}& \one_{\ell_j} &\zero_{\ell_j,t_j}\end{array}\right]U^\dag,\] 
where $\ell_j = 2(2j+1)d_j$ and $s_j=\sum_{k=0(1/2)}^{j-1} \ell_k$ and $t_j = \sum_{k=j+1}^{N_B/2} \ell_k$.

We now aim to express the interaction Hamiltonian in Eq.\,\eqref{fullH} in the new basis. Let us first rewrite $H_{\text{int}}$ in terms of raising and lowering operators $J_\pm\equiv(J_x\pm iJ_y)$, i.e., 
\begin{align*}
    H_{\text{int}} =\frac{1}{2}\left(A_x\sigma_x^{(1)}\frac{J_++J_-}{2} + A_y\sigma_y^{(1)}i\frac{J_- - J_+}{2} + A_z\sigma_z^{(1)}J_z\right).
\end{align*}
Recalling that $J_\pm\ket{a}\otimes\ket{j,m;\alpha}=\lambda_{j,m,\pm}\ket{a}\otimes\ket{j,m\pm1;\alpha}$, with $\lambda_{j,m,\pm}=\sqrt{(j\mp m)(j\pm m +1)}$, we can further observe that for $q=x,y$ we have: 
\begin{align*}
    U^\dag \sigma_q^{(1)}J_+ U &= \bigg(\bigoplus_j K_{2,(2j+1)d_j}^\dag\bigg) K_{2^{N-1},2}^\dag \big(\one_2\otimes V^\dag\big) \sigma_q^{(1)}J_+ \left(\one_2\otimes V\right) K_{2^{N-1},2} \bigg(\bigoplus_j K_{2,(2j+1)d_j}\bigg)\\
    &= \bigg(\bigoplus_j K_{2,(2j+1)d_j}^\dag\bigg) K_{2^{N-1},2}^\dag \bigg(\sigma_q\otimes \sum_{j,m,\alpha} \lambda_{j,m,+} \ketbra{j,m+1}{j,m} \otimes \ketbra{j,\alpha}{j,\alpha} \bigg) K_{2^{N-1},2} \bigg(\bigoplus_j K_{2,(2j+1)d_j}\bigg)\\
    &= \bigg(\bigoplus_j K_{2,(2j+1)d_j}^\dag\bigg)\bigg(\bigg[\bigoplus_{j} \check{J}_{+,j} \otimes \one_{G,j}\bigg]\otimes \sigma_q  \bigg)\bigg(\bigoplus_j K_{2,(2j+1)d_j}\bigg)\\
    &= \bigg(\bigoplus_j K_{2,(2j+1)d_j}^\dag\bigg)\bigg(\bigoplus_{j} \big[\check{J}_{+,j} \otimes \one_{G,j} \otimes \sigma_q\big] \bigg)\bigg(\bigoplus_j K_{2,(2j+1)d_j}\bigg)\\
    &= \bigoplus_{j} K_{2,(2j+1)d_j}^\dag\left[\check{J}_{+,j} \otimes \one_{G,j} \otimes \sigma_q\right] K_{2,(2j+1)d_j} 
    = \bigoplus_{j} \left[\sigma_q\otimes\check{J}_{+,j} \otimes \one_{G,j} \right].
\end{align*}
Thus, summarizing and extending to the other operators, we have:
\begin{align*}
    U^\dag \sigma_q^{(1)}J_+ U &= \bigoplus_j \sigma_q\otimes \underbrace{\bigg[\sum_{m=-j}^{j-1} \lambda_{j,m,+}\ketbra{m+1}{m} \bigg]}_{\equiv\check{J}_{+,j}} \otimes\one_{G,j}, \quad q=x,y, \\
    U^\dag \sigma_q^{(1)}J_- U &= \bigoplus_j \sigma_q\otimes \underbrace{\bigg[\sum_{m=-j+1}^{j} \lambda_{j,m,-}\ketbra{m-1}{m} \bigg]}_{\equiv \check{J}_{-,j}} \otimes\one_{G,j},\quad q=x,y, \\
    U^\dag \sigma_z^{(1)}J_z U &= \bigoplus_j \sigma_z\otimes \underbrace{\bigg[\sum_{m=-j}^j m\ketbra{m}{m}\bigg]}_{\equiv\check{J}_{z,j}} \otimes\one_{G,j}.
\end{align*}
Combining these results, one can re-write the Hamiltonian $H_{\text{int}}$ in this basis obtaining: 
\begin{equation}
    U^\dag H_{\text{int}} U = \bigoplus_j \frac{1}{2} \Big( A_x \sigma_x\otimes\check{J}_{x,j} + A_y \sigma_y\otimes\check{J}_{y,j} + A_z \sigma_z\otimes\check{J}_{z,j}\Big) \otimes \one_{G,j},
    \label{eq:Ham_in_a_new_basis}
\end{equation}
where $\check{J}_{x,j} = \frac{1}{2}(\check{J}_{+,j}+\check{J}_{-,j})$ and $\check{J}_{y,j} = \frac{i}{2}(\check{J}_{-,j}-\check{J}_{+,j})$. To compute the reduced Hamiltonian $\check{H}_{\text{int}}$ one can then resort to Corollary \ref{corollary:Hamiltonian}, with $\Jc^\dag(X) = \bigoplus_j \tr_{\Hc_{G,j}}\left[W_j X W_j^\dag\right]/\dim(\Hc_{G,j})$, obtaining \cite{tit2023}:
\[\check{H}_{\text{int}} \equiv \Jc^\dag(H_{\text{int}}) = \bigoplus_j \frac{1}{2} \Big( A_x \sigma_x\otimes\check{J}_{x,j} + A_y \sigma_y\otimes\check{J}_{y,j} + A_z \sigma_z\otimes\check{J}_{z,j}\Big).\]
In this particular case, since $H_{\text{int}}\in\Cb\Gs'$ and $\Cb\Gs'$ is a unital algebra, the result of $\Jc^\dag(H_{\text{int}})$ is equivalent to the removal of the identity terms $\one_{G,j}$ from its representation Eq.\,\eqref{eq:Ham_in_a_new_basis}, as shown in Eq.\,\eqref{reducedObHamLind}.

\section{Weak symmetry in boundary-driven XXZ models} 
\label{sec_appendix_symm}

The XXZ boundary-driven model discussed in Sec.\,\ref{sec:examples_XXZ} has been studied in multiple works, including \cite{Buca_2012,popkov2013manipulating}. Let us consider the  
continuous unitary group $\Gs_{\varphi} =e^{-i\varphi M}$ generated by the total magnetization operator. Ref.\,\cite{Buca_2012} states that $\Gs_\varphi$ is a strong symmetry group for a similar model to the one considered here (same Hamiltonian, different noise operators). However, it was not pointed out that it is also a {weak symmetry} for the same model we consider. Moreover, \cite{popkov2013manipulating} states that $\Gs_\varphi$ is a symmetry group for the non-equilibrium SS (see Footnote 5 therein); however, this does not imply that it is a symmetry for the dynamics. Here, we explicitly prove that $\Gs_\varphi$ is a group of weak symmetries for the boundary-driven XXZ chain we consider. We note that similar calculations have been carried out for a XXZ Hamiltonian 
in \cite{fragmentation} (see in particular Sec.\,II.B.2).

First, since the operators in the set $\{\sigma_z^{(j)}\}_{j=1}^{N}$ are mutually commuting, we may write $U_\varphi = \prod_{j=1}^{N} U_{\varphi, j}$ where, by using Pauli algebra, we have $U_{\varphi, j} = e^{-i\varphi/2 \sigma_z^{(j)} } = \cos(\frac{\varphi}{2}) \one_{2^N} -i \sin(\frac{\varphi}{2})\sigma_z^{(j)}$.

Before we proceed to the actual proof, we establish the following preliminary result. 
\begin{lemma}
    Consider $\sigma_+^{(k)}$ and $\sigma_-^{(k)}$, with $k=1,\dots,N$. Then, $U_\varphi \sigma_+^{(k)} = e^{-i\varphi}\sigma_+^{(k)}U_\varphi$ and $U_\varphi \sigma_-^{(k)} = e^{i\varphi}\sigma_-^{(k)}U_\varphi$.
\end{lemma}
\begin{proof}
The proof follows from direct calculation:
\begin{align*}
    U_\varphi \sigma_+^{(k)} &= \prod_{j=1}^{N} U_{\varphi,j} \sigma_+^{(k)} = U_{\varphi,k} \sigma_+^{(k)} \prod_{\substack{j=1\\j\neq k}}^{N} U_{\varphi,j} = (\cos(\frac{\varphi}{2})\one_{2^N} -i\sin(\frac{\varphi}{2})\sigma_z^{(k)}) \sigma_+^{(k)} \prod_{\substack{j=1\\j\neq k}}^{N} U_{\varphi,j}\\
    & =[\cos(\frac{\varphi}{2})\sigma_+^{(k)}-i\sin(\frac{\varphi}{2})( \underbrace{\sigma_z^{(k)}\sigma_x^{(k)}}_{i\sigma_y^{(k)}} + \underbrace{i \sigma_z^{(k)}\sigma_y^{(k)}}_{\sigma_x^{(k)}}  )]\prod_{\substack{j=1\\j\neq k}}^{N} U_{\varphi,j}\\ 
    & =[\cos(\frac{\varphi}{2})\sigma_+^{(k)}-i\sin(\frac{\varphi}{2})\sigma_+^{(k)}]\prod_{\substack{j=1\\j\neq k}}^{N} U_{\varphi,j}  =e^{-i\varphi/2}\sigma_+^{(k)}\prod_{\substack{j=1\\j\neq k}}^{N} U_{\varphi,j} \\
    &=e^{-i\varphi/2}\sigma_+^{(k)} U_{\varphi,k}^\dag \underbrace{U_{\varphi,k} \prod_{\substack{j=1\\j\neq k}}^{N} U_{\varphi,j}}_{U_\varphi}  =e^{-i\varphi/2}(\sigma_x^{(k)}+i\sigma_y^{(k)})(\cos(\frac{\varphi}{2})\one_{2^N} +i\sin(\frac{\varphi}{2})\sigma_z^{(k)}) U_\varphi \\
    &=e^{-i\varphi/2}[\cos(\frac{\varphi}{2})(\sigma_x^{(k)}+i\sigma_y^{(k)})-i\sin(\frac{\varphi}{2})({i\sigma_y^{(k)}}+{\sigma_x^{(k)}})] U_\varphi = e^{-i\varphi} \sigma_+^{(k)} U_\varphi .
\end{align*}
One similarly finds $U_\varphi \sigma_-^{(k)} = e^{i\varphi}\sigma_-^{(k)}U_\varphi$.
\end{proof}

We now prove that $U_\varphi = e^{-iM\varphi}$ is a weak symmetry for the boundary-driven XXZ chain. 

\begin{proof}
Consider first the XXZ Hamiltonian. We have $U_\varphi \sigma_z^{(j)} U_\varphi^\dag = \sigma_z^{(j)}$, for all $j$. Then, using the above lemma one obtains:
\[U_\varphi H U_\varphi^\dag = \sum_{j=1}^{N-1} U_\varphi \big[ \sigma_x^{(j)}\sigma_x^{(j+1)}+\sigma_y^{(j)}\sigma_y^{(j+1)}+\Delta\sigma_z^{(j)}\sigma_z^{(j+1)} \big]U_\varphi = H. \]
It remains to prove that $U_\varphi\Dc_{L}(\rho) U_\varphi^\dag = \Dc_{L}(U_\varphi \rho U_\varphi^\dag)$, for $L\in\{\sqrt{\alpha} \sigma_{+}^{(1)}, \sqrt{\beta} \sigma_{+}^{(N)}, \sqrt{\beta} \sigma_{-}^{(1)}, \sqrt{\alpha} \sigma_{-}^{(N)} \}$. 
We first notice that $\sigma_+^{(k)}\sigma_-^{(k)} = \frac{\one_{2^{N}}+\sigma_z^{(k)}}{2}$ and $\sigma_-^{(k)}\sigma_+^{(k)} = \frac{\one_{2^{N}}-\sigma_z^{(k)}}{2}$ and thus both commute with $U_\varphi$. This means that $U_\varphi \{L^\dag L, \rho\} U_\varphi = \{L^\dag L, U_\varphi \rho U_\varphi^\dag\}$ for all the noise operators $L$ here considered.
Then, using the lemma here above, we have, for example, $U_\varphi \sigma_+^{(1)} \rho \sigma_-^{(1)} U_\varphi^\dag = \cancel{e^{-2i\varphi}}\cancel{e^{2i\varphi}} \sigma_-^{(1)} U_\varphi \rho U_\varphi \sigma_+^{(1)}$. In connection with the previous observation, it follows that $U_\varphi \Dc(\rho) U_\varphi^\dag = \Dc(U_\varphi \rho U_\varphi^\dag)$, which concludes the proof.  \end{proof}

From the above proof, we can observe that $e^{-i\varphi M}$ is a weak symmetry for the model, regardless of the choice of the parameters $\alpha, \beta$, and $\Delta$. Even more interestingly, the proof holds also if we were to modify the dynamics by introducing some site-dependent parameters: for example, consider    the noise operators $$L\in\{\sqrt{\alpha_1} \sigma_{+}^{(1)}, \sqrt{\beta_N} \sigma_{+}^{(N)}, \sqrt{\beta_1} \sigma_{-}^{(1)}, \sqrt{\alpha_N} \sigma_{-}^{(N)} \},$$ with location-dependent parameters $\alpha_1, \alpha_N, \beta_1, \beta_N$, together with the Hamiltonian 
$$H=\sum_{j=1}^{N} A_j \big[\sigma_x^{(j)}\sigma_x^{(j+1)}+\sigma_y^{(j)}\sigma_y^{(j+1)} \big]+\Delta_j\sigma_z^{(j)}\sigma_z^{(j+1)}.$$ 
The same would hold even adding creation and annihilation or dephasing noise acting locally on any of the spins.

\smallskip

To conclude this section, we further explicitly verify that Proposition \ref{strongcond} holds, i.e., we verify that $\{O_i\}\subseteq \Gs_\varphi'$. 
Now, trivially, $U_\varphi \sigma_z^{(j)} U_\varphi^\dag = \sigma_z^{(j)}$ since all $U_{\varphi, j}$ commute with all $\sigma_z^{(k)}$. 
The following lemma proves that $J_j\in\Gs_\varphi$.
\begin{lemma}
     For all $j=1,\dots,N$, it holds $U_\varphi (\sigma_x^{(j)}\sigma_x^{(j+1)} + \sigma_y^{(j)}\sigma_y^{(j+1)} ) U_\varphi^\dag = \sigma_x^{(j)}\sigma_x^{(j+1)} + \sigma_y^{(j)}\sigma_y^{(j+1)}$ and $U_\varphi (\sigma_x^{(j)}\sigma_y^{(j+1)} - \sigma_y^{(j)}\sigma_x^{(j+1)} ) U_\varphi^\dag = \sigma_x^{(j)}\sigma_y^{(j+1)} - \sigma_y^{(j)}\sigma_x^{(j+1)}$.
\end{lemma}
\begin{proof}
    We only prove the first statement, as the second follows an identical proof.
First, observe that all $U_{\varphi,k}$ with $k\neq j$ commute with $\sigma_x^{(j)} $ and $\sigma_y^{(j)}$. Therefore, we have 
\begin{align*}
    U_\varphi (\sigma_x^{(j)}\sigma_x^{(j+1)}+\sigma_y^{(j)}\sigma_y^{(j+1)}) U_\varphi^\dag &= U_{\varphi, j+1} U_{\varphi, j} (\sigma_x^{(j)}\sigma_x^{(j+1)}+\sigma_y^{(j)}\sigma_y^{(j+1)}) U_{\varphi, j}^\dag U_{\varphi, j+1}^\dag \\
    &= (U_{\varphi, j} \sigma_x^{(j)} U_{\varphi, j}^\dag) (U_{\varphi, j+1} \sigma_x^{(j+1)} U_{\varphi, j+1}^\dag) + (U_{\varphi, j} \sigma_y^{(j)} U_{\varphi, j}^\dag) (U_{\varphi, j+1} \sigma_y^{(j+1)} U_{\varphi, j+1}^\dag).
\end{align*}
Then, starting from \(U_{\varphi, j} \sigma_x^{(j)} U_{\varphi, j}^\dag\) we also have:
\begin{align*}
    U_{\varphi, j} \sigma_x^{(j)} U_{\varphi, j}^\dag
 &= (\cos(\frac{\varphi}{2}) \one_{2^N} -i \sin(\frac{\varphi}{2})\sigma_z^{(j)}) \sigma_x^{(j)} (\cos(\frac{\varphi}{2}) \one_{2^N} +i \sin(\frac{\varphi}{2})\sigma_z^{(j)})\\
 &= (\cos(\frac{\varphi}{2}) \one_{2^N} -i \sin(\frac{\varphi}{2})\sigma_z^{(j)})(\cos(\frac{\varphi}{2}) \sigma_x^{(j)} +i \sin(\frac{\varphi}{2})\sigma_x^{(j)}\sigma_z^{(j)})\\
 &= \cos^2(\frac{\varphi}{2}) \sigma_x^{(j)} 
 +i \cos(\frac{\varphi}{2})\sin(\frac{\varphi}{2})\sigma_x^{(j)}\sigma_z^{(j)}
 -i \cos(\frac{\varphi}{2})\sin(\frac{\varphi}{2})\sigma_z^{(j)}\sigma_x^{(j)} 
 + \sin^2(\frac{\varphi}{2}) \sigma_z^{(j)}\sigma_x^{(j)}\sigma_z^{(j)})\\
 &= [\cos^2(\frac{\varphi}{2})-\sin^2(\frac{\varphi}{2})]\sigma_x^{(j)}  +  2\cos(\frac{\varphi}{2})\sin(\frac{\varphi}{2})\sigma_y^{(j)} .
 \end{align*}
Similarly, one obtains $U_{\varphi, j} \sigma_y^{(j)} U_{\varphi, j}^\dag = [\cos^2(\frac{\varphi}{2})-\sin^2(\frac{\varphi}{2})]\sigma_y^{(j)}  -  2\cos(\frac{\varphi}{2})\sin(\frac{\varphi}{2})\sigma_x^{(j)}$. Combining these results for $\sigma_x^{(j)}\sigma_x^{(j+1)}$ and $\sigma_y^{(j)}\sigma_y^{(j+1)}$, one obtains 
\begin{align*}
    &U_\varphi \sigma_x^{(j)}\sigma_x^{(j+1)} U_\varphi^\dag = \\ &= \left\{[\cos^2(\frac{\varphi}{2})-\sin^2(\frac{\varphi}{2})]\sigma_x^{(j)}  +  2\cos(\frac{\varphi}{2})\sin(\frac{\varphi}{2})\sigma_y^{(j)}\right\} \left\{[\cos^2(\frac{\varphi}{2})-\sin^2(\frac{\varphi}{2})]\sigma_x^{(j+1)}  +  2\cos(\frac{\varphi}{2})\sin(\frac{\varphi}{2})\sigma_y^{(j+1)}\right\}, \\
    &U_\varphi \sigma_y^{(j)}\sigma_y^{(j+1)} U_\varphi^\dag = \\ &= \left\{[\cos^2(\frac{\varphi}{2})-\sin^2(\frac{\varphi}{2})]\sigma_y^{(j)}  -  2\cos(\frac{\varphi}{2})\sin(\frac{\varphi}{2})\sigma_x^{(j)}\right\} \left\{[\cos^2(\frac{\varphi}{2})-\sin^2(\frac{\varphi}{2})]\sigma_y^{(j+1)}  -  2\cos(\frac{\varphi}{2})\sin(\frac{\varphi}{2})\sigma_x^{(j+1)}\right\},
\end{align*}
which, using the identity $[\cos^2(\frac{\varphi}{2})-\sin^2(\frac{\varphi}{2})]^2 = 1- 4\cos^2(\frac{\varphi}{2})\sin^2(\frac{\varphi}{2})$, leads to their sum  \( U_\varphi \sigma_x^{(j)}\sigma_x^{(j+1)} U_\varphi^\dag + U_\varphi \sigma_y^{(j)}\sigma_y^{(j+1)} U_\varphi^\dag = \sigma_x^{(j)}\sigma_x^{(j+1)} + \sigma_y^{(j)}\sigma_y^{(j+1)} \).
\end{proof}

\medskip

\noindent\hrulefill
\bigskip

\twocolumn

\bibliography{ref}

\end{document}